\documentclass{amsart}
\usepackage[latin1]{inputenc}
\usepackage{latexsym}
\usepackage{amsmath}
\usepackage{amssymb}
\usepackage{amsfonts}
\usepackage{mathrsfs}               % \mathscr
\usepackage{bbm}                    % \mathbbm 
\usepackage{bbold}                  % \mathbb 
\usepackage{textcomp}               
\usepackage{amstext}
\usepackage{enumerate}
\usepackage{units}
\usepackage{stmaryrd}
\usepackage{multicol}
\usepackage{enumitem}
 \newtheorem{theorem}{Theorem}[section]
 \newtheorem{corollary}[theorem]{Corollary}
 \newtheorem{lemma}[theorem]{Lemma}
 \newtheorem{proposition}[theorem]{Proposition}
 \theoremstyle{definition}
 
 \theoremstyle{remark}
 \newtheorem{remark}[theorem]{Remark}
 \newtheorem{example}[theorem]{Example}
 \numberwithin{equation}{section}
\newcommand{\CC}{\mathbb{C}}

\newcommand{\NN}{\mathbb{N}}

\newcommand{\RR}{\mathbb{R}}

\newcommand{\supp}{\mathrm{supp}}
\newcommand{\dist}{\mathrm{dist}}
\newcommand{\Ran}{\mathrm{Ran}}

\newcommand{\sgn}{\mathrm{sgn}}
\newcommand{\loc}{\mathrm{loc}}

\newcommand{\Div}{{\mathrm{div}}}
\newcommand{\rot}{{\mathrm{rot}}}

\newcommand{\Id}{\mathrm{d}}

\newcommand{\SPn}[2]{\langle #1|#2\rangle} 
\newcommand{\SPb}[2]{\big\langle #1\big|#2\big\rangle} 
\newcommand{\SPB}[2]{\Big\langle \,#1\,\Big|\,#2\, \Big\rangle}

\newcommand{\ol}[1]{\overline{#1}} 
\newcommand{\ul}[1]{\underline{#1}} 
\newcommand{\mr}[1]{\mathring{#1}}

\newcommand{\nf}[2]{\nicefrac{#1}{#2}}
\newcommand{\eh}{{\nf{1}{2}}}
\newcommand{\mh}{{-\nf{1}{2}}}
\newcommand{\cA}{\mathcal{A}}
\newcommand{\cO}{\mathcal{O}} 
\newcommand{\cC}{\mathcal{C}}
\newcommand{\cD}{\mathcal{D}} 
\newcommand{\cE}{\mathcal{E}}\newcommand{\cQ}{\mathcal{Q}}

\newcommand{\cH}{\mathcal{H}}

\newcommand{\cM}{\mathcal{M}}       
\newcommand{\cZ}{\mathcal{Z}} 
\newcommand{\sC}{\mathscr{C}}
\newcommand{\sD}{\mathscr{D}} 
\newcommand{\sE}{\mathscr{E}}\newcommand{\sQ}{\mathscr{Q}}
\newcommand{\sF}{\mathscr{F}}

\newcommand{\sU}{\mathscr{U}}

\newcommand{\sK}{\mathscr{K}}

\newcommand{\fA}{\mathfrak{A}}

\newcommand{\fP}{\mathfrak{P}}

\newcommand{\fS}{\mathfrak{S}}

\newcommand{\V}[1]{\boldsymbol{#1}}
\newcommand{\vsigma}{\boldsymbol{\sigma}}

\newcommand{\ve}{\varepsilon}
\newcommand{\vp}{\varphi}
\newcommand{\vo}{\varpi}

\newcommand{\vr}{\varrho}
\newcommand{\vt}{\vartheta}
\newcommand{\vs}{\varsigma}
\newcommand{\id}{\mathbbm{1}}                   % Identity 
\newcommand{\dom}{\cD}                          % domain of definition
\newcommand{\fdom}{\cQ}                         % form domain 
\newcommand{\HP}{\mathfrak{h}}                  % photon Hilbert space 
\newcommand{\LO}{\mathcal{B}}                   % set of bounded, linear operators
\newcommand{\expv}[1]{\epsilon(#1)}
\newcommand{\ee}{\mathfrak{e}} 
\renewcommand{\Re}{\mathrm{Re}}
\renewcommand{\Im}{\mathrm{Im}}
\renewcommand{\le}{\leqslant}        
\renewcommand{\ge}{\geqslant}  
\begin{document}

\title[Pauli-Fierz operators]{Pauli-Fierz type operators with singular electromagnetic 
potentials on general domains}

\author[O.~Matte]{Oliver~Matte}

\address{Oliver Matte, Institut for Matematik, {Aa}rhus Universitet,
Ny Munkegade 118, DK-8000 Aarhus C, Denmark}
\email{matte@math.au.dk}

\maketitle

\begin{abstract}
We consider Dirichlet realizations of Pauli-Fierz type operators 
generating the dynamics of non-relativistic matter particles which are confined to an arbitrary  
open subset of the Euclidean position space and coupled to quantized radiation fields. We find
sufficient conditions under which their domains and a natural class of operator cores are determined 
by the domains and operator cores of the corresponding Dirichlet-Schr\"{o}dinger operators 
and the radiation field energy. Our results also extend previous ones dealing with the entire Euclidean 
space, since the involved electrostatic potentials might be unbounded at infinity with local singularities 
that can only be controlled in a quadratic form sense, and since locally square-integrable classical 
vector potentials are covered as well.  We further discuss Neumann realizations of 
Pauli-Fierz type operators on Lipschitz domains.

\bigskip

 
\bigskip

\end{abstract}
%
%
%
%%%%%%%%%%%%%%%%%%%%%%%%%%%%%%%%%%%%%%%%%%%%%%%%%
%%%%%%%%%%%%%%%%%%%%%%%%%%%%%%%%%%%%%%%%%%%%%%%%%
%%%%%%%%%%%%%%%%%%%%%%%%%%%%%%%%%%%%%%%%%%%%%%%%%

\section{Introduction}\label{sec-intro}

\noindent
A lot of attention has been devoted to the mathematical analysis of physical models for 
a conserved number of non-relativistic quantum-mechanical matter particles in interaction with a
quantized radiation field comprised of an undetermined number of relativistic bosons. 
The prime example for such a model is the standard model of non-relativistic
quantum electrodynamics, where electrons interact with the quantized electromagnetic field
(photon field). In this example quantized field operators are introduced via minimal coupling and
the resulting Hamiltonian is often called the Pauli-Fierz operator. The aim of the present article is to 
extend existing non-perturbative results on the self-adjointness properties of Hamilonians of this type 
\cite{Falconi2015,HaslerHerbst2008,Hiroshima2000esa,Hiroshima2002} to a
larger class of exterior electromagnetic potentials appearing as coefficients in the Hamiltonian.
Furthermore, we shall allow for an arbitrary open subset of the Euclidean space as
position space for the matter particles, while the aforementioned papers deal with the entire
Euclidean space only. We consider general open position spaces
because many interesting effects appear in non-relativistic quantum electrodynamics in bounded 
cavities or on unbounded domains confined by perfectly conducting grounded walls. Prominent
examples are the Casimir or van der Vaals forces; see, e.g., the textbooks
\cite{Dutra2005,Milonni1994}; a detailed discussion of the corresponding formal minimal coupling 
Hamiltonians can also be found in \cite{PowerThirunamachandran1982}. 
Further examples of position spaces which are proper subsets of the Euclidean space are encountered
when the nuclei in a molecular system are treated as static particles of finite extent, as in the
theory of hard-core multi-body Schr\"{o}dinger operators; see, e.g., \cite{ItoSkipsted2014} and
the references given there. We should mention
that, at least so far, the mathematical analysis of minimal coupling Hamiltonians requires the 
introduction of an artificial ultra-violet regularization damping the matter-radiation interaction at 
very high frequencies.

In what follows we shall describe our results in more detail.
Throughout the whole article we assume that $\nu,\tilde{\nu},s\in\NN$ and
$\Lambda\subset\RR^\nu$ is open. We put $\Lambda_*:=\Lambda\times\{1,\ldots,s\}$,
where the second factor in the Cartesian product accounts for spin degrees of freedom (if any)
of the matter particles. Furthermore, $(\cM,\fA,\mu)$ is a $\sigma$-finite measure space such that 
the Hilbert space for a single boson,
\begin{align}\label{defHP}
\HP&:=L^2(\cM,\fA,\mu),
\end{align}
is separable. The measurable function $\omega:\cM\to\RR$ plays the role of the
dispersion relation of a single boson. We always assume that $\omega$ is $\mu$-a.e. strictly 
positive. The symbol $\sF$ denotes the bosonic Fock space over $\HP$. 
Our main goal is to characterize the domain and operator cores of a Dirichlet realization of the 
Pauli-Fierz operator acting in $L^2(\Lambda_*,\sF)$. This operator is formally given by
\begin{align}\label{PF-formal}
\frac{1}{2}(-i\nabla-\V{A}-\vp(\V{G}))^2-\vsigma\cdot\vp(\V{F})
-\vsigma\cdot\V{B}+\Id\Gamma(\omega)+V,
\end{align}
where, for every $\V{x}\in\Lambda$,  the formal vectors 
$\vp(\V{G}_{\V{x}}):=(\vp(G_{1,\V{x}}),\ldots,\vp(G_{\nu,\V{x}}))$ and 
$\vp(\V{F}_{\V{x}}):=(\vp(F_{1,\V{x}}),\ldots,\vp(F_{\tilde{\nu},\V{x}}))$ are tuples of field operators
and $\Id\Gamma(\omega)$ is the radiation field energy. The notations $\sF$, $\vp(f)$, and 
$\Id\Gamma(\omega)$ will be explained in Subsect.~\ref{ssec-Fock}. Furthermore,
$\vsigma:=(\sigma_1,\ldots,\sigma_{\tilde{\nu}})$ is a tuple of Hermitian $s$\texttimes$s$-matrices.
These matrices are the only terms in \eqref{PF-formal} that act on the spin variables in 
$\{1,\ldots,s\}$; see Rem.~\ref{remZeeman} for precise definitions.

The main originality of this article lies in the rather general conditions imposed on the data
in \eqref{PF-formal}. For instance, the only requirement on the positive part of the 
electrostatic potential $V:\Lambda\to\RR$ is local integrability, 
while its negative part is assumed to be {\em form}-bounded with respect to 
$-1/2$ times the Dirichlet-Laplacian with relative form bound $<1$. 
The classical vector potential $\V{A}=(A_1,\ldots,A_\nu):\Lambda\to\RR^\nu$ 
only needs to be locally square-integrable, which is the natural requirement for the construction of 
magnetic Schr\"{o}dinger operators via quadratic forms.
Of course, the classical Zeeman term $\vsigma\cdot\V{B}$ (if any) should contain the curl of $\V{A}$.
In our discussion we may, however, ignore this relation and simply keep the assumptions on $\V{A}$, 
$\V{B}$, and $V$ as general as our arguments permit. In an application where spin degrees of 
freedom are taken into account together with an exterior magnetic field and $V$ as above, our 
results cover the case where the components of $\V{B}:\Lambda\to\RR^{\tilde{\nu}}$ are sums of 
bounded terms and contributions that are infinitesimally form-bounded with respect to the 
negative Dirichlet-Laplacian.

The quantities $\omega$, $\V{G}$, and $\V{F}$ satisfy the weakest assumptions appearing in this 
context either. Namely, to determine the domain of the Dirichlet-Pauli-Fierz operator
we shall eventually assume that
\begin{align}\nonumber
&\V{G}\in L^\infty(\Lambda,\fdom(\omega^{-1}+\omega)^\nu),\quad
\Div\V{G}\in L^\infty(\Lambda,\fdom(\omega^{-1})),
\\\label{hypGFintro}
&\V{F}\in L^\infty(\Lambda,\fdom(\omega^{-1})^{\tilde{\nu}}).
\end{align}
Here $\fdom$ stands for the form domain and the Hilbert space-valued
divergence is understood in a weak sense; see Subsect.~\ref{ssec-Sobolev}.
These conditions are slightly milder than the ones in \cite{HaslerHerbst2008} where the case 
$\Lambda=\RR^\nu$ is treated. 
In applications to cavity quantum electrodynamics the data $(\omega,\V{G},\V{F})$ 
should correspond to solutions of the Maxwell equations with perfect electric conductor boundary 
conditions after suitable assumptions on the regularity of $\partial\Lambda$ have been added. 
This is, however, a physical requirement and the behavior of $\V{G}$, $\V{F}$, and $\Lambda$
at the boundary is in fact immaterial for our results on the Dirichlet-Pauli-Fierz operator to hold.
Dirichlet realizations of the Pauli-Fierz operator on a non-trivial domain can also 
appear for technical reasons when localization arguments are applied to non-confined 
systems as, for instance, in \cite{LiebLoss2003}. In such a case $\V{G}$ and $\V{F}$ do not
necessarily satisfy physical boundary conditions.

The main result of this article (Thm.~\ref{thm-dom} in the case $s=1$,
$\V{F}=\V{0}$, $\V{B}=\V{0}$, with a simple extension to $s>1$ and non-vanishing 
$\V{F}$ and $\V{B}$ in Rem.~\ref{remZeeman}) asserts that the domain of the 
Dirichlet realization of \eqref{PF-formal} is equal to the intersection of the domain of the 
Dirichlet-Schr\"odinger operator corresponding to $(V,\V{A},\V{B})$ with the domain of 
$\Id\Gamma(\omega)$, when the latter two operators are considered as operators in 
$L^2(\Lambda_*,\sF)$ in the canonical way. That is, the domain of the Dirichlet realization of 
\eqref{PF-formal} neither depends on $\V{G}$ nor on $\V{F}$ thanks to the 
$L^\infty$-conditions in \eqref{hypGFintro}. Moreover, Thm.~\ref{thm-dom} and 
Rem.~\ref{remZeeman} identify natural operator cores of the Dirichlet-Pauli-Fierz operator in terms
of the cores of the Dirichlet-Schr\"odinger operator and $\Id\Gamma(\omega)$.

In the case where $\Lambda=\RR^\nu$, $V$ is relatively {\em operator}-bounded with respect to 
$-\tfrac{1}{2}\Delta$ with relative bound $<1$, $\V{A}$ is bounded with bounded and continuous
first derivative, $\V{B}=\rot\V{A}$, and $\V{G}$ and $\V{F}$ satisfy certain slightly stronger hypotheses, 
all results of Thm.~\ref{thm-dom} and Rem.~\ref{remZeeman} are well-known. 
Their first non-perturbative proofs have been given in \cite{Hiroshima2000esa,Hiroshima2002} 
in this case. Starting with the case where only $\V{G}$ is non-zero, the arguments in 
\cite{Hiroshima2000esa} are based on the invariant domain method for the study of essential 
self-adjointness and Feynman-Kac formulas. Then a diamagnetic inequality for the semi-group 
associated with the Pauli-Fierz operator is employed to argue that a 
$-\tfrac{1}{2}\Delta$-small potential $V$  (in the operator sense) is also small with respect to the 
free Pauli-Fierz operator \cite{Hiroshima2000esa,Hiroshima2002}. 
For {\em infinitesimally} Laplace-bounded $V$, one can avoid the use of diamagnetic inequalities in 
the determination of the domain of the Pauli-Fierz operator.
For such $V$ and vanishing $\V{A}$, simpler, analytic proofs have been given 
in \cite{HaslerHerbst2008}. Earlier proofs in a perturbative situation based on the
Kato-Rellich theorem can be found in \cite{Arai1981,BFS1998b}. The article \cite{Arai1981}
also contains a non-perturbative result on the 
dipole approximation to non-relativistic quantum electrodynamics.

To find a natural domain on which the Dirichlet-Pauli-Fierz operator is {\em essentially} self-adjoint it is
actually sufficient to assume that
\begin{align}\label{hypGFintroM}
\V{G}\in L^\infty(\Lambda,\HP^\nu),\quad\Div\V{G}\in L^\infty(\Lambda,\HP),\quad
\V{F}\in L^\infty(\Lambda,\HP^{\tilde{\nu}}).
\end{align}
This has been observed in \cite[\textsection4.3]{Falconi2015}. 
As one example for the application of a general theorem
the latter article explicitly covers the case where $\Lambda=\RR^\nu$, 
$V$ is non-negative and locally square-integrable, $\V{A}$ and $\V{B}$ are zero, and
$\V{G}$ and $\V{F}$ are given by the usual plane wave solutions to the Maxwell equations.
(If $\V{G}$ is an affine function of $\V{x}\in\RR^\nu$, then essential self-adjointness in
an otherwise similar situation also follows from \cite[Ex.~3]{Arai1991}.)
It is, however, clear that the abstract theorem in \cite{Falconi2015} also applies to more general 
situations. Nevertheless, we shall give an alternative proof for the essential self-adjointness under 
the condition \eqref{hypGFintroM} in Thm.~\ref{thmMM}, utilizing a variant of an argument due
to M.~K\"{o}nenberg \cite{KMS2013}. This is because some of the bounds and ideas employed in 
the proof of Thm.~\ref{thmMM} are also needed to characterize the domain of the self-adjoint
realization under the condition \eqref{hypGFintro}, which is the key aspect of our results.

Recall that the Neumann realization of the magnetic Schr\"{o}dinger operator can be defined
as the self-adjoint operator representing a canonical maximal Schr\"{o}dinger form, while the
Dirichlet realization represents a minimal form, which is a restriction of the maximal one. 
We shall mimic these constructions in the presence of quantized fields and, for mathematical
curiosity, we will derive an analogue of our main Thm.~\ref{thm-dom} for the Neumann-Pauli-Fierz 
operator on Lipschitz domains in Sect.~\ref{appNeumann}. In the Neumann case an additional 
boundary condition on $\V{G}$ is required for such an analogue to hold, corresponding to solutions of 
the Maxwell equations with ``perfect magnetic conductor'' boundary conditions. We do, however, not 
know whether the Neumann-Pauli-Fierz operator is of any physical significance, which is also the 
reason why we refrained from investigating more general boundary conditions. As it is the case for
Schr\"{o}dinger forms \cite{SimonJOT1979}, we shall see that the minimal and maximal 
Pauli-Fierz forms agree when $\Lambda=\RR^\nu$.

The organization of this article is given as follows. In Sect.~\ref{sec-prel} we collect some remarks on 
Hilbert space-valued weak derivatives, recall some facts on Fock space calculus, 
and derive some Leibniz rules for vector-valued Sobolev functions that are multiplied by field 
operators. Although many parts of Sect.~\ref{sec-prel} are straightforward or well-known, 
we think that a presentation of these topics taylor-made for our later sections might be convenient for 
the reader. In Sect.~\ref{sec-dia} we add a new, pointwise diamagnetic inequality for a sum of a 
classical and a quantized vector potential to the list of diamagnetic inequalities in non-relativistic 
quantum electrodynamics shown earlier; see the first paragraph of that section for references. 
This pointwise diamagnetic inequality will be used in the crucial step of our proof of 
Thm.~\ref{thm-dom}. Self-adjoint realizations of the Schr\"odinger and 
Pauli-Fierz operators will be defined via quadratic forms in Sect.~\ref{sec-ham}. 
Our main result, Thm.~\ref{thm-dom}, is stated and proved in Sect.~\ref{sec-dom} by further
elaborating on a general strategy that we applied to fiber Hamiltonians in \cite[App.~2]{GMM2016}.
Some examples are provided in Sect.~\ref{appex} before we treat the Neumann case 
in Sect.~\ref{appNeumann}.

\subsubsection*{Some general notation} The symbol $\dom(T)$ denotes the domain of a linear
operator $T$, and $\fdom(T)$ is the form domain of a semi-bounded self-adjoint operator $T$
in a Hilbert space $\sK$. If $T\ge0$, then we consider $\fdom(T)$ as a Hilbert space 
with scalar product 
$\SPn{\phi}{\psi}_{\fdom(T)}=\SPn{T^\eh\phi}{T^\eh\psi}_\sK+\SPn{\phi}{\psi}_\sK$,
$\phi,\psi\in\fdom(T)$. 

If $\mathfrak{t}$ is a quadratic form in $\sK$ that is semi-bounded from below and if $c$ denotes
the corresponding greatest lower bound, then the form norm corresponding to $\mathfrak{t}$ is
given by $\|\psi\|_{\mathfrak{t}}^2=\mathfrak{t}[\psi]+(1-c)\|\psi\|_\sK^2$, $\psi\in\dom(\mathfrak{t})$.
The sesqui-linear form associated with $\mathfrak{t}$ via the polarization identity is denoted
by $\mathfrak{t}[\phi,\psi]$, $\phi,\psi\in\dom(\mathfrak{t})$.

If we write $\Omega\Subset\Lambda$, then $\Omega$ is a subset of $\RR^\nu$ whose closure is
compact and contained in $\Lambda$.

If $\sC\subset L^2(\Lambda)$ or $\sC\subset L^2(\Lambda_*)$ 
is a subspace and $\sE$ a vector space, then we set
\begin{equation}\label{def-CtensorD}
\sC\otimes\sE:=\mathrm{span}_{\CC}\big\{f\psi:\,f\in\sC,\,\psi\in\sE\big\}.
\end{equation}

We shall write $a\wedge b:=\min\{a,b\}$ and $a\vee b:=\max\{a,b\}$, for $a,b\in\RR$.

%%%%%%%%%%%%%%%%%%%%%%%%%%%%%%%%%%%%%%%%%%%%%%%%%
%%%%%%%%%%%%%%%%%%%%%%%%%%%%%%%%%%%%%%%%%%%%%%%%%
%%%%%%%%%%%%%%%%%%%%%%%%%%%%%%%%%%%%%%%%%%%%%%%%%

\section{Preliminaries}\label{sec-prel}

\subsection{Vector-valued weak partial derivatives and divergences}\label{ssec-Sobolev}

\noindent
A well-known complication in the study of magnetic Schr\"{o}dinger operators 
with merely locally square-integrable vector potentials is the fact that
the weak partial derivatives of functions in magnetic Sobolev spaces are
in general not square-integrable. We shall encounter the same difficulty in dealing with the
Fock space-valued functions in the (form) domains of our Pauli-Fierz operators. As a preparation,
we thus collect some basic remarks on weak partial derivatives of Hilbert space-valued
functions in this subsection.

Throughout the whole subsection, $\sK$ is a separable Hilbert space. Let $\sE\subset\sK$ be
a subspace. Then we denote the space of $\sE$-valued test functions on $\Lambda$ by
\begin{align}\label{defvecvalTF}
\sD(\Lambda,\sE)&:=C_0^\infty(\Lambda)\otimes\sE,\quad\sD(\Lambda):=C_0^\infty(\Lambda).
\end{align}

For every $j\in\{1,\ldots,\nu\}$, we say that $\Upsilon_j\in L^1_\loc(\Lambda,\sK)$ is a
weak partial derivative of $\Psi\in L^1_\loc(\Lambda,\sK)$ with respect to $x_j$, iff
\begin{align}\label{defwpd}
\int_{\Lambda}\SPn{\partial_{x_j}\eta(\V{x})}{\Psi(\V{x})}_\sK\Id\V{x}
=-\int_\Lambda\SPn{\eta(\V{x})}{\Upsilon_j(\V{x})}_\sK\Id\V{x}
\end{align}
holds, for all $\eta\in\sD(\Lambda,\sK)$.

\begin{remark}\label{remwpd}
Let $j\in\{1,\ldots,\nu\}$ and $\Psi,\Upsilon_j\in L^1_\loc(\Lambda,\sK)$. Then the following holds:
\begin{enumerate}[leftmargin=0.67cm]
\item[{\rm(1)}] $\Upsilon_j$ is a weak partial derivative with respect to $x_j$ of $\Psi$, if and only if
$\SPn{\phi}{{\Upsilon}_j}_{\sK}$ is a weak partial derivatives with respect to $x_j$ of 
$\SPn{\phi}{\Psi}_{\sK}$, for every $\phi\in\sK$.
In particular, $\Upsilon_j$ is unique in the affirmative case.
\item[{\rm(2)}] 
Let $\sE\subset\sK$ be a total subset.
If $\Upsilon_j$ satisfies \eqref{defwpd} for all $\eta$ of the form
$\eta=f\phi$ with $f\in\sD(\Lambda)$ and $\phi\in\sE$, then, by linearity and dominated convergence, 
it is a weak partial derivative of $\Psi$ with respect to $x_j$. 
\item[{\rm(3)}] Since taking the scalar product with a fixed vector in $\sK$ and the 
$\sK$-valued Bochner-Lebesgue integral commute, $\Upsilon_j$ is a
weak partial derivative of $\Psi$ with respect to $x_j$, if and only if
\begin{align*}
\int_{\Lambda}(\partial_{x_j}\eta)(\V{x})\Psi(\V{x})\Id\V{x}
=-\int_\Lambda\eta(\V{x})\Upsilon_j(\V{x})\Id\V{x},\quad\eta\in\sD(\Lambda).
\end{align*}
\end{enumerate}
\end{remark}

If it exists, then we denote the unique weak partial derivative of $\Psi\in L^1_\loc(\Lambda,\sK)$ with 
respect to $x_j$ by $\partial_{x_j}\Psi$.

\begin{remark}\label{remgradlevel}
If $\Psi\in L^1_\loc(\Lambda,\sK)$ has a weak partial derivative respect to $x_j$,
then $\partial_{x_j}\Psi=0$ almost everywhere on $\{\Psi=0\}$.

This follows from the same assertion in the case $\sK=\CC$ 
(cf. the proof of \cite[Thm.~6.19]{LiebLoss2001})
upon applying it to $\SPn{\phi}{\Psi}_{\sK}$, for every $\phi$ in a countable total subset of $\sK$, and 
taking Rem.~\ref{remwpd}(1) into account.
\end{remark}

Of course, the definition of the weak partial derivatives depends on the topology on $\sK$. 
Hence, we shall sometimes say that they are {\em computed in $\sK$}. Since the coupling functions
appearing in the Pauli-Fierz operators attain values in the domain of certain unbounded operators,
it thus makes sense to note the following:

\begin{remark}\label{remsKfdomT}
Let $j\in\{1,\ldots,\nu\}$, $p\in[1,\infty]$, $T$ be a non-negative self-adjoint operator in $\sK$, and 
$\Psi\in L_\loc^p(\Lambda,\fdom(T))$. Then $\Psi$ has a weak partial derivative with respect to
$x_j$ computed in $\sK$ and satisfying $\partial_{x_j}\Psi\in L_\loc^p(\Lambda,\fdom(T))$, if and
only if it has a weak partial derivative with respect to $x_j$ computed in $\fdom(T)$ and belonging
to $L_\loc^p(\Lambda,\fdom(T))$.
The same assertion holds true, if the subscripts ``loc'' are dropped everywhere.

We drop the straightforward proof which uses that $\Ran(T+1)=\sK$, that $\dom(T)$ is dense in 
$\fdom(T)$ with respect to the form norm, and Rem.~\ref{remwpd}(2) with $\sE=\dom(T)$.
\end{remark}

\begin{remark}\label{remSob}
Let $p\in[1,\infty]$, $j\in\{1,\ldots,\nu\}$, and assume that
$\Psi\in L^p_\loc(\Lambda,\sK)$ has a weak partial derivative with respect to
$x_j$ such that $\partial_{x_j}\Psi\in L^p_\loc(\Lambda,\sK)$. Furthermore,
let $\rho\in C_0^\infty(\RR^\nu,\RR)$ satisfy $\rho\ge0$, $\rho(\V{x})=0$, for $|\V{x}|\ge1$, and
$\|\rho\|_1=1$. Set $\Lambda_n:=\{\V{y}\in\Lambda:\dist(\V{y},\partial\Lambda)>1/n\}$
and $\rho_n(\V{x}):=n^{\nu}\rho(n\V{x})$, $\V{x}\in\RR^\nu$, for all $n\in\NN$. Define
\begin{align}\label{klaus1}
\Psi_n(\V{x}):=\int_{\Lambda}\rho_n(\V{x}-\V{y})\Psi(\V{y})\Id\V{y},\quad 
\V{x}\in\Lambda_n,\,n\in\NN.
\end{align}
Then $\Psi_n\in C^\infty(\Lambda_n,\sK)$, if $\Lambda_n\not=\emptyset$, 
and, for every measurable $\Omega\Subset\Lambda$,
\begin{align*}
\|\Psi_n-\Psi\|_{L^p(\Omega,\sK)}+
\|\partial_{x_j}\Psi_n-\partial_{x_j}\Psi\|_{L^p(\Omega,\sK)}\xrightarrow{\;\;n\to\infty\;\;}0,
\quad\text{if $p<\infty$.}
\end{align*}
If $p=\infty$, then $\Psi_n\to\Psi$ and $\partial_{x_j}\Psi_n\to\partial_{x_j}\Psi$ a.e. on $\Lambda$. 

See, e.g., \cite[\textsection4.2.1]{EvansGariepy} for a proof in the scalar case.
On account of Rem.~\ref{remwpd}(3) this proof carries over to the vector-valued case.
To cover the case $p=\infty$, we also use the fact that the Lebesgue point theorem holds for the
Bochner-Lebesgue integral as well \cite{HillePhillips1957}.
\end{remark}

The next lemma will be used to prove our diamagnetic inequality.
Given a representative $\Psi(\cdot)$ of $\Psi\in L_\loc^1(\Omega,\sK)$ and $\delta>0$, we define
\begin{align}\nonumber
Z_\delta(\Psi)&:=(\delta^2+\|\Psi\|_{\sK}^2)^\eh,\quad{\fS_{\delta,\Psi}}:=Z_\delta(\Psi)^{-1}\Psi,
\\\label{defsgn}
{\fS_{\Psi}}(\V{x})&:=\left\{\begin{array}{ll}
\|\Psi(\V{x})\|_\sK^{-1}\Psi(\V{x}),&\V{x}\in\{\Psi(\cdot)\not=0\},
\\
0,& \V{x}\in\{\Psi(\cdot)=0\}.
\end{array}\right.
\end{align}

\begin{lemma}\label{lem-abs-val}
Let $j\in\{1,\ldots,\nu\}$, $p\in[1,\infty]$, $\delta>0$, and let $\Psi\in L_\loc^{p}(\Lambda,\sK)$ have a 
weak partial derivative with respect to $x_j$ satisfying $\partial_{x_j}\Psi\in L_\loc^{p}(\Lambda,\sK)$. 
Then $\|\Psi\|_\sK,Z_\delta(\Psi)\in L_{\loc}^{p}(\Lambda)$ have weak partial derivatives with 
respect to $x_j$ as well. The latter are in $L_{\loc}^{p}(\Lambda)$ and given by
\begin{align}\label{dia00}
\partial_{x_j}Z_\delta(\Psi)&=\Re\SPn{{\fS_{\delta,\Psi}}}{\partial_{x_j}\Psi}_{\sK},\quad
\partial_{x_j}\|\Psi\|_\sK=\Re\SPn{{\fS_\Psi}}{\partial_{x_j}\Psi}_{\sK}.
\end{align}
\end{lemma}

\begin{proof} 
Let $f\in\sD(\Lambda)$. 
Then we find some open $\Omega\Subset\Lambda$ such that $\supp(f)\subset\Omega$.
For the $\Psi_n$ defined in \eqref{klaus1} and sufficiently large $n_0\in\NN$, we then get
\begin{align}\label{uffe1}
\int_{\Omega}(\partial_{x_j}f){Z_{\delta}(\Psi_n)}\Id\V{x}
&=-\int_{\Omega}f{\Re\SPn{\fS_{\delta,\Psi_n}}{\partial_{x_j}\Psi_n}_{\sK}}\Id\V{x},\quad n\ge n_0.
\end{align}
On account of
$|Z_{\delta}(\Psi_n)-Z_\delta(\Psi)|\le|\|\Psi_n\|_{\sK}-\|\Psi\|_{\sK}|\le\|\Psi_n-\Psi\|_{\sK}$
and $\Psi_n\to\Psi$ in $L^1(\Omega,\sK)$ (because 
$L^{p}_\loc(\Lambda,\sK)\subset L^{1}_\loc(\Lambda,\sK)$), we see that, as $n\to\infty$, 
the left hand side of \eqref{uffe1} converges to the left hand side of 
\begin{align}\label{tim1}
\int_\Omega(\partial_{x_j}f)Z_\delta(\Psi)\Id\V{x}=-
\int_\Omega f{\Re\SPn{\fS_{\delta,\Psi}}{\partial_{x_j}\Psi}_{\sK}}\Id\V{x}.
\end{align}
Employing the Riesz-Fischer theorem for $L^1(\Omega,\sK)$ we can find integers
$n_0\le n_1< n_2<\ldots$ such that $\Psi_{n_\ell}\to\Psi$ and 
$\partial_{x_j}\Psi_{n_\ell}\to\partial_{x_j}\Psi$, a.e. on $\Omega$ as $\ell\to\infty$. 
The Riesz-Fischer theorem further implies the existence of some $R\in L^1(\Omega)$ such that 
$\|\partial_{x_j}\Psi_{n_\ell}\|_{\sK}\le R$, a.e. on $\Omega$, for every $\ell\in\NN$. 
(This is not always stated explicitly in every textbook treating the Riesz-Fischer theorem,
but it can usually be read off from the proof; see \cite[Thm.~2.7]{LiebLoss2001}.)
Then the dominated convergence theorem guarantees that, along a subsequence, the right hand side
of \eqref{uffe1} converges to the right hand side of \eqref{tim1} as well. 
Altogether this proves the first identity in \eqref{dia00}. 

Since $\Psi,\partial_{x_j}\Psi\in L^1(\Omega,\sK)$, we may now pass to the limit 
$\delta\downarrow0$ in \eqref{tim1} by dominated convergence,
which yields the second identity in \eqref{dia00}. 
\end{proof}

Next, we fix some conventions concerning weak vector-valued
divergences of $\nu$-tuples of $\sK$-valued functions. 
We shall say that $q\in L^1_\loc(\Lambda,\sK)$ is a weak divergence (computed in $\sK$)
of $\V{G}=(G_1,\ldots,G_\nu)\in L^1_\loc(\Lambda,\sK^\nu)$, iff
\begin{align}\label{defwdiv}
\int_\Lambda\SPn{\nabla \eta(\V{x})}{\V{G}({\V{x}})}_{\sK^\nu}\Id\V{x}
=-\int_\Lambda\SPn{\eta(\V{x})}{q(\V{x})}_\sK\Id\V{x},
\end{align}
for all $\eta\in\sD(\Lambda,\sK)$. Then we also write $\Div\V{G}:=q$.

If $\sK\subset\HP$ is a subspace of the one-boson space, then we shall usually write
$\V{G}_{\V{x}}:=(G_{1,\V{x}},\ldots,G_{\nu,\V{x}}):=\V{G}(\V{x})$ and 
$q_{\V{x}}:=q(\V{x})$ for the latter objects.

\begin{remark}\label{remwdiv}
As in Rem.~\ref{remwpd}(1) we can show that a weak divergence (if any) of 
$\V{G}\in L^1_\loc(\Lambda,\sK^\nu)$ is necessarily unique. Furthermore, to conclude that
$q\in L^1_\loc(\Lambda,\sK)$ is a weak divergence of $\V{G}$ is suffices to check
\eqref{defwdiv} only for test functions of the form $\eta=f\phi$ with $f\in\sD(\Lambda)$ and
$\phi\in\sE$, where $\sE$ is some fixed total subset of $\sK$. Finally, if $p\in[1,\infty]$ and 
$T$ is a non-negative self-adjoint operator in $\sK$, then
$\V{G}\in L^p_\loc(\Lambda,\fdom(T)^\nu)$ has a weak divergence computed in $\sK$ satisfying
$\Div\V{G}\in L_\loc^p(\Lambda,\fdom(T))$, if and only if it has a weak divergence computed in
$\fdom(T)$ which belongs to $L_\loc^p(\Lambda,\fdom(T))$. The last assertion still holds true,
if the subscripts ``loc'' are dropped everywhere.
\end{remark}

\begin{lemma}\label{LeibnizL1div}
Let $p\in[1,\infty]$ with conjugated exponent $p'$. Let
$\sK_1,\sK_2,\sK_3$ be real or complex separable Hilbert spaces and 
$b:\sK_1\times\sK_2\to\sK_3$ be real bilinear and continuous. Suppose that 
$\V{G}\in L^p_\loc(\Lambda,\sK_1^\nu)$ has a weak divergence $q\in L^p_\loc(\Lambda,\sK_1)$ 
and that $\Psi\in L^{p'}_\loc(\Lambda,\sK_2)$ has weak partial derivatives 
$\partial_{x_1}\Psi,\ldots,\partial_{x_\nu}\Psi\in L^{p'}_\loc(\Lambda,\sK_2)$. 
Then $b(\V{G},\Psi):=(b(G_1,\Psi),\ldots,b(G_\nu,\Psi))\in L^1_\loc(\Lambda,\sK_3^\nu)$ 
has the weak divergence 
\begin{align}\label{Leibnizbdiv}
\Div b(\V{G},\Psi)=b(q,\Psi)+\sum_{j=1}^\nu b(G_j,\partial_{x_j}\Psi)
\quad\text{in $L^1_\loc(\Lambda,\sK_3)$.}
\end{align}
\end{lemma}

\begin{proof}
Let $\{e_\ell:\ell\in\NN\}$ be an orthonormal basis of $\sK_2$ and define the projections
$P_n\phi:=\sum_{\ell=1}^n\SPn{e_\ell}{\psi}_{\sK_2}e_\ell$, $\psi\in\sK_2$, $n\in\NN$.
Define $\Psi_n$, $n\in\NN$, as in \eqref{klaus1} and put $\Phi_n:=P_n\Psi_n$.
Since $P_n\to\id_{\sK_2}$ strongly, as $n\to\infty$, it follows from Rem.~\ref{remSob} that
$\Phi_n\to\Psi$ and $\partial_{x_j}\Phi_n\to\partial_{x_j}\Psi$ in $L^{p'}(\Omega,\sK_2)$,
for all measurable $\Omega\Subset\Lambda$ and $j\in\{1,\ldots,\nu\}$, provided that $p'<\infty$.
If $p'=\infty$, then $\Phi_n\to\Psi$ and $\partial_{x_j}\Phi_n\to\partial_{x_j}\Psi$ a.e. on $\Lambda$.

Now let $n\in\NN$ with $\Lambda_n\not=\emptyset$ and
$\eta\in\sD(\Lambda_n,\sK_3)$. Then $\eta=\sum_{i=1}^m\eta_i\phi_i$, for some
$\eta_i\in\sD(\Lambda_n)$, $\phi_i\in\sK_3$, and $m\in\NN$.
By virtue of Riesz' representation theorem we find vectors $g_{i,\ell}\in\sK_1$
representing the bounded real linear functionals 
$\sK_1\ni G\mapsto \SPn{\phi_i}{b(G,e_\ell)}_{\sK_3}$.
Then $\SPn{\eta}{b({q},\Phi_n)}_{\sK_3}=\SPn{\tilde{\eta}_n}{{q}}_{\sK_1}$ as well as
$\SPn{\partial_{x_j}\eta}{b({G_j},\Phi_n)}_{\sK_3}=\SPn{\partial_{x_j}\tilde{\eta}_n}{{G_j}}_{\sK_1}
-\SPn{\eta}{b(G_j,\partial_{x_j}\Phi_n)}_{\sK_3}$
with $\tilde{\eta}_n:=\sum_{i=1}^m\sum_{\ell=1}^n\eta_i\SPn{\Psi_n}{e_\ell}_{\sK_2}g_{i,\ell}$,
so that $\tilde{\eta}_n\in\sD(\Lambda_n,\sK_1)$. Using these observations, we deduce that 
$b(\V{G},\Phi_n)\in L^1_\loc(\Lambda_n,\sK_3)$ has the weak divergence
\begin{align}\label{Leibnizbdiv0}
\Div b(\V{G},\Phi_n)&=b(q,\Phi_n)+\sum_{j=1}^\nu b(G_j,\partial_{x_j}\Phi_n)
\quad\text{in $L^1_\loc(\Lambda_n,\sK_3)$.}
\end{align}

Next, let $\Omega\Subset\Lambda$ be measurable. 
If $p'=\infty$, then we pick some compact $K\subset\Lambda$
with $\ol{\Omega}\subset\mr{K}$ and observe that, for all sufficiently large $n\in\NN$, we have the
dominations $\|\Phi_n\|_{\sK_2}\le\mathrm{ess\,sup}_K\|\Psi\|_{\sK_2}$
and $\|\partial_{x_j}\Phi_n\|_{\sK_2}\le\mathrm{ess\,sup}_K\|\partial_{x_j}\Psi\|_{\sK_2}$,
$j\in\{1,\ldots,\nu\}$, a.e. on $\Omega$.
If $p'<\infty$, then we find integers $1\le n_1<n_2<\ldots$ and $R,R'\in L^{p'}(\Omega)$ such that
$\Phi_{n_\ell}\to\Psi$ and $\partial_{x_j}\Phi_{n_\ell}\to\partial_{x_j}\Psi$, a.e. on $\Omega$
as $\ell\to\infty$, as well as $\|\Phi_{n_\ell}\|_{\sK_2}\le R$
and $\|\partial_{x_j}\Phi_{n_\ell}\|_{\sK_2}\le R'$, a.e. on $\Omega$ for all $\ell\in\NN$.
In all cases, \eqref{Leibnizbdiv} now follows from the defining relation \eqref{defwdiv},
\eqref{Leibnizbdiv0}, the boundedness of $b$, and the dominated convergence theorem.
\end{proof}

%%%%%%%%%%%%%%%%%%%%%%%%%%%%%%%%%%%%%%%%%%%%%%%

\subsection{Some Fock space calculus}\label{ssec-Fock}

\noindent
In this section we recall the definition of the bosonic Fock space and introduce
some important operators acting in it via the Weyl representation. A textbook exposition
of the latter can be found, e.g., in \cite{Parthasarathy1992}.

Recall that the $\sigma$-finite measure space $(\cM,\fA,\mu)$ and the corresponding, by assumption
separable $L^2$-space $\HP$ were introduced in \eqref{defHP} and the paragraph preceding it.
For every $n\in\NN$, let $\mu^n$ denote the $n$-fold product measure of $\mu$ with itself 
defined on the $n$-fold product $\sigma$-algebra $\fA^n$. Let $\sF^{(n)}$ denote the
closed subspace in $L^2(\cM^n,\fA^n,\mu^n)$ of all its elements $\psi^{(n)}$ satisfying
$$
\psi^{(n)}(k_{\pi(1)},\ldots,k_{\pi(n)})=
\psi^{(n)}(k_{1},\ldots,k_{n}),\quad\text{$\mu^n$-a.e. $(k_{1},\ldots,k_{n})\in\cM^n$,}
$$
for all permutations $\pi$ of $\{1,\ldots,n\}$. Then the bosonic Fock space $\sF$ modeled over
$\HP$ is given by
\begin{equation*}
\sF:=\CC\oplus\bigoplus_{n\in\NN}\sF^{(n)}.
\end{equation*}
For every $h\in\HP$, the corresponding exponential vector $\expv{h}\in\sF$ is defined by
\begin{align*}
\expv{h}:=(1,h,2^\mh h^{\otimes_2},\ldots,(n!)^\mh h^{\otimes_n},\ldots\:),
\end{align*}
where $h^{\otimes_n}(k_1,\ldots,k_1):=h(k_1)\cdots h(k_n)$, $\mu^n$-a.e. We observe that
\begin{align}\label{SPexpv}
\SPn{\expv{g}}{\expv{h}}&=e^{\SPn{g}{h}},\quad g,h\in\HP.
\end{align}

Let $f\in\HP$ and $U\in\sU[\HP]$, where $\sU[\sK]$ denotes the set of unitary operators on some 
Hilbert space $\sK$ equipped with the topology corresponding to strong convergence of operators. 
We also set $\cE[\mathfrak{d}]:=\{\expv{h}:h\in\mathfrak{d}\}$ and
$\cC[\mathfrak{d}]:=\mathrm{span}\cE[\mathfrak{d}]$, for any subset $\mathfrak{d}\subset\HP$.
Then $\cE[\HP]$ is linearly independent and $\cC[\mathfrak{d}]$ is dense in $\sF$,
whenever $\mathfrak{d}$ is dense in $\HP$. In particular, the prescription
\begin{align}\label{defWeyl}
W(f,U)\expv{h}:=e^{-\|f\|^2/2-\SPn{f}{Uh}}\expv{f+Uh},\quad h\in\HP,
\end{align}
determines a linear bijection on $\cC[\HP]$.
Since $W(f,U)$ turns out to be isometric, it extends uniquely
to a unitary operator on $\sF$. The latter is again denoted by $W(f,U)$
and called the Weyl operator corresponding to $f$ and $U$. 
Computing on exponential vectors, we verify the Weyl relations
\begin{align}\label{Weylrel}
W(f_1,U_1)W(f_2,U_2)=e^{-i\Im\SPn{f_1}{U_1f_2}}W(f_1+U_1f_2,U_1U_2),
\end{align}
for all $f_1,f_2\in\HP$ and $U_1,U_2\in\sU[\HP]$.
The thus obtained Weyl representation 
$W:\HP\times\sU[\HP]\to\sU[\sF]$, $(f,U)\mapsto W(f,U)$, is strongly continuous. 

These remarks imply that $(W(-itf,\id))_{t\in\RR}$ is a strongly continuous unitary group, whose 
self-adjoint generator is denoted by $\vp(f)$ and called the field operator corresponding to $f\in\HP$. 
In view of \eqref{SPexpv} and \eqref{defWeyl} we then have, for instance,
\begin{align}\label{mevpexpv}
\SPn{\expv{g}}{\vp(f)\expv{h}}&=\big(\SPn{f}{h}+\SPn{g}{f}\big)e^{\SPn{g}{h}},\quad f,g,h\in\HP.
\end{align}
Recall that, throughout the whole article, $\omega:\cM\to[0,\infty)$ is assumed to be measurable 
and $\mu$-a.e. strictly positive. We denote the associated maximal multiplication operator again by 
$\omega$. Then $(W(0,e^{-it\omega}))_{t\in\RR}$ is a strongly continuous unitary group as well. 
Its generator is denoted by $\Id\Gamma(\omega)$ and called the 
(differential) second quantization of $\omega$. 

We shall sometimes use the following fact, where $\mathfrak{d}$ might for instance be
$\dom(\omega^n)$ with $n\in\NN$ or the set of analytic vectors of $\omega$.

\begin{lemma}\label{lemdGammaesa} 
Let $\mathfrak{d}$ be any subset of $\dom(\omega)$ which is dense in $\HP$ and such that 
$e^{-it\omega}\mathfrak{d}\subset\mathfrak{d}$, for all $t\in\RR$.
Then $\Id\Gamma(\omega)$ is essentially self-adjoint on $\cC[\mathfrak{d}]$.
\end{lemma}

\begin{proof} 
It is straightforward to show that the map $\HP\ni h\mapsto\expv{h}\in \sF$ is analytic.
Since $\RR\ni t\mapsto e^{-it\omega}h$, belongs to $C^1(\RR,\HP)$, if $h\in\dom(\omega)$,
it is thus clear for a start that $\cC[\mathfrak{d}]\subset\dom(\Id\Gamma(\omega))$.

The claim now follows from the invariant domain method (see, e.g., \cite[p.~366]{Chernoff1977}):
In fact, let $A$ denote the restriction of $\Id\Gamma(\omega)$ to $\cC[\mathfrak{d}]$ and
suppose that $\psi\in\dom(A^*)$ satisfies $A^*\psi=\pm i\psi$, for some choice of the sign. 
Since $W(0,e^{-it\omega})\expv{h}=\expv{e^{-it\omega}h}$ and
$e^{-it\omega}h\in\mathfrak{d}$, for all $h\in\mathfrak{d}$, we see that 
$W(0,e^{-it\omega})$ maps $\cC[\mathfrak{d}]$ into itself, for all $t\in\RR$. Now, let 
$\phi\in\cC[\mathfrak{d}]$ and set $b(t):=\SPn{W(0,e^{-it\omega})\phi}{\psi}$. Then
\begin{align*}
b'(t)&=\SPn{\Id\Gamma(\omega)W(0,e^{-it\omega})\phi}{i\psi}
\\
&=\SPn{AW(0,e^{-it\omega})\phi}{i\psi}
=\SPn{W(0,e^{-it\omega})\phi}{iA^*\psi}=\mp b(t),\quad t\in\RR,
\end{align*}
i.e., $b(t)=b(0)e^{\mp t}$. Since $b$ is bounded, we must have
$0=b(0)=\SPn{\phi}{\psi}$. Since $\cC[\mathfrak{d}]$ is dense in $\sF$, this implies $\psi=0$.
\end{proof}

Under the assumption $f\in\dom(\omega^\mh)=\fdom(\omega^{-1})$ it is known 
(see, e.g., \cite{BFS1998b}) that $\dom(\Id\Gamma(\omega)^\eh)\subset\dom(\vp(f))$ and
\begin{align}\label{rbvp}
\|\vp(f)\psi\|&\le2\|(\omega^{\mh}\vee1)f\|\|(1+\Id\Gamma(\omega))^\eh\psi\|,
\\\label{qfbvp}
|\SPn{\phi}{\vp(f)\psi}|&\le\|\omega^\mh f\|\big(\|\Id\Gamma(\omega)^\eh\phi\|\|\psi\|
+\|\phi\|\|\Id\Gamma(\omega)^\eh\psi\|\big),
\end{align}
for all $\phi,\psi\in\dom(\Id\Gamma(\omega)^\eh)$. Furthermore,
\begin{align*}
\vp(f_1+\lambda f_2)\psi=\vp(f_1)\psi+\lambda\vp(f_2)\psi,\quad
f_1,f_2\in\fdom(\omega^{-1}),\,\lambda\in\RR,
\end{align*}
for all $\psi\in\dom(\Id\Gamma(\omega)^\eh)$.
The latter two remarks imply that
\begin{align}\label{contvp}
\text{$\fdom(\omega^{-1})\times\fdom(\Id\Gamma(\omega))\ni(f,\psi)\mapsto\vp(f)\psi\in\sF$ 
is continuous.}
\end{align}
Here $\fdom(\omega^{-1})$ and $\fdom(\Id\Gamma(\omega))$ are equipped with the form norms 
of $\omega^{-1}$ and $\Id\Gamma(\omega)$, respectively.

We shall also make use of the commutation relation
\begin{align}\label{CRvpdG}
\SPn{\vp(f)\phi}{\Id\Gamma(\omega)\psi}-\SPn{\Id\Gamma(\omega)\phi}{\vp(f)\psi}
&=\SPn{\phi}{i\vp(i\omega f)\psi},
\end{align}
valid for all $f\in\dom(\omega^\mh+\omega)$ and $\phi,\psi\in\dom(\Id\Gamma(\omega))$.
It can be verified by taking derivatives at $(s,t)=(0,0)$ of
\begin{align*}
\RR\times\RR\ni(s,t)\longmapsto\SPn{&W(-isf,\id)\expv{g}}{W(0,e^{-it\omega})\expv{h}}
\\
&\quad-\SPn{W(0,e^{-it\omega})\expv{g}}{W(-isf,\id)\expv{h}},
\end{align*}
with $g,h\in\dom(\omega)$, before and after applying \eqref{defWeyl} and \eqref{Weylrel}, 
and using that $\cC[\dom(\omega)]$ is a core for $\Id\Gamma(\omega)$ together with \eqref{rbvp}.
Combining \eqref{qfbvp}, \eqref{contvp}, and \eqref{CRvpdG}, we see that the bound
\begin{align}\nonumber
\big|\SPn{\vp(f)\phi}{&(r+\Id\Gamma(\omega))\psi}-\SPn{(r+\Id\Gamma(\omega))\phi}{\vp(f)\psi}\big|
\\\label{bdCRvpdG}
&\le\|\omega^\eh f\|\big(\|\Id\Gamma(\omega)^\eh\phi\|\|\psi\|
+\|\phi\|\|\Id\Gamma(\omega)^\eh\psi\|\big),\quad r\in\RR,
\end{align}
holds for all $\phi,\psi\in\dom(\Id\Gamma(\omega))$
under the weaker assumption $f\in\fdom(\omega^{-1}+\omega)$.

The next lemma follows, e.g., from a more general discussion in \cite[\textsection12]{Matte2016},
but we shall give a shorter and independent proof for the convenience of the reader.

\begin{lemma}\label{lem-comm}
For all $f\in\fdom(\omega^{-1}+\omega)$, the following assertions hold true:
\begin{enumerate}[leftmargin=0.67cm]
\item[{\rm(1)}] Let $\ve>0$ and set $\theta_\ve:=1+\ve\Id\Gamma(\omega)$.
Then the operator defined by $\dom(C_\ve(f)):=\fdom(\Id\Gamma(\omega\wedge1))$ and
$$
C_\ve(f)\psi:=\theta_\ve^\mh\vp(f)\psi-\vp(f)\theta_\ve^\mh\psi,\quad\psi\in\dom(C_\ve(f)),
$$ 
is bounded with
$$
\|C_\ve(f)\|\le (4/\pi)\ve^\eh\|\omega^\eh f\|.
$$
\item[{\rm(2)}] Set $\theta:=1+\Id\Gamma(\omega)$.
Then the operator given by $\dom(T(f)):=\fdom(\Id\Gamma(\omega\wedge1))$ and
$$
T(f)\psi:=\theta^\eh\vp(f)\theta^\mh\psi-\vp(f)\psi,\quad\psi\in\dom(T(f)),
$$ 
is well-defined and bounded with
$$
\|T(f)\|\le 2\|\omega^\eh f\|.
$$
Furthermore, $T(f)^*\!\!\upharpoonright_{\dom(\Id\Gamma(\omega))}\subset C_1(f)\theta^\eh$.
\item[{\rm(3)}] $\vp(f)$ maps $\dom(\Id\Gamma(\omega))$ into $\dom(\Id\Gamma(\omega)^\eh)$.
\end{enumerate}
\end{lemma}

\begin{proof}
Inserting $\phi=(r+\Id\Gamma(\omega))^{-1}\xi$ and $\psi=(r+\Id\Gamma(\omega))^{-1}\eta$
with $r>0$ and $\xi,\eta\in\fdom(\Id\Gamma(\omega\wedge1))$ into \eqref{bdCRvpdG}, we obtain
\begin{align}\label{prometeus1}
&\Big|\SPb{\xi}{(r+\Id\Gamma(\omega))^{-1}\vp(f)\eta}
-\SPb{\vp(f)\xi}{(r+\Id\Gamma(\omega))^{-1}\eta}\Big|
\\\nonumber
&\le\|\omega^\eh f\|
\bigg(\big\|\Id\Gamma(\omega)^\eh(r+\Id\Gamma(\omega))^{-1}\xi\big\|\frac{\|\eta\|}{r}+
\big\|(r+\Id\Gamma(\omega))^{-1}\xi\big\|\frac{\|\eta\|}{r^\eh}\bigg).
\end{align}
Applying this bound with $r=(1+t)/\ve$, $t\ge0$, and observing that the formula
$A^\mh=\int_0^\infty(t+A)^{-1}\Id t/\pi t^\eh$, valid for any strictly positive self-adjoint operator
$A$ in some Hilbert space, implies
\begin{align*}
&\SPn{\xi}{C_\ve(f)\eta}
\\
&=\int_0^\infty\Big(\SPb{\xi}{(t+1+\ve\Id\Gamma(\omega))^{-1}\vp(f)\eta}
-\SPb{\vp(f)\xi}{(t+1+\ve\Id\Gamma(\omega))^{-1}\eta}\Big)\frac{\Id t}{\pi t^\eh},
\end{align*}
we deduce that
\begin{align*}
|\SPn{\xi}{C_\ve(f)\eta}|&\le2\ve^\eh\|\omega^\eh f\|\int_0^\infty\frac{1}{(1+t)^{\nf{3}{2}}}
\frac{\Id t}{\pi t^\eh}\|\xi\|\|\eta\|,
\end{align*}
for all $\xi,\eta\in\fdom(\Id\Gamma(\omega\wedge1))$,
which proves Part~(1).

Choosing $r=1+t$ and $\xi=\theta^\eh\zeta$ with $\zeta\in\dom(\Id\Gamma(\omega))$
we further infer from \eqref{prometeus1} that
\begin{align}\label{prometeus2}
|\SPn{\theta^\eh\zeta}{C_1(f)\eta}|&\le2\|\omega^\eh f\|\int_0^\infty\frac{1}{1+t}
\frac{\Id t}{\pi t^\eh}\|\zeta\|\|\eta\|,
\end{align}
for all $\eta\in\fdom(\Id\Gamma(\omega\wedge1))$.
Since $\dom(\Id\Gamma(\omega))$ is a core for $\theta^\eh$,
\eqref{prometeus2} implies that the range of $C_1(f)$ is contained in $\dom(\theta^\eh)$ 
and that $\theta^\eh C_1(f)$ is bounded with $\|\theta^\eh C_1(f)\|\le 2\|\omega^\eh f\|$. 
Now, if $\psi\in\fdom(\Id\Gamma(\omega\wedge1))$, then
$\vp(f)\theta^\mh\psi=\theta^\mh\vp(f)\psi-C_1(f)\psi$. We conclude that
$\vp(f)\theta^\mh\psi\in\dom(\theta^\eh)$, which proves Part~(3) and shows
that $T(f)$ is well-defined with $T(f)\psi=-\theta^\eh C_1(f)\psi$. The latter relation
finally entails $T(f)^*\supset-C_1(f)^*\theta^\eh$ and it is clear that
$-C_1(f)^*\!\!\upharpoonright_{\dom(\theta^\eh)}=C_1(f)\!\!\upharpoonright_{\dom(\theta^\eh)}$.
\end{proof}

\begin{remark}\label{remrbvp2}
Let $f,g\in\fdom(\omega^{-1})$ and $\psi\in\dom(\Id\Gamma(\omega))$. 
Then we may apply Lem.~\ref{lem-comm} with $\omega$
replaced by $\omega\wedge 1$, so that $\theta=1+\Id\Gamma(\omega\wedge1)$. 
According to its third part $\vp(f)\psi\in\dom(\theta^\eh)\subset\dom(\vp(g))$, 
and we may write $\theta^\eh\vp(f)\psi=T(f)\theta^\eh\psi+\vp(f)\theta^\eh\psi$. Taking also
\eqref{rbvp} (with $\omega$ replaced by $\omega\wedge1$) and the bound in the second part 
of Lem.~\ref{lem-comm} into 
account, we obtain $\|\theta^\eh\vp(f)\psi\|\le4\|(\omega^\mh\vee1)f\|\|\theta\psi\|$.
Applying \eqref{rbvp} once more we arrive at the familiar bound
\begin{align}\label{rbvp2}
\|\vp(g)\vp(f)\psi\|
&\le8\|(\omega^\mh\vee1)g\|\|(\omega^\mh\vee1)f\|\|(1+\Id\Gamma(\omega))\psi\|.
\end{align}
\end{remark}

%%%%%%%%%%%%%%%%%%%%%%%%%%%%%%%%%%%%%%%%%%%%%%%%%

\subsection{Discussion of the interaction terms}\label{ssec-inter}

\noindent
Next, we discuss direct integrals of the field operators introduced in the previous subsection
employing the remarks given in Subsect.~\ref{ssec-Sobolev}. Our main aim is to verify a
Leibniz rule involving classical and quantized vector potentials.

First, we observe that, if $\Lambda\ni\V{x}\mapsto G_{\V{x}}\in\HP$ is measurable,
then the strong continuity of the Weyl representation implies measurability of the
maps $\Lambda\ni\V{x}\mapsto e^{i\vp(G_{\V{x}})}\psi$ with $\psi\in\sF$. Therefore,
the direct integral $\vp(G):=\int_{\Lambda}\vp(G_{\V{x}})\Id\V{x}$
is a well-defined self-adjoint operator in $L^2(\Lambda,\sF)$. The symbol $\dom(\vp(G))$ will always
denote the domain of the latter direct integral operator.

Since the weak partial derivatives of magnetic Sobolev functions are in general not
square-integrable we shall, under some extra assumptions on $G$ and the vectors it is applied to, 
generalize the meaning of the symbol $\vp(G)$ as follows: 
Consider $\fdom(\omega^{-1})$ as a Hilbert space equipped with the form 
norm of $\omega^{-1}$ and assume that
$G:\Lambda\to\fdom(\omega^{-1})$ and $\Psi:\Lambda\to\fdom(\Id\Gamma(\omega))$
are measurable. Then \eqref{contvp} shows that
\begin{align*}
(\vp(G)\Psi)(\V{x})=\vp(G_{\V{x}})\Psi(\V{x}),\quad\text{a.e. $\V{x}\in\Lambda$,}
\end{align*}
defines an equivalence class $\vp(G)\Psi$ of measurable functions from $\Lambda$ to $\sF$. If 
$\Lambda\ni\V{x}\mapsto\|G_{\V{x}}\|_{\fdom(\omega^{-1})}\|\Psi(\V{x})\|_{\fdom(\Id
\Gamma(\omega))}$
is in $L^p_\loc(\Lambda)$, for some $p\in[1,\infty]$, then $\vp(G)\Psi\in L^p_\loc(\Lambda,\sF)$ on 
account of \eqref{rbvp}; the same remark holds true, if the subscripts ``loc'' are dropped.

The next lemma is a generalization of \cite[Lem.~13]{HaslerHerbst2008}.  

\begin{lemma}\label{lem-Leibniz-vp}
Let $p\in[1,\infty]$ and $\V{G}\in L^{p'}_\loc(\Lambda,\fdom(\omega^{-1})^\nu)$ have a weak 
divergence $q:=\mathrm{div}\V{G}\in L^{p'}_\loc(\Lambda,\fdom(\omega^{-1}))$.
Assume that  $\Psi\in L^{p}_\loc(\Lambda,\fdom(\Id\Gamma(\omega)))$ has
weak partial derivatives with respect to all variables such that 
$\partial_{x_1}\Psi,\ldots,\partial_{x_\nu}\Psi\in L^p_\loc(\Lambda,\fdom(\Id\Gamma(\omega)))$. 
Then $\vp(q)\Psi$ and $\vp(\V{G})\cdot\nabla\Psi:=\sum_{\ell=1}^\nu\vp(G_\ell)\partial_{x_\ell}\Psi$ 
belong to $L^1_\loc(\Lambda,\sF)$ and
$\vp(\V{G})\Psi:=(\vp(G_1)\Psi,\ldots,\vp(G_\nu)\Psi)\in L^1_\loc(\Lambda,\sF^\nu)$ 
has a weak divergence given by
$$
\mathrm{div}(\vp(\V{G})\Psi)=\vp(q)\Psi+\vp(\V{G})\cdot\nabla\Psi
\quad\text{in $L^1_\loc(\Lambda,\sF)$}.
$$
\end{lemma}

\begin{proof}
The assertion follows from Lem.~\ref{LeibnizL1div} and \eqref{contvp}.
\end{proof}

We shall need a variant of the previous lemma including a classical vector potential
$\V{A}=(A_1,\ldots,A_\nu)\in L_\loc^2(\Lambda,\RR^\nu)$. We shall first define scalar and
Fock space-valued weak covariant derivatives associated with $\V{A}$. As we will apply them only to 
square-integrable functions, we shall invoke some Hilbert space theory in their definitions.
So, let $j\in\{1,\ldots,\nu\}$. Then the corresponding (maximal) scalar covariant derivative is
the adjoint of the symmetric operator in $L^2(\Lambda)$ given by
\begin{align}\label{defwjscal}
{\sf w}_j:={(-i\partial_{x_j}-A_j)\!\!\upharpoonright_{\sD(\Lambda)}}.
\end{align}
In analogy we define a symmetric operator in $L^2(\Lambda,\sF)$ by
\begin{align}\label{def-wj}
w_{j}&:={(-i\partial_{x_j}-A_j)\!\!\upharpoonright_{\sD(\Lambda,\sF)}}.
\end{align}
\begin{remark}\label{remwCCw}
Let $j\in\{1,\ldots,\nu\}$. Then the following holds:
\begin{enumerate}[leftmargin=0.67cm]
\item[{\rm(1)}]
$\dom({\sf w}_j^*)\otimes\sF\subset\dom(w_{j}^*)$ and
$w_{j}^*(f\psi)=({\sf w}_j^*f)\psi$, for all $f\in\dom({\sf w}_j^*)$ and $\phi\in\sF$.
Furthermore, $\SPn{\phi}{\Psi}_{\sF}\in\dom({\sf w}_j^*)$ with 
${\sf w}_j^*\SPn{\phi}{\Psi}_{\sF}=\SPn{\phi}{w_{j}^*\Psi}_{\sK}$, for all $\phi\in\sF$ and 
$\Psi\in\dom(w_{j}^*)$.

Both observations follow from $w_{j}(g\phi)=({\sf w}_jg)\phi$, $g\in \sD(\Lambda)$, $\phi\in\sF$.
\item[{\rm(2)}]
A function $\Psi\in L^2(\Lambda,\sF)$ belongs to $\dom(w_{j}^*)$, if and only if it has a weak partial 
derivative with respect to $x_j$ such that the sum of locally integrable functions 
$-i\partial_{x_j}\Psi-A_j\Psi$ is in $L^2(\Lambda,\sF)$. In the affirmative case
\begin{align}\label{WAWloc}
w_{j}^*\Psi=-i\partial_{x_j}\Psi-A_j\Psi\quad\text{in $L^1_\loc(\Lambda,\sF)$. }
\end{align}

This follows from the fact that ${A}_j\Phi\in L^1_\loc(\Lambda,\sF)$, for all $\Phi\in L^2(\Lambda,\sF)$, 
and from the definitions of the adjoint $w_{j}^*$ and the weak partial derivative. 
\end{enumerate}
\end{remark}

In the next lemma and henceforth we shall use the shorthands
\begin{align}\label{shorty0}
\V{A}\cdot\vp(\V{G})\Psi&:=\sum_{\ell=1}^\nu A_\ell\vp(G_\ell)\Psi,\quad
\vp(\V{G})\cdot\V{w}^*\Psi:=\sum_{\ell=1}^\nu\vp(G_\ell)w_{\ell}^*\Psi.
\end{align}

\begin{lemma}\label{lemmagnLeibniz}
Let $\V{A}\in L^2_\loc(\Lambda,\RR^\nu)$, 
let $\Psi\in L^2(\Lambda,\fdom(\Id\Gamma(\omega)))\cap\bigcap_{j=1}^\nu\dom(w_j^*)$ 
be such that $w_j^*\Psi\in L^2(\Lambda,\fdom(\Id\Gamma(\omega)))$, for all $j\in\{1,\ldots,\nu\}$, and 
assume that $\V{G}\in L^\infty(\Lambda,\fdom(\omega^{-1})^\nu)$ has a weak divergence satisfying  
$q:=\mathrm{div}\V{G}\in L^\infty(\Lambda,\fdom(\omega^{-1}))$.
Then $\vp(\V{G})\cdot\V{w}^*\Psi,\vp(q)\Psi\in L^2(\Lambda,\sF)$ and 
\begin{equation}\label{bakithi100div}
\sum_{j=1}^\nu\SPb{{w}_j\Phi}{\vp(G_j)\Psi}
=\SPb{\Phi}{\vp(\V{G})\cdot\V{w}^*\Psi-i\vp(q)\Psi},\quad\Phi\in\sD(\Lambda,\sF).
\end{equation}
\end{lemma}

\begin{proof}
The assumption on $\Psi$ and \eqref{WAWloc} entail
$\partial_{x_j}\Psi\in L^1_\loc(\Lambda,\fdom(\Id\Gamma(\omega)))$, for all 
$j\in\{1,\ldots,\nu\}$. The assumptions on $\V{G}$ and Lem.~\ref{lem-Leibniz-vp} now imply
$\vp(\V{G})\cdot\nabla\Psi,\vp(q)\Psi\in L^1_\loc(\Lambda,\sF)$ and
$\mathrm{div}(\vp(\V{G})\Psi)=\vp(\V{G})\cdot\nabla\Psi+\vp(q)\Psi$. In view of \eqref{rbvp}
we further have $\vp(\V{G})\Psi\in L^2(\Lambda,\sF^\nu)$, thus
$\V{A}\cdot\vp(\V{G})\Psi\in L^1_\loc(\Lambda,\sF)$. Hence, by \eqref{WAWloc},
\begin{equation}\label{bakithi99b}
-i\mathrm{div}(\vp(\V{G})\Psi)-\V{A}\cdot\vp(\V{G})\Psi=\vp(\V{G})\cdot\V{w}^*\Psi
-i\vp(q)\Psi\quad\text{in $L^1_\loc(\Lambda,\sF)$.}
\end{equation}
A posteriori, the conditions $\Psi,w_j^*\Psi\in L^2(\Lambda,\fdom(\Id\Gamma(\omega)))$ 
and \eqref{rbvp} show that both summands on the right hand side of \eqref{bakithi99b} actually 
belong to $L^2(\Lambda,\sF)$. Scalar multiplying \eqref{bakithi99b} with $\Phi\in\sD(\Lambda,\sF)$ 
we thus arrive at \eqref{bakithi100div}.
\end{proof}

%%%%%%%%%%%%%%%%%%%%%%%%%%%%%%%%%%%%%%%%%%%%%%%%%
%%%%%%%%%%%%%%%%%%%%%%%%%%%%%%%%%%%%%%%%%%%%%%%%%
%%%%%%%%%%%%%%%%%%%%%%%%%%%%%%%%%%%%%%%%%%%%%%%%%

\section{A diamagnetic inequality}\label{sec-dia}

\noindent
Our aim in the following is to derive a pointwise diamagnetic inequality for a sum of a classical and 
a quantized field by transferring the proof of \cite[Thm.~7.21]{LiebLoss2001} to the vector-valued 
case. There already exist a number of diamagnetic inequalities with different proofs in the literature 
on non-relativistc quantum electrodynamics. For infra-red regularized vector potentials, 
a diamagnetic inequality for the 
semi-group has been proven by Trotter-product expansions and dressing transformations in 
\cite{Hiroshima1996}. A Feynman-Kac formula that immediately leads to a diamagnetic inequality 
for the semi-group has been derived in \cite{Hiroshima1997}. Analytic proofs for various 
diamagnetic inequalities, in particular a generalized Kato-type inequality and an inequality for the 
square root of the Laplacian with quantized vector potential, have been worked out in
\cite{KMS2013} by adapting ideas from \cite{ReedSimonII,Simon1979}.
We did not find the pointwise bound \eqref{dia1} in the literature. For more information and references 
on the classical case we refer the reader to \cite{HundertmarkSimon}.

Before we prove the diamagnetic inequality at the end of this subsection we shall first
discuss some basic properties of the Fock space-valued covariant derivative appearing
in it. The latter is defined as the adjoint of the symmetric operator $v_j$ in $L^2(\Lambda,\sF)$
given as follows:
If $\V{A}\in L^2_\loc(\Lambda,\RR^\nu)$, $\V{G}\in L^2_\loc(\Lambda,\HP^\nu)$, 
and $j\in\{1,\ldots,\nu\}$, then we set $\dom(v_j):=\sD(\Lambda,\fdom(\Id\Gamma(1)))$ and
\begin{align}\label{def-vj}
v_j\Psi&:=(-i\partial_{x_j}-A_j-\vp(G_j))\Psi,\quad\Psi\in\dom(v_j).
\end{align}

\begin{remark}\label{remvj}
Let $\V{A}\in L_\loc^2(\Lambda,\RR^\nu)$, $\V{G}\in L_\loc^2(\Lambda,\fdom(\omega^{-1})^\nu)$,
$j\in\{1,\ldots,\nu\}$, and let $\Psi:\Lambda\to\fdom(\Id\Gamma(\omega\wedge1))$ be measurable. 
Then the following assertions follow easily from the definitions of $v_j^*$ and $w_{j}^*$,
\eqref{rbvp}, and Rem.~\ref{remwpd}(2):
\begin{enumerate}[leftmargin=0.67cm]
\item[{\rm(1)}] Assume that the map
$\V{x}\mapsto\|G_{j,\V{x}}\|_{\fdom(\omega^{-1})}\|\Psi(\V{x})\|_{\fdom(\Id\Gamma(\omega\wedge1))}$
is in $L^2(\Lambda)$. Then $\vp(G_j)\Psi\in L^2(\Lambda,\sF)$ and the equivalence
$\Psi\in\dom(w_{j}^*)$ $\Leftrightarrow$ $\Psi\in\dom(v_j^*)$ holds.
In the affirmative case $v_j^*\Psi=w_{j}^*\Psi-\vp(G_j)\Psi$.
\item[{\rm(2)}] Assume that $\Psi\in L^2(\Lambda,\sF)$ and
$\V{x}\mapsto\|G_{j,\V{x}}\|_{\fdom(\omega^{-1})}\|\Psi(\V{x})\|_{\fdom(\Id\Gamma(\omega\wedge1))}$
is in $L^1_\loc(\Lambda)$. Then $A_j\Psi,\vp(G_{j})\Psi\in L^1_\loc(\Lambda,\sF)$. Furthermore,
$\Psi\in\dom(v_j^*)$, if and only if $\Psi$ has a weak partial derivative with respect to $x_j$ and 
$-i\partial_{x_j}\Psi-A_j\Psi-\vp(G_{j})\Psi\in L^2(\Lambda,\sF)$. In the affirmative case 
$v_j^*\Psi=-i\partial_{x_j}\Psi-A_j\Psi-\vp(G_j)\Psi$.
\end{enumerate}
\end{remark}

Recall that the operator $C_\ve(f)$ has been defined in Lem.~\ref{lem-comm}(1). Since the 
real linear map
$\fdom(\omega)\ni f\mapsto C_\ve(f)^*\in\LO(\sF)$ is continuous by Lem.~\ref{lem-comm}(1),
the formula $(C^*_\ve(G)\Psi)(\V{x}):=C_\ve(G_{\V{x}})^*\Psi(\V{x})$, $\V{x}\in\Lambda$, defines a
measurable $\sF$-valued function, for all measurable $G:\Lambda\to\fdom(\omega)$
and $\Psi:\Lambda\to\sF$.

\begin{lemma}\label{lemvj} 
Assume that $\V{A}\in L_\loc^2(\Lambda,\RR^\nu)$ and 
$\V{G}\in L_\loc^2(\Lambda,\fdom(\omega^{-1}+\omega)^\nu)$. Let $j\in\{1,\ldots,\nu\}$ and 
$\Psi\in\dom(v_j^*)$. Furthermore, let $\ve>0$ and set 
$\theta_\ve:=1+\ve\Id\Gamma(\omega)$ and $\Psi_\ve:=\theta_\ve^\mh\Psi$. Then 
$\vp(G_{j})\Psi_\ve$, $A_j\Psi$, and $C^*_\ve(G)\Psi$ all belong to $L^1_\loc(\Lambda,\sF)$, 
and $\Psi_\ve$ has a weak partial derivative with respect to $x_j$ given by 
\begin{align}\label{partialjvjstar}
\partial_{x_j}\Psi_\ve&=i\theta_\ve^\mh v_j^*\Psi+iA_j\Psi_\ve+i\vp(G_j)\Psi_\ve
+iC_\ve^*(G_{j})\Psi\quad\text{in $L^1_\loc(\Lambda,\sF)$.}
\end{align}
\end{lemma}

\begin{proof} 
First, we observe that $A_j\in L^2_\loc(\Lambda,\RR)$ implies
$A_j\Psi\in L^1_\loc(\Lambda,\sF)$ and \eqref{rbvp} entails 
$\vp(G_{j})\Psi_\ve\in L^1_\loc(\Lambda,\sF)$.
Lem.~\ref{lem-comm} further ensures that $C^*_\ve(G)\Psi\in L^1_\loc(\Lambda,\sF)$.
Of course $\theta_\ve^\mh\dom(v_j)\subset\dom(v_j)$, and
Lem.~\ref{lem-comm} implies that, for all $\eta\in\dom(v_j)$,
\begin{align}\label{gudni1}
(v_j\theta_\ve^\mh\eta)(\V{x})=(\theta_\ve^\mh v_j\eta)(\V{x})+C_\ve(G_{j,\V{x}})\eta(\V{x}),
\quad\V{x}\in\Lambda. 
\end{align}
For all $\eta\in\dom(v_j)$, we thus obtain
\begin{align*}
\SPn{\eta&}{\theta_\ve^\mh v^*_j\Psi}=\SPn{v_j\theta_\ve^\mh\eta}{\Psi}
\\
&=i\SPn{\partial_{x_j}\eta}{\theta_\ve^\mh\Psi}
\\
&\quad-\int_\Lambda\SPb{\eta(\V{x})}{(A_j(\V{x})
+\vp(G_{j,\V{x}}))\theta_\ve^\mh\Psi(\V{x})+C_\ve(G_{j,\V{x}})^*\Psi(\V{x})}_{\sF}\Id\V{x}.
\end{align*}
We conclude by comparing this with the definition of the weak partial derivatives and 
applying Rem.~\ref{remwpd}(2).
\end{proof}

\begin{lemma}\label{lem-vstar}
Let $\V{A}\in L_\loc^2(\Lambda,\RR^\nu)$, 
$\V{G}\in L_\loc^2(\Lambda,\fdom(\omega^{-1}+\omega)^\nu)$,
$j\in\{1,\ldots,\nu\}$, $\Psi\in\dom(v_j^*)$, $\theta_\ve:=1+\ve\Id\Gamma(\omega)$, and
$\Psi_\ve:=\theta_\ve^\mh\Psi$, $\ve>0$. Assume hat
$\V{x}\mapsto\|\omega^\eh G_{j,\V{x}}\|_{\HP}\|\Psi(\V{x})\|_{\sF}$ is in $L^2(\Lambda)$. 
Then $\Psi_\ve\in\dom(v_j^*)$, for 
every $\ve>0$, and $\Psi_\ve\to\Psi$, $\ve\downarrow0$, with respect to the graph norm of $v_j^*$.
If we additionally assume that $\V{x}\mapsto\|\omega^\eh G_{j,\V{x}}\|_{\HP}$ is essentially
bounded on $\Lambda$, then each $\theta_\ve$, $\ve>0$, maps the graph of $v_j^*$
continuously into itself.
\end{lemma}

\begin{proof}
From \eqref{gudni1} we infer that $\Psi_\ve\in\dom(v_j^*)$ with 
$v_j^*\Psi_\ve=\theta_\ve^\mh v_j^*\Psi-C_\ve^*(G_j)\Psi$, where
$C_\ve^*(G_j)\Psi\to0$, $\ve\downarrow0$, in $L^2(\Lambda,\sF)$ due to 
Lem.~\ref{lem-comm}(1) and the assumptions on $\Psi$.
If $\V{x}\mapsto\|\omega^\eh G_{j,\V{x}}\|_{\HP}$ is essentially bounded, then every
$C_\ve^*(G_j)$, $\ve>0$, defines a bounded operator on $L^2(\Lambda,\sF)$ 
by Lem.~\ref{lem-comm}(1).
\end{proof}

\begin{lemma}\label{lemvjstarlevel}
Let $\V{A}\in L_\loc^2(\Lambda,\RR^\nu)$, $\V{G}\in L_\loc^2(\Lambda,\HP^\nu)$,
$j\in\{1,\ldots,\nu\}$, and $\Psi\in\dom(v_j^*)$.
Then $v_j^*\Psi=0$ almost everywhere on $\{\Psi=0\}$.
\end{lemma}

\begin{proof}
Let $\ve>0$ and set 
$\theta_\ve:=1+\ve\Id\Gamma(1)$ and $\Psi_\ve:=\theta_\ve^\mh\Psi$.
By Lem.~\ref{lem-vstar} (applied to $\omega=1$) we know that $\Psi_\ve\in\dom(v_j^*)$.
Thanks to Rem.~\ref{remvj}(2) (applied to $\omega=1$) we may conclude that $\Psi_\ve$
has a weak partial derivative with respect to $x_j$ and
$v_j^*\Psi_\ve=-i\partial_{x_j}\Psi_\ve-A_j\Psi_\ve-\vp(G_j)\Psi_\ve$ in $L^1_\loc(\Lambda,\sF)$.
Of course, $A_j\Psi_\ve+\vp(G_j)\Psi_\ve=0$ a.e. on $\{\Psi=0\}$.
Furthermore, $\partial_{x_j}\Psi_\ve=0$ a.e. on $\{\Psi=0\}$ according to Rem.~\ref{remgradlevel}.
In view of Lem.~\ref{lem-vstar} and the Riesz-Fischer theorem we finally find a zero
sequence $\ve_n>0$, $n\in\NN$, such that $v_j^*\Psi_{\ve_n}\to v_j^*\Psi$ a.e. on $\Lambda$,
which altogether proves the lemma.
\end{proof}

In the following theorem we again use the notation introduced in \eqref{defsgn}.

\begin{theorem}\label{thm-dia0}
Assume that $\V{A}\in L_\loc^2(\Lambda,\RR^\nu)$ and 
$\V{G}\in L_\loc^2(\Lambda,\HP^\nu)$.
Let $j\in\{1,\ldots,\nu\}$ and $\Psi\in\dom(v_j^*)$. Then $\|\Psi\|_\sF:\Lambda\to\RR$ has a 
weak partial derivative with respect to $x_j$ that belongs to $L^2(\Lambda)$ and is given by
\begin{align}\label{dia0}
\partial_{{x}_j}\|\Psi(\V{x})\|_{\sF}=\Re\SPn{{\fS_\Psi}(\V{x})}{(i{v}_{j}^*\Psi)(\V{x})}_{\sF},
\quad\text{a.e. $\V{x}\in\Lambda$.}
\end{align}
In particular, the following {\em diamagnetic inequality} holds,
\begin{align}\label{dia1}
\big|\partial_{{x}_j}\|\Psi(\V{x})\|_{\sF}\big|&\le\|({v}_{j}^*\Psi)(\V{x})\|_{\sF},
\quad\text{a.e. $\V{x}\in\Lambda$.}
\end{align}
\end{theorem}

\begin{proof}
Let $\theta_\ve:=1+\ve\Id\Gamma(1)$, $\Psi_\ve:=\theta_\ve^\mh\Psi$, $\ve>0$. 
By Lem.~\ref{lem-abs-val} and Lem.~\ref{lemvj} (applied with $\omega=1$)
we may plug $\Psi_\ve$ into the second formula in \eqref{dia00}. Subtracting the expression
\begin{equation}\label{tim5}
\Re\SPb{\fS_{\Psi_\ve}(\V{x})}{i(A_j(\V{x})+\vp(G_{j,\V{x}}))\Psi_\ve(\V{x})}_{\sF}=0,\quad
\text{a.e. $\V{x}\in\Lambda$,}
\end{equation}
from the corresponding right hand side we arrive at
\begin{align}\nonumber
\partial_{x_j}&\|\Psi_\ve(\V{x})\|_\sF
\\\nonumber
&=
\Re\SPb{\fS_{\Psi_\ve}(\V{x})}{\partial_{x_j}\Psi_\ve(\V{x})-i(A_j(\V{x})+\vp(G_{j,\V{x}}))
\Psi_\ve(\V{x})}_{\sF}
\\\label{tim87}
&=\Re\SPb{\fS_{\Psi_\ve}(\V{x})}{i\theta_\ve^\mh(v_j^*\Psi)(\V{x})+iC_\ve(G_{j,\V{x}})^*
\Psi(\V{x})}_{\sF},\quad\text{a.e. $\V{x}\in\Lambda$.}
\end{align}
Notice that $\fS_{\Psi_\ve}(\V{x})\in\fdom(\Id\Gamma(1))$, a.e. $\V{x}\in\Lambda$, 
so that \eqref{tim5} follows from the symmetry of $\vp(G_{j,\V{x}})$ on $\fdom(\Id\Gamma(1))$; 
in the second step of \eqref{tim87} we took \eqref{partialjvjstar} into account.
Since $\theta_{\ve}^\mh$ converges strongly to the identity operator on $\sF$, it is clear that
$\fS_{\Psi_{\ve}}(\V{x})\to\fS_\Psi(\V{x})$ and $\theta_\ve^\mh(v_j^*\Psi)(\V{x})\to(v_j^*\Psi)(\V{x})$, 
for all $\V{x}\in\Lambda$, as $\ve\downarrow0$. 
We may thus invoke the dominated convergence theorem to show that
$\Re\SPn{\fS_{\Psi_{\ve}}}{i\theta_\ve^\mh v_{j}^*\Psi}_{\sF}\to\Re\SPn{{\fS_\Psi}}{iv_{j}^*\Psi}_{\sF}$ 
in $L^2(\Lambda)$, as $\ve\downarrow0$. By virtue of the bound 
$\|C_\ve(G_{j,\V{x}})^*\|\le(4/\pi)\ve^\eh\|{G}_{j,\V{x}}\|$ (due to Lem.~\ref{lem-comm}
applied to $\omega=1$), the fact that 
$\V{x}\mapsto\|{G}_{j,\V{x}}\|\|\Psi(\V{x})\|_{\sF}$ is in $L_\loc^1(\Lambda)$, 
and $\|\Psi_{\ve}\|_{\sF}\to\|\Psi\|_{\sF}$ in $L^2(\Lambda)$, we thus arrive at \eqref{dia0}.
\end{proof}

%%%%%%%%%%%%%%%%%%%%%%%%%%%%%%%%%%%%%%%%%%
%%%%%%%%%%%%%%%%%%%%%%%%%%%%%%%%%%%%%%%%%%
%%%%%%%%%%%%%%%%%%%%%%%%%%%%%%%%%%%%%%%%%%

\section{Definition of the Schr\"odinger and Pauli-Fierz operators}\label{sec-ham}

\noindent
We are now in a position to give precise definitions via quadratic forms of the Schr\"odinger and 
Pauli-Fierz operators we are interested in. 
In the whole section we assume that $V_+,V_-:\Lambda\to\RR$ are non-negative and locally
integrable.  We set $V:=V_+-V_-$.
For the definition of all forms below it suffices to assume that $\V{A}\in L^2_\loc(\Lambda,\RR^\nu)$ 
and $\V{G}\in L^2_\loc(\Lambda,\HP^\nu)$. The latter condition on $\V{G}$ will be
strengthened in the two lemmas and their two corollaries at the end of this subsection.

We recall the notation $\sD(\Lambda)=C_0^\infty(\Lambda)$ and
${\sf w}_j={(-i\partial_{x_j}-A_j)\!\!\upharpoonright_{\sD(\Lambda)}}$
already used earlier. Then a well-known diamagnetic inequality ensures that
$f\in\bigcap_{j=1}^\nu\dom({\sf w}_j^*)$ implies $|f|\in{W}^{1,2}(\Lambda)$ with
\begin{align}\label{dia-w}
\big|\partial_{x_j}|f(\V{x})|\big|
\le|({\sf w}_j^*f)(\V{x})|,\;\;\text{a.e. $\V{x}\in\Lambda$, $j\in\{1,\ldots,\nu\}$;}
\end{align}
see, e.g., \cite[Thm.~7.21]{LiebLoss2001}.
The {\em maximal Schr\"odinger form} associated with $\Lambda$, $\V{A}$, and $V_+$ is defined by 
\begin{align}\nonumber
\dom({\mathfrak{s}}_{\Lambda,\mathrm{N}}^{\V{A},V_+})
&:=\fdom(V_+)\cap\bigcap_{j=1}^\nu\dom({\sf w}_j^*),
\\\label{deffrsN}
{\mathfrak{s}}_{\Lambda,\mathrm{N}}^{\V{A},V_+}[f]&:=\frac{1}{2}\sum_{j=1}^\nu
\|{\sf w}_j^*f\|^2+\int_{\Lambda}V_+(\V{x})|f(\V{x})|^2\Id\V{x},\quad 
f\in\dom({\mathfrak{s}}_{\Lambda,\mathrm{N}}^{\V{A},V_+}).
\end{align}
It is non-negative and closed as a sum of non-negative closed forms 
\cite[Ex.~VI.1.23 \& Thm.~VI.1.31]{Kato}, and the unique self-adjoint operator representing it,
call it $S_{\Lambda,\mathrm{N}}^{\V{A},V_+}$, is interpreted as the Neumann realization of
the Schr\"odinger operator associated with $\Lambda$, $\V{A}$, and $V_+$. The restriction of 
${\mathfrak{s}}_{\Lambda,\mathrm{N}}^{\V{A},V_+}$ to $\sD(\Lambda)$ is closable and its closure,
\begin{align*}
{\mathfrak{s}}_{\Lambda,\mathrm{D}}^{\V{A},V_+}
&:=\ol{{\mathfrak{s}}_{\Lambda,\mathrm{N}}^{\V{A},V_+}\!\!\upharpoonright_{\sD(\Lambda)}},
\end{align*}
is called the {\em minimal Schr\"odinger form}. The unique self-adjoint operator representing 
${\mathfrak{s}}_{\Lambda,\mathrm{D}}^{\V{A},V_+}$, call it $S_{\Lambda,\mathrm{D}}^{\V{A},V_+}$, 
is interpreted as the Dirichlet realization of the Schr\"odinger operator. 

In the case $\Lambda=\RR^\nu$ it is known that ${\mathfrak{s}}_{\RR^\nu,\mathrm{D}}^{\V{A},V_+}
={\mathfrak{s}}_{\RR^\nu,\mathrm{N}}^{\V{A},V_+}$, \cite{SimonJOT1979}.

Next, we add negative parts to the electrostatic potential. Let 
$\diamond\in\{\mathrm{D},\mathrm{N}\}$ and suppose that there exist $a\in[0,1)$ and $b>0$ such that
\begin{align}\label{KLMN}
\int_{\Lambda}V_-(\V{x})|f(\V{x})|^2\Id\V{x}
\le a\mathfrak{s}_{\Lambda,\diamond}^{\V{0},V_+}[f]+b\|f\|^2,\quad 
f\in\dom(\mathfrak{s}_{\Lambda,\diamond}^{\V{0},V_+}).
\end{align}
Note that $|g|\in\dom(\mathfrak{s}_{\Lambda,\mathrm{N}}^{\V{0},V_+})
=W^{1,2}(\Lambda)\cap\fdom(V_+)$, for all 
$g\in\dom(\mathfrak{s}_{\Lambda,\mathrm{N}}^{\V{A},V_+})$. Likewise, a well-known analogue
of Cor.~\ref{cor-dia} below shows that $|g|\in\dom(\mathfrak{s}_{\Lambda,\mathrm{D}}^{\V{0},V_+})$, 
for all $g\in\dom(\mathfrak{s}_{\Lambda,\mathrm{D}}^{\V{A},V_+})$.
From these remarks and \eqref{dia-w} we infer that the inequality in \eqref{KLMN} also holds true
with $\mathfrak{s}_{\Lambda,\diamond}^{\V{0},V_+}$ replaced by
$\mathfrak{s}_{\Lambda,\diamond}^{\V{A},V_+}$, provided that
$f\in\dom(\mathfrak{s}_{\Lambda,\diamond}^{\V{A},V_+})$. Thus, by the KLMN theorem, the form
defined by $\dom({\mathfrak{s}}_{\Lambda,\diamond}^{\V{A},V}):=
\dom({\mathfrak{s}}_{\Lambda,\diamond}^{\V{A},V_+})$ and
\begin{align}\label{KLMN2}
{\mathfrak{s}}_{\Lambda,\diamond}^{\V{A},V}[f]
&:={\mathfrak{s}}_{\Lambda,\diamond}^{\V{A},V_+}[f]
-\int_\Lambda V_-(\V{x})|f(\V{x})|^2\Id\V{x},
\quad f\in\dom({\mathfrak{s}}_{\Lambda,\diamond}^{\V{A},V}),
\end{align}
is semi-bounded and closed, and we denote the unique self-adjoint operator representing it by 
${S}_{\Lambda,\diamond}^{\V{A},V}$.

We shall now mimic these constructions in the case where quantized fields are added. Thus, we put
\begin{align*}
\mathfrak{v}_\Lambda^\pm[\Psi]&:=\int_\Lambda V_\pm(\V{x})\|\Psi(\V{x})\|_{\sF}^2\Id\V{x},\quad
\Psi\in\dom(\mathfrak{v}_\Lambda^\pm):=\fdom(V_\pm\id_{\sF})\subset L^2(\Lambda,\sF),
\end{align*}
and introduce a maximal form
\begin{align*}
{\mathfrak{t}}_{\Lambda,\mathrm{N}}^{\V{G},\V{A},V_+}[\Psi]
&:=\frac{1}{2}\sum_{j=1}^\nu\|{v}_j^*\Psi\|^2+\mathfrak{v}_\Lambda^+[\Psi],\quad 
\Psi\in\dom({\mathfrak{t}}_{\Lambda,\mathrm{N}}^{\V{G},\V{A},V_+})
:=\dom(\mathfrak{v}_\Lambda^+)\cap\bigcap_{j=1}^\nu\dom({v}_j^*),
\end{align*}
and a minimal form
\begin{align*}
{\mathfrak{t}}_{\Lambda,\mathrm{D}}^{\V{G},\V{A},V_+}
&:=\ol{{\mathfrak{t}}_{\Lambda,\mathrm{N}}^{\V{G},\V{A},V_+}\!\!\upharpoonright_{
\sD(\Lambda,\fdom(\Id\Gamma(1)))}}.
\end{align*}
Recall that $v_j$ is defined in \eqref{def-vj} and depends on $\V{A}$ and $\V{G}$. If we set $\V{G}$
equal to zero, then $\ol{v}_j=\ol{w}_{j}$ and $v_j^*=w_{j}^*$. At this point we need the following
observation: 

\begin{corollary}\label{cor-dia}
Let $\V{A}\in L^2_\loc(\Lambda,\RR^\nu)$, $\V{G}\in L^2_\loc(\Lambda,\HP^\nu)$,
and $\diamond\in\{\mathrm{D},\mathrm{N}\}$. Then
$\Psi\in\dom({\mathfrak{t}}_{\Lambda,\diamond}^{\V{G},\V{A},V_+})$ implies
$\|\Psi\|_{\sF}\in\dom(\mathfrak{s}_{\Lambda,\diamond}^{\V{0},V_+})$.
\end{corollary}

\begin{proof}
For $\diamond=\mathrm{N}$, the assertion is immediately clear from Thm.~\ref{thm-dia0}. 

To prove it for $\diamond=\mathrm{D}$, we shall use that
$\sQ\subset\dom(\mathfrak{s}_{\Lambda,\mathrm{D}}^{0,V_+})$ where
$\sQ:=\{f\in W^{1,2}(\Lambda)\cap\fdom(V_+):\,\supp(f)\Subset\Lambda\}$.

Let $\Psi\in\dom({\mathfrak{t}}_{\Lambda,\mathrm{D}}^{\V{G},\V{A},V_+})$. By definition, we 
then find $\Psi_n\in\sD(\Lambda,\fdom(\Id\Gamma(1)))$, $n\in\NN$,
such that $v_j\Psi_n\to{v}_j^*\Psi$, $j\in\{1,\ldots,\nu\}$, $n\to\infty$, 
in $L^2(\Lambda,\sF)$, and $\Psi_n\to\Psi$ in $\dom(\mathfrak{v}_\Lambda^+)$.
By Lem.~\ref{lem-abs-val}, $\|\Psi_n\|_{\sF}\in\sQ$, and
the latter convergence and the inverse triangle inequality for $\|\cdot\|_{\sF}$ imply
that $\|\Psi_n\|_{\sF}\to\|\Psi\|_{\sF}$ in $\fdom(V_+)$. In view of Lem.~\ref{lem-abs-val}
it remains to show that
\begin{align}\label{harald11}
\Re\SPn{\fS_{\Psi_n}}{iv_j^*\Psi_n}_{\sF}\xrightarrow{\;\;n\to\infty\;\;}
\Re\SPn{\fS_{\Psi}}{iv_j^*\Psi}_{\sF}
\quad\text{in $L^2(\Lambda)$},\quad j\in\{1,\ldots,\nu\}.
\end{align}
Passing to suitable subsequences, if necessary, we may assume that the convergences
$\Psi_n\to\Psi$ and $v_j\Psi_n\to{v}_j^*\Psi$, $j\in\{1,\ldots,\nu\}$, 
also take place pointwise a.e. on $\Lambda$ and that
$\|v_j\Psi_n(\V{x})\|\le R_j(\V{x})$, a.e. $\V{x}\in\Lambda$, $n\in\NN$, for some 
$R_1,\ldots,R_\nu\in L^2(\Lambda)$. Then $\fS_{\Psi_n}\to\fS_{\Psi}$ a.e. on
$\{\Psi\not=0\}$. Furthermore, $v_j^*\Psi_n\to0$ a.e. on ${\{\Psi=0\}}$ since
$v_j^*\Psi=0$ a.e. on $\{\Psi=0\}$ by Lem.~\ref{lemvjstarlevel}.
Now \eqref{harald11} follows from the dominated convergence theorem.
\end{proof}

Now let $\diamond\in\{\mathrm{D},\mathrm{N}\}$. On account of Cor.~\ref{cor-dia} we may
plug $f=\|\Psi\|_{\sF}$ into \eqref{KLMN}, for every 
$\Psi\in\dom({\mathfrak{t}}_{\Lambda,\diamond}^{\V{G},\V{A},V_+})$.
Employing our diamagnetic inequality \eqref{dia1} instead of \eqref{dia-w}, we then observe that 
$\dom(\mathfrak{t}_{\Lambda,\diamond}^{\V{G},\V{A},V_+})\subset\dom(\mathfrak{v}_\Lambda^-)$ and
\begin{align}\label{harald1}
\mathfrak{v}_\Lambda^-[\Psi]&\le a\mathfrak{t}_{\Lambda,\diamond}^{\V{G},\V{A},V_+}[\Psi]
+b\|\Psi\|^2,\quad\Psi\in\dom({\mathfrak{t}}_{\Lambda,\diamond}^{\V{G},\V{A},V_+}).
\end{align}
Again we conclude that the form defined by 
$\dom({\mathfrak{t}}_{\Lambda,\diamond}^{\V{G},\V{A},V})
:=\dom({\mathfrak{t}}_{\Lambda,\diamond}^{\V{G},\V{A},V_+})$ and
\begin{align}\label{harald2}
{\mathfrak{t}}_{\Lambda,\diamond}^{\V{G},\V{A},V}[\Psi]
&:={\mathfrak{t}}_{\Lambda,\diamond}^{\V{G},\V{A},V_+}[\Psi]-\mathfrak{v}_\Lambda^-[\Psi],
\quad\Psi\in\dom({\mathfrak{t}}_{\Lambda,\diamond}^{\V{G},\V{A},V}),
\end{align}
is semi-bounded and closed. Later on, we shall also need the following familiar
consequence of \eqref{harald1} and \eqref{harald2},
\begin{align}\label{harald3}
{\mathfrak{t}}_{\Lambda,\diamond}^{\V{G},\V{A},0}[\Psi]
&\le{\mathfrak{t}}_{\Lambda,\diamond}^{\V{G},\V{A},V_+}[\Psi]
\le\frac{1}{1-a}{\mathfrak{t}}_{\Lambda,\diamond}^{\V{G},\V{A},V}[\Psi]+\frac{b}{1-a}\|\Psi\|^2,
\end{align}
for all $\Psi\in\dom({\mathfrak{t}}_{\Lambda,\diamond}^{\V{G},\V{A},V})$.

We denote the unique self-adjoint operator representing 
${\mathfrak{t}}_{\Lambda,\diamond}^{\V{G},\V{A},V}$ by ${T}_{\Lambda,\diamond}^{\V{G},\V{A},V}$.

For later reference we note some quite elementary observations:

\begin{lemma}\label{lemstST}
Let $\diamond$ be $\mathrm{D}$ or $\mathrm{N}$, $\V{A}\in L^2_\loc(\Lambda,\RR^\nu)$, and assume 
that $V_+,V_-\in L^1_\loc(\Lambda)$, $V_\pm\ge0$, satisfy \eqref{KLMN} for some $a\in[0,1)$ 
and $b>0$. Then the following holds:
\begin{enumerate}[leftmargin=0.67cm]
\item[{\rm(1)}] We have the following inclusions, where $\phi\in\sF$,
$$
\dom(\mathfrak{s}_{\Lambda,\diamond}^{\V{A},V})\otimes\sF
\subset\dom(\mathfrak{t}_{\Lambda,\diamond}^{\V{0},\V{A},V}),\quad\big\{\SPn{\phi}{\Psi}_\sF:\Psi\in
\dom(\mathfrak{t}_{\Lambda,\mathrm{N}}^{\V{0},\V{A},V})\big\}\subset
\dom(\mathfrak{s}_{\Lambda,\mathrm{N}}^{\V{A},V}).
$$
\item[{\rm(2)}] Let $N\in\NN$, $f_1,\ldots,f_N\in\dom(\mathfrak{s}_{\Lambda,\mathrm{N}}^{\V{A},V})$, 
and let $e_1,\ldots,e_N$ mutually orthonormal elements of $\sF$. Then
\begin{align}\label{clarissa1}
\Big\|\sum_{\ell=1}^n
f_\ell e_\ell\Big\|_{\mathfrak{t}_{\Lambda,\mathrm{N}}^{\V{0},\V{A},V}}^2
=\sum_{\ell=1}^n\|f_\ell\|_{\mathfrak{s}_{\Lambda,\mathrm{N}}^{\V{A},V}}^2.
\end{align}
\item[{\rm(3)}] $\dom(\mathfrak{s}_{\Lambda,\diamond}^{\V{A},V})\otimes\sF$ is a core of
$\mathfrak{t}_{\Lambda,\diamond}^{\V{0},\V{A},V}$.
\item[{\rm(4)}] $\dom({S}_{\Lambda,\diamond}^{\V{A},V})\otimes\sF\subset
\dom(T_{\Lambda,\diamond}^{\V{0},\V{A},V})$ with
$$
T_{\Lambda,\diamond}^{\V{0},\V{A},V}(f\phi)=({S}_{\Lambda,\diamond}^{\V{A},V}f)\phi,\quad
f\in\dom({S}_{\Lambda,\diamond}^{\V{A},V}),\,\phi\in\sF.
$$
\end{enumerate}
\end{lemma}

\begin{proof}
The second inclusion in (1) and the first one for $\diamond=\mathrm{N}$ follow from 
Rem.~\ref{remwCCw}.

For all $f,g\in\dom(\mathfrak{s}_{\Lambda,\mathrm{N}}^{\V{A},V})$ and $\phi,\psi\in\sF$, 
we further infer from Rem.~\ref{remwCCw} that
\begin{align}\label{clarissa0}
\mathfrak{s}_{\Lambda,\mathrm{N}}^{\V{A},V}[f,g]\SPn{\phi}{\psi}_{\sF}
&=\mathfrak{t}_{\Lambda,\mathrm{N}}^{\V{0},\V{A},V}[f\phi,g\psi],
\end{align}
which entails the formula in Part~(2).

Now, we can prove the first inclusion in Part~(1) in the case $\diamond=\mathrm{D}$. In fact,
every $\Psi\in\dom(\mathfrak{s}_{\Lambda,\mathrm{D}}^{\V{A},V_+})\otimes\sF$ has the form
$\sum_{\ell=1}^N f_\ell e_\ell$ with $f_\ell\in\dom(\mathfrak{s}_{\Lambda,\mathrm{D}}^{\V{A},V_+})$
and a suitable orthonormal basis $\{e_\ell:\ell\in\NN\}$ of $\sF$. We then find sequences
$\{f_{\ell,n}\}_{n\in\NN}$ in $\sD(\Lambda)$ such that $f_{\ell,n}\to f_\ell$, $n\to\infty$, 
with respect to the form norm
of $\mathfrak{s}_{\Lambda,\mathrm{D}}^{\V{A},V_+}$, for all $\ell\in\{1,\ldots,N\}$.
Set $\Psi_n:=\sum_{\ell=1}^Nf_{\ell,n}e_\ell$. In view of \eqref{clarissa1} we then see that
$\{\Psi_n\}_{n\in\NN}$ is a Cauchy sequence with respect to the form norm of 
$\mathfrak{t}_{\Lambda,\mathrm{D}}^{\V{0},\V{A},V_+}$, thus
$\Psi\in\dom(\mathfrak{t}_{\Lambda,\mathrm{D}}^{\V{0},\V{A},V_+})$.

In the case $\diamond=\mathrm{D}$, the assertion (3) holds by definition of the 
minimal forms and the first inclusion in Part~(1). 
To prove (3) for $\diamond=\mathrm{N}$, let $\{e_\ell:\ell\in\NN\}$ be an
orthonormal basis of $\sF$, $p_n:=\sum_{\ell=1}^n|e_\ell\rangle\langle e_\ell|$, $n\in\NN$,
and $\Psi\in\dom(w_j^*)$. Then, on the one hand, the second inclusion in Part~(1) implies 
$p_n\Psi\in\dom(\mathfrak{s}_{\Lambda,\mathrm{N}}^{\V{A},V_+})\otimes\sF$. 
On the other hand, the obvious relations $p_nw_j\Phi=w_jp_n\Phi$, $\Phi\in \sD(\Lambda,\sF)$, 
show that $p_n\Psi\in\dom(w_j^*)$ and 
$w_j^*p_n\Psi=p_nw_j^*\Psi\to w_j^*\Psi$ in $L^2(\Lambda,\sF)$. Since also
$V_+^\eh p_n\Psi\to V_+^\eh\Psi$ in $L^2(\Lambda,\sF)$, this concludes the proof of Part~(3).

Finally, for all $f\in\dom({S}_{\Lambda,\diamond}^{\V{A},V})$, 
$g\in\dom(\mathfrak{s}_{\Lambda,\diamond}^{\V{A},V})$ 
and $\phi,\psi\in\sF$, the formula \eqref{clarissa0} yields
\begin{align*}
\SPn{({S}_{\Lambda,\diamond}^{\V{A},V} f)\phi}{g\psi}
&=\SPn{{S}_{\Lambda,\diamond}^{\V{A},V} f}{g}\SPn{\phi}{\psi}_{\sF}
\\
&=\mathfrak{s}_{\Lambda,\diamond}^{\V{A},V}[f,g]\SPn{\phi}{\psi}_{\sF}
=\mathfrak{t}_{\Lambda,\mathrm{N}}^{\V{0},\V{A},V}[f\phi,g\psi]
=\mathfrak{t}_{\Lambda,\diamond}^{\V{0},\V{A},V}[f\phi,g\psi].
\end{align*}
In the last step we also took Part~(1) into account.
Since $\dom(\mathfrak{s}_{\Lambda,\diamond}^{\V{A},V})\otimes\sF$ is a core for
$\mathfrak{t}_{\Lambda,\diamond}^{\V{0},\V{A},V}$, this computation implies the assertion;
see \cite[Thm.~VI.2.1(iii)]{Kato}.
\end{proof}

Next, we add the radiation field energy to our forms and Hamiltonians. The corresponding form
is given by 
$$
\mathfrak{f}_\Lambda[\Psi]:=\int_\Lambda\|\Id\Gamma(\omega)^\eh\Psi(\V{x})\|_{\sF}^2\Id\V{x},
\quad\Psi\in\dom(\mathfrak{f}_\Lambda):=L^2(\Lambda,\fdom(\Id\Gamma(\omega))).
$$
It is closed and obviously non-negative. The closed form defined by
\begin{align*}
{\mathfrak{q}}_{\Lambda,\mathrm{N}}^{\V{G},\V{A},V}[\Psi]
&:=\mathfrak{t}_{\Lambda,\mathrm{N}}^{\V{G},\V{A},V}[\Psi]+\mathfrak{f}_\Lambda[\Psi],\quad
\Psi\in\dom({\mathfrak{q}}_{\Lambda,\mathrm{N}}^{\V{G},\V{A},V})
:=\dom(\mathfrak{t}_{\Lambda,\mathrm{N}}^{\V{G},\V{A},V_+})\cap\dom(\mathfrak{f}_\Lambda),
\end{align*}
will be called the {\em maximal Pauli-Fierz form} and
\begin{align}\label{defqD}
{\mathfrak{q}}_{\Lambda,\mathrm{D}}^{\V{G},\V{A},V}
&:=\ol{{\mathfrak{q}}_{\Lambda,\mathrm{N}}^{\V{G},\V{A},V}\!\!\upharpoonright_{
\sD(\Lambda,\fdom(\Id\Gamma(1\vee\omega)))}},
\end{align}
the {\em minimal Pauli-Fierz form}. For $\diamond\in\{\mathrm{D},\mathrm{N}\}$, 
the corresponding {\em Pauli-Fierz operator} $H_{\Lambda,\diamond}^{\V{G},\V{A},V}$ is defined as
the unique self-adjoint operator representing ${\mathfrak{q}}_{\Lambda,\diamond}^{\V{G},\V{A},V}$.

\begin{remark}
In analogy to the aforementioned result of \cite{SimonJOT1979}, a series of 
approximation arguments reveals that
${\mathfrak{t}}_{\RR^\nu,\mathrm{D}}^{\V{G},\V{A},V_+}=
{\mathfrak{t}}_{\RR^\nu,\mathrm{N}}^{\V{G},\V{A},V_+}$ and
${\mathfrak{q}}_{\RR^\nu,\mathrm{D}}^{\V{G},\V{A},V_+}=
{\mathfrak{q}}_{\RR^\nu,\mathrm{N}}^{\V{G},\V{A},V_+}$
under the mere assumption that $\V{A}\in L^2_\loc(\RR^\nu,\RR^\nu)$, 
$\V{G}\in L^2_\loc(\RR^\nu,\fdom(\omega^{-1})^\nu)$, and $V_+\ge0$ is locally integrable.
As we do not use this result we refrain from presenting its space-consuming proof.
See, however, Cor.~\ref{corminmax} for a special case.
\end{remark}

As a direct consequence of the definitions,
$\dom({\mathfrak{q}}_{\Lambda,\mathrm{N}}^{\V{G},\V{A},V})
=\dom({\mathfrak{q}}_{\Lambda,\mathrm{N}}^{\V{G},\V{A},V_+})$ and
\begin{align}\label{zerlqN}
{\mathfrak{q}}_{\Lambda,\mathrm{N}}^{\V{G},\V{A},V}
&=\mathfrak{t}_{\Lambda,\mathrm{N}}^{\V{G},\V{A},V_+}-\mathfrak{v}_\Lambda^-
+\mathfrak{f}_\Lambda={\mathfrak{q}}_{\Lambda,\mathrm{N}}^{\V{G},\V{A},V_+}
-\mathfrak{v}_\Lambda^-.
\end{align}
In view of \eqref{harald3} we further observe that the form norms associated with
${\mathfrak{q}}_{\Lambda,\mathrm{N}}^{\V{G},\V{A},V}$ and 
${\mathfrak{q}}_{\Lambda,\mathrm{N}}^{\V{G},\V{A},V_+}$ are equivalent.
Furthermore, ${\mathfrak{t}}_{\Lambda,\mathrm{D}}^{\V{G},\V{A},V}+\mathfrak{f}_\Lambda$ 
is closed as a sum of two semi-bounded 
closed forms and its domain contains $\sD(\Lambda,\fdom(\Id\Gamma(1\vee\omega)))$, whence
\begin{align}\label{egon99}
{\mathfrak{q}}_{\Lambda,\mathrm{D}}^{\V{G},\V{A},V}
\subset{\mathfrak{t}}_{\Lambda,\mathrm{D}}^{\V{G},\V{A},V}+\mathfrak{f}_\Lambda.
\end{align}

\begin{lemma}\label{lem-egon}
Let $\V{A}\in L^2_\loc(\Lambda,\RR^\nu)$, 
$\V{G}\in L^\infty(\Lambda,\fdom(\omega^{-1}+\omega)^\nu)$,
and let $V_\pm\in L^1_\loc(\RR^\nu)$, $V_\pm\ge0$,
satisfy \eqref{KLMN} with $\diamond=\mathrm{D}$ for some $a\in[0,1)$ and $b<\infty$.  
Then ${\mathfrak{q}}_{\Lambda,\mathrm{D}}^{\V{G},\V{A},V}
={\mathfrak{t}}_{\Lambda,\mathrm{D}}^{\V{G},\V{A},V}+\mathfrak{f}_\Lambda
={\mathfrak{q}}_{\Lambda,\mathrm{D}}^{\V{G},\V{A},V_+}-\mathfrak{v}_\Lambda^-$.
\end{lemma}

\begin{proof}
We drop the superscripts $\V{A}$, $\V{G}$, and the subscript $\Lambda$ in this
proof as all these quantities are kept fixed.

To prove the inclusion converse to \eqref{egon99}, let 
$\Psi\in\dom({\mathfrak{t}}_{\mathrm{D}}^V)\cap\dom(\mathfrak{f})$.
Then there exist $\Psi_n\in\sD(\Lambda,\fdom(\Id\Gamma(1)))$, $n\in\NN$, such that
$\Psi_n\to\Psi$, $V_+^\eh\Psi_n\to V_+^\eh\Psi$, and $v_j\Psi_n\to v_j^*\Psi$, 
$j\in\{1,\ldots,\nu\}$, in $L^2(\Lambda,\sF)$. Set $\theta_\ve:=1+\ve\Id\Gamma(\omega)$,
$\Psi_\ve:=\theta_\ve^\mh\Psi$, and 
$\Psi_{n,\ve}:=\theta_\ve^\mh\Psi_n\in\sD(\Lambda,\fdom(\Id\Gamma(1\vee\omega)))$, 
$n\in\NN$, $\ve>0$.
Since, according to Lem.~\ref{lem-vstar}, $\theta_\ve^\mh$ maps the graph of every $v_j^*$
continuously into itself, and since $\theta_\ve^\mh$ commutes with $V_+^\eh$, it follows for
every $\ve>0$ that $\Psi_{n,\ve}\to\Psi_\ve$, $V_+^\eh\Psi_{n,\ve}\to V_+^\eh\Psi_\ve$, and 
$v_j\Psi_{n,\ve}\to v_j^*\Psi_\ve$, i.e.,
$\|\Psi_{n,\ve}-\Psi_\ve\|_{\mathfrak{t}_{\mathrm{N}}^{V_+}}\to0$, as $n\to\infty$.
Since $\Id\Gamma(\omega)^\eh\theta_\ve^\mh$ is bounded,
we also have $\|\Psi_{n,\ve}-\Psi_\ve\|_{\mathfrak{f}}\to0$. Altogether this
shows that $\Psi_\ve\in\dom({\mathfrak{q}}_{\mathrm{D}}^V)$, $\ve>0$.
Lem.~\ref{lem-vstar} also implies, however, that $v_j^*\Psi_\ve\to v_j^*\Psi$, as
$\ve\downarrow0$. Since $\Psi\in\fdom(V_+\id_{\sF})\cap\dom(\mathfrak{f})$, the
dominated convergence theorem further shows that $V_+^\eh\Psi_\ve\to V_+^\eh\Psi$
in $L^2(\Lambda,\sF)$ and $\mathfrak{f}[\Psi_\ve-\Psi]\to0$.
Hence, $\|\Psi_{\ve}-\Psi\|_{\mathfrak{q}^{V_+}_{\mathrm{N}}}\to0$,
$\ve\downarrow0$, thus $\Psi\in\dom({\mathfrak{q}}^V_{\mathrm{D}})$.
\end{proof}

\begin{lemma}\label{lem-ida}
Let $\V{A}\in L_\loc^2(\Lambda,\RR^\nu)$, 
$\V{G}\in L^\infty(\Lambda,\fdom(\omega^{-1})^\nu)$,
$V_+\in L^1_\loc(\Lambda)$, $V_+\ge0$, and $\diamond\in\{\mathrm{D},\mathrm{N}\}$. 
Then the domain of the form
${\mathfrak{q}}_{\Lambda,\diamond}^{\V{G},\V{A},V_+}$ is equal to the domain of
${\mathfrak{q}}_{\Lambda,\diamond}^{\V{0},\V{A},V_+}$
and the form norm $\|\cdot\|_{{\mathfrak{q}}_{\Lambda,\diamond}^{\V{G},\V{A},V_+}}$ is 
equivalent to $\|\cdot\|_{{\mathfrak{q}}_{\Lambda,\diamond}^{\V{0},\V{A},V_+}}$.
\end{lemma}

\begin{proof}
We drop all sub/superscripts $\Lambda$, $\V{A}$, or $V_+$ in this proof. 

{\em Step~1.} First, we consider the domains of the maximal forms.
For every $j\in\{1,\ldots,\nu\}$, Rem.~\ref{remvj}(1) implies that
$\dom(\mathfrak{f})\cap\dom(v_j^*)=\dom(\mathfrak{f})\cap\dom(w_j^*)$ and that
any vector $\Psi$ in the latter domain satisfies $v_j^*\Psi=w_j^*\Psi-\vp({G_j})\Psi$. In particular,
$\dom(\mathfrak{q}_{\mathrm{N}}^{\V{0}})=\dom(\mathfrak{q}_{\mathrm{N}}^{\V{G}})$.
 
{\em Step~2.} 
Next, we prove the asserted equivalence of norms for the maximal forms. By Step~1,
the identity $\mathfrak{t}_{\mathrm{N}}^{\V{G}}=\mathfrak{t}_{\mathrm{N}}^{\V{0}}
+\mathfrak{b}^{\V{G}}+\mathfrak{c}^{\V{G}}$ holds on
$\dom(\mathfrak{q}_{\mathrm{N}}^{\V{0}})$, where
$$
\mathfrak{b}^{\V{G}}[\Psi]:=\frac{1}{2}\sum_{j=1}^\nu\|\vp({G_j})\Psi\|^2,\quad
\mathfrak{c}^{\V{G}}[\Psi]:=-\sum_{j=1}^\nu\Re\SPn{w_j^*\Psi}{\vp({G_j})\Psi},
\quad\Psi\in\dom(\mathfrak{q}_{\mathrm{N}}^{\V{0}}).
$$
On account of \eqref{rbvp},
$$
|\mathfrak{c}^{\V{G}}|\le\frac{1}{2}\mathfrak{t}_\diamond^{\V{0}}
+2\mathfrak{b}^{\V{G}}\le\frac{1}{2}\mathfrak{t}_{\mathrm{N}}^{\V{0}}+\rho\mathfrak{f}
+\rho\|\cdot\|^2_{L^2(\Lambda,\sF)}
\quad\text{on $\dom(\mathfrak{q}_{\mathrm{N}}^{\V{0}})$},
$$
where $\rho>0$ is chosen such that 
$\rho\ge4\|(\omega^\mh\vee1)\V{G}\|_{L^\infty(\Lambda,\HP^\nu)}^2$. Further assuming
$\rho\ge1/2$, we finally deduce that
\begin{align}\label{ida1a}
\mathfrak{q}_{\mathrm{N}}^{\V{G}}&=\mathfrak{t}_{\mathrm{N}}^{\V{0}}+\mathfrak{f}
+\mathfrak{b}^{\V{G}}+\mathfrak{c}^{\V{G}}\le
\Big(1+\frac{3\rho}{2}\Big)\mathfrak{q}_{\mathrm{N}}^{\V{0}}+
\frac{3\rho}{2}\|\cdot\|^2_{L^2(\Lambda,\sF)}
\quad\text{on $\dom(\mathfrak{q}_{\mathrm{N}}^{\V{0}})$,}
\end{align}
and
\begin{align}\nonumber
\mathfrak{q}_{\mathrm{N}}^{\V{G}}&\ge\frac{1}{2\rho}\mathfrak{t}_{\mathrm{N}}^{\V{G}}+\mathfrak{f}
\ge\frac{1}{2\rho}\big(\mathfrak{t}_{\mathrm{N}}^{\V{0}}+\mathfrak{c}^{\V{G}}\big)
+\mathfrak{f}\ge\frac{1}{4\rho}\mathfrak{t}_{\mathrm{N}}^{\V{0}}+\frac{1}{2}\mathfrak{f}
-\frac{1}{2}\|\cdot\|^2_{L^2(\Lambda,\sF)}
\\\label{ida1b}
&\ge\frac{1}{4\rho}\mathfrak{q}_{\mathrm{N}}^{\V{0}}-\frac{1}{2}\|\cdot\|^2_{L^2(\Lambda,\sF)}
\quad\text{on $\dom(\mathfrak{q}_{\mathrm{N}}^{\V{0}})$.}
\end{align}

{\em Step~3.} According to Step~2, the closure of $\sD(\Lambda,\fdom(\Id\Gamma(1\vee\omega)))$
with respect to $\|\cdot\|_{\mathfrak{q}_{\mathrm{N}}^{\V{G}}}$ is the same as its closure with
respect to $\|\cdot\|_{\mathfrak{q}_{\mathrm{N}}^{\V{0}}}$. By definition of the minimal forms
this means that 
$\dom(\mathfrak{q}_{\mathrm{D}}^{\V{0}})=\dom(\mathfrak{q}_{\mathrm{D}}^{\V{G}})$.
In particular, the latter equal domains are subsets of $\dom(\mathfrak{q}_{\mathrm{N}}^{\V{0}})$,
whence the equivalence of $\|\cdot\|_{\mathfrak{q}_{\mathrm{D}}^{\V{0}}}$ and 
$\|\cdot\|_{\mathfrak{q}_{\mathrm{D}}^{\V{G}}}$ follows from \eqref{ida1a} and \eqref{ida1b}.
\end{proof}

\begin{corollary}\label{corformcore}
Under the assumptions of Lem.~\ref{lem-ida}, let
$\diamond\in\{\mathrm{D},\mathrm{N}\}$ and $V_-\in L^1_\loc(\Lambda)$, $V_-\ge0$,
satisfy \eqref{KLMN} for some $a\in[0,1)$ and $b<\infty$. Let $\sC$ be a core for the form
$\mathfrak{s}_{\Lambda,\diamond}^{\V{A},V_+}$ and $\sD$ a form core for $\Id\Gamma(\omega)$.
Then $\sC\otimes\sD$ is a core for $\mathfrak{q}_{\Lambda,\diamond}^{\V{G},\V{A},V}$.
\end{corollary}

\begin{proof}
Combining Lem.~\ref{lemstST}(1) with \eqref{zerlqN} and Lem.~\ref{lem-egon}
(applied to $\V{G}=\V{0}$) we observe that
$\dom(\mathfrak{s}_{\Lambda,\diamond}^{\V{A},V})\otimes\fdom(\Id\Gamma(\omega))
\subset\dom(\mathfrak{q}_{\Lambda,\diamond}^{\V{0},\V{A},V})$.
On account of Lem.~\ref{lem-ida} this ensures for a start that 
$\sC\otimes\sD\subset\dom(\mathfrak{s}_{\Lambda,\diamond}^{\V{A},V})\otimes\fdom(\Id
\Gamma(\omega))\subset\dom(\mathfrak{q}_{\Lambda,\diamond}^{\V{G},\V{A},V})$.

Again by Lem.~\ref{lem-ida}, it remains to show that $\sC\otimes\sD$ is a core for
$\mathfrak{q}_{\Lambda,\diamond}^{\V{0},\V{A},V_+}$. 
So let $\Psi\in\dom(\mathfrak{q}_{\Lambda,\diamond}^{\V{0},\V{A},V_+})$.
In view of Lem.~\ref{lemstST}(2)\&(3) we then find $\Psi_n\in\sC\otimes\sF$, $n\in\NN$,
such that $\Psi_n\to\Psi$ in $\dom(\mathfrak{t}_{\Lambda,\diamond}^{\V{0},\V{A},V_+})$.
For $\ve>0$, let $\theta_\ve:=1+\ve\Id\Gamma(\omega)$, $\Psi_\ve:=\theta_\ve^\mh$, and
$\Psi_{n,\ve}:=\theta_\ve^\mh\Psi_n$. As in the proof of Lem.~\ref{lem-egon} it then follows
that $\Psi_{n,\ve}\to\Psi_\ve$, $V_+^\eh\Psi_{n,\ve}\to V_+^\eh\Psi_{\ve}$,
$\mathfrak{f}_\Lambda[\Psi_{n,\ve}-\Psi_\ve]\to0$, and
$w_j^*\Psi_{n,\ve}\to w_j^*\Psi_\ve$, $j\in\{1,\ldots,\nu\}$, as $n$ goes to infinity.
Moreover, $V_+^\eh\Psi_{\ve}\to V_+^\eh\Psi_{}$ and
$\mathfrak{f}_\Lambda[\Psi_{\ve}-\Psi]\to0$, as $\ve\downarrow0$, by dominated convergence,
and $w_j^*\Psi_\ve\to w_j^*\Psi$ by Lem.~\ref{lem-vstar} in the trivial case $\V{G}=\V{0}$.
For fixed $n\in\NN$ and $\ve>0$, it follows, however, from Rem.~\ref{remwCCw} that
every $\Psi_{n,\ve}\in\sC\otimes\fdom(\Id\Gamma(\omega))$ can be approximated
by elements in $\sC\otimes\sD$ with respect to the form norm of 
$\mathfrak{q}_{\Lambda,\mathrm{N}}^{\V{0},\V{A},V_+}$.
\end{proof}

Before stating the next corollay we recall that 
$S_{\RR^\nu,\mathrm{D}}^{\V{0},V_+}=S_{\RR^\nu,\mathrm{N}}^{\V{0},V_+}
=:S_{\RR^\nu}^{\V{0},V_+}$.

\begin{corollary}\label{corminmax}
Consider the special case $\Lambda=\RR^\nu$ with $\V{A}\in L^2_\loc(\RR^\nu,\RR^\nu)$ and
$\V{G}\in L^\infty(\RR^\nu,\fdom(\omega^{-1})^\nu)$. Let $V_\pm\in L^1_\loc(\RR^\nu)$,
$V_\pm\ge0$, such that $V_-$ is relatively form bounded with respect to $S_{\RR^\nu}^{\V{0},V_+}$
with relative form bound $<1$. Then $\mathfrak{q}_{\RR^\nu,\mathrm{D}}^{\V{G},\V{A},V}=
\mathfrak{q}_{\RR^\nu,\mathrm{N}}^{\V{G},\V{A},V}$.
\end{corollary}

\begin{proof}
The result
$\mathfrak{s}_{\RR^\nu,\mathrm{D}}^{\V{A},V}=\mathfrak{s}_{\RR^\nu,\mathrm{N}}^{\V{A},V}$
of \cite{SimonJOT1979} and Cor.~\ref{corformcore} imply that
$\dom(\mathfrak{s}_{\RR^\nu,\mathrm{D}}^{\V{A},V})\otimes\fdom(\Id\Gamma(\omega))$
is a common form core of $\mathfrak{q}_{\RR^\nu,\mathrm{D}}^{\V{G},\V{A},V}$
and $\mathfrak{q}_{\RR^\nu,\mathrm{N}}^{\V{G},\V{A},V}$. 
In view of \eqref{defqD} this implies the assertion.
\end{proof}

%%%%%%%%%%%%%%%%%%%%%%%%%%%%%%%%%%%%%%%%%%
%%%%%%%%%%%%%%%%%%%%%%%%%%%%%%%%%%%%%%%%%%
%%%%%%%%%%%%%%%%%%%%%%%%%%%%%%%%%%%%%%%%%%

\section{Domain and cores of the Dirchlet-Pauli-Fierz operator}\label{sec-dom}

\noindent
In this section we prove the main result of this paper, Thm.~\ref{thm-dom} below, 
asserting that the domain of the Dirchlet-Pauli-Fierz 
operator $H_{\Lambda,\mathrm{D}}^{\V{G},\V{A},V}$ is equal to the intersection of the domain
of the vector-valued Dirichlet-Schr\"odinger operator ${T}_{\Lambda,\mathrm{D}}^{\V{0},\V{A},V}$ 
with the domain $L^2(\Lambda,\dom(\Id\Gamma(\omega)))$ of the radiation field energy.
In combination with Lem.~\ref{lem-free-ham}(2), the theorem further shows that the determination 
of the operator cores of the Dirichlet-Pauli-Fierz operator
of the type \eqref{def-CtensorD} boils down to the determination of the operator cores of the 
scalar Schr\"{o}dinger operator ${S}_{\Lambda,\mathrm{D}}^{\V{A},V}$ and $\Id\Gamma(\omega)$. 
In Rem.~\ref{remZeeman} we extend these results to the case where classical and quantized Zeeman
terms are added to $H_{\Lambda,\mathrm{D}}^{\V{G},\V{A},V}$.

Throughout this section we shall again use the notation
\begin{align*}
\theta=1+\Id\Gamma(\omega).
\end{align*}
We start with some elementary remarks on the operator
$H_{\Lambda,\mathrm{D}}^{\V{0},\V{A},V}$, 
where the interaction between the matter particles and the radiation field is turned off:

\begin{lemma}\label{lem-free-ham}
Let $\V{A}\in L^2_\loc(\Lambda,\RR^\nu)$, 
$\diamond\in\{\mathrm{D},\mathrm{N}\}$, and assume that $V_\pm\ge0$ satisfy \eqref{KLMN} 
for some $a\in[0,1)$ and $b\ge0$.
Assume further that $\mathfrak{s}_{\Lambda,\diamond}^{\V{A},V}\ge0$, which can always be
achieved by adding a suitable non-negative constant to $V_+$. Then the following holds:
\begin{enumerate}[leftmargin=0.67cm]
\item[{\rm(1)}] $H_{\Lambda,\diamond}^{\V{0},\V{A},V}=T_{\Lambda,\diamond}^{\V{0},\V{A},V}
+\Id\Gamma(\omega)$ which in particular includes the equality of domains
$\dom(H_{\Lambda,\diamond}^{\V{0},\V{A},V})=\dom(T_{\Lambda,\diamond}^{\V{0},\V{A},V})\cap
L^2(\Lambda,\dom(\Id\Gamma(\omega)))$. Furthermore,
\begin{align*}
\big(\|T_{\Lambda,\diamond}^{\V{0},\V{A},V}\Psi\|^2
+\|\Id\Gamma(\omega)\Psi\|^2\big)^\eh&\le\|H_{\Lambda,\diamond}^{\V{0},\V{A},V}\Psi\|
+\|\Psi\|,\quad\Psi\in\dom(H_{\Lambda,\diamond}^{\V{0},\V{A},V}).
\end{align*}
\item[{\rm(2)}] Let $\sC\subset L^2(\Lambda)$ be a core for 
${S}_{\Lambda,\diamond}^{\V{A},V}$ and $\sD\subset\sF$ be a core for 
$\Id\Gamma(\omega)$. Then $\sC\otimes\sD$ is a core for $H_{\Lambda,\diamond}^{\V{0},\V{A},V}$.
\item[{\rm(3)}] $\Psi\in\dom(H_{\Lambda,\diamond}^{\V{0},\V{A},V})$ implies
$w_j^*\Psi\in\dom(\theta^\eh)$, $\theta^\eh\Psi\in\dom(w_j^*)$, and
$\theta^\eh w_j^*\Psi=w_j^*\theta^\eh\Psi$, for every $j\in\{1,\ldots,\nu\}$.
Moreover,
\begin{align}\label{bakithi1b}
\frac{1}{2}\sum_{j=1}^\nu\|\theta^\eh w_j^*\Psi\|^2
\le\frac{1+b}{1-a}\|(T^{\V{0},\V{A},V}_{\Lambda,\diamond}+1)\Psi\|\|\theta\Psi\|,
\quad\Psi\in\dom(H_{\Lambda,\diamond}^{\V{0},\V{A},V}).
\end{align}
\end{enumerate}
\end{lemma}

\begin{proof}
We shall drop all sub/superscripts $\Lambda$, $\V{A}$, and $V$ in this proof, so that for instance
$\mathfrak{t}_{\diamond}^{\V{0}}=\mathfrak{t}_{\Lambda,\diamond}^{\V{0},\V{A},V}$ and 
$\mathfrak{q}_{\diamond}^{\V{0}}=\mathfrak{q}_{\Lambda,\diamond}^{\V{0},\V{A},V}$.
First, we prove Parts~(1) and (2) simultaneously.

Thanks to Lem.~\ref{lemstST}(4) we know that 
$\dom({S}_\diamond)\otimes\sF\subset\dom(T_\diamond^{\V{0}})$ with
$T_\diamond^{\V{0}}(f\phi)=({S}_\diamond f)\phi$, $f\in\dom({S}_\diamond)$, $\phi\in\sF$.
On account of $\dom(T_\diamond^{\V{0}})\cap L^2(\Lambda,\dom(\Id\Gamma(\omega)))
\subset\dom(\mathfrak{t}_{\diamond}^{\V{0}})\cap L^2(\Lambda,\fdom(\Id\Gamma(\omega)))
=\dom(\mathfrak{q}_\diamond^{\V{0}})$,
we further have $\mathfrak{q}_\diamond^{\V{0}}[\Phi,\Psi]=\mathfrak{t}_\diamond^{\V{0}}[\Phi,\Psi]
+\mathfrak{f}[\Phi,\Psi]=\SPn{\Phi}{T_\diamond^{\V{0}}\Psi}+
\SPn{\Phi}{\Id\Gamma(\omega)\Psi}$, $\Phi\in\dom(\mathfrak{q}_\diamond^{\V{0}})$,
$\Psi\in\dom(T_\diamond^{\V{0}})\cap L^2(\Lambda,\dom(\Id\Gamma(\omega)))$, which shows that
$T_\diamond^{\V{0}}+\Id\Gamma(\omega)\subset H_\diamond^{\V{0}}$ by the first representation 
theorem for quadratic forms; see \cite[Thm.~VI.2.1(iii)]{Kato}.

Let $U:L^2(\Lambda,\sF)\to L^2(\Lambda)\hat{\otimes}\sF$ be the canonical unitary transform
onto the completed tensor product of Hilbert spaces that maps $f\phi$ to $f\otimes \phi$, 
for $f\in L^2(\Lambda)$ and $\phi\in\sF$. Then in view of the above remarks, 
$UH_\diamond^{\V{0}}\!\!\upharpoonright_{\sC\otimes\sD}U^*=
({S}_\diamond\!\!\upharpoonright_{\sC})\otimes\id
+\id\otimes(\Id\Gamma(\omega)\!\!\upharpoonright_{\sD})$, which is known to be
essentially self-adjoint; see, e.g., \cite[Thm.~8.33]{Weidmann}.
Hence, $H_\diamond^{\V{0}}\!\!\upharpoonright_{\sC\otimes\sD}$ is essentially self-adjoint as well.

Finally, we consider
$\Psi\in\dom({S_\diamond})\otimes\dom(\Id\Gamma(\omega))\subset\dom(H_\diamond^{\V{0}})$.
After an application of the Gram-Schmid orthogonalization scheme,
we may write $\Psi$ in the form $\Psi=\sum_{\ell=1}^{N}f_{\ell}\phi_{\ell}$,
where $\phi_{1},\ldots,\phi_{N}$ are mutually orthonormal with respect to the scalar product
$[\phi,\psi]:=\SPn{\theta^\eh\phi}{\theta^\eh\psi}_{\sF}$. By our earlier remarks,
$H_\diamond^{\V{0}}\!\!\upharpoonright_{\dom({S_\diamond})\otimes\dom(\Id\Gamma(\omega))}
\subset T_\diamond^{\V{0}}+\Id\Gamma(\omega)$, 
which permits to get
\begin{align*}
\|(H_\diamond^{\V{0}}+1)\Psi\|^2&=\big\|T_\diamond^{\V{0}}\Psi+\theta\Psi\big\|^2
\\
&=\|T_\diamond^{\V{0}}\Psi\|^2+\|\theta\Psi\|^2
+2\sum_{j,\ell=1}^{N}\Re\big\{\SPn{f_{j}}{{S}_\diamond f_{\ell}}[\phi_{j},\phi_{\ell}]\big\}
\\
&=\|T_\diamond^{\V{0}}\Psi\|^2+\|\theta\Psi\|^2+2\sum_{\ell=1}^{N}\mathfrak{s}_\diamond[f_{\ell}]
\ge\|T_\diamond^{\V{0}}\Psi\|^2+\|\Id\Gamma(\omega)\Psi\|^2.
\end{align*}
Here we used the assumption that $\mathfrak{s}_\diamond$ is non-negative in the last step.
Since $H_\diamond^{\V{0}}$ is essentially self-adjoint on 
$\dom({S_\diamond})\otimes\dom(\Id\Gamma(\omega))$ and $T_\diamond^{\V{0}}$ and 
$\Id\Gamma(\omega)$ are closed, this bound entails 
$H_\diamond^{\V{0}}\subset T_\diamond^{\V{0}}+\Id\Gamma(\omega)$.

To prove Part~(3) we again abbreviate $\theta_\ve:=1+\ve\Id\Gamma(\omega)$ 
and put $\Upsilon_\ve:=\theta^\eh\theta_\ve^\mh$, $\ve>0$.
Obviously, $\Upsilon_\ve$ maps $\dom(w_j)=\sD(\Lambda,\sF)$ into itself 
and $w_j\Upsilon_\ve=\Upsilon_\ve w_j$. Since $\Upsilon_\ve$ is bounded, self-adjoint, 
and continuously invertible, this implies $\Upsilon_\ve w_j^*=w_j^*\Upsilon_\ve$. 
For all $\Psi\in\dom(H_\diamond^{\V{0}})\subset\dom(T_\diamond^{\V{0}})
\subset\fdom(V_+)\cap\bigcap_{j=1}^\nu\dom(w_j^*)$,
we infer from these remarks and \eqref{harald3} that
\begin{align*}
\frac{1}{2}\sum_{j=1}^\nu\|\Upsilon_\ve w_j^*\Psi\|^2
&\le\frac{1}{2}\sum_{j=1}^\nu\|w_j^*\Upsilon_\ve\Psi\|^2+\big\|V_+^\eh\Upsilon_\ve\Psi\big\|^2
\le c\,\mathfrak{t}^{\V{0}}_{\diamond}[\Upsilon_\ve\Psi]+c\|\Upsilon_\ve\Psi\|^2
%\\
%&=\frac{c}{2}\sum_{j=1}^\nu\SPb{w_j^*\Psi}{w_j^*\Upsilon_\ve^2\Psi}+c
%\int_\Lambda\big(1+V(\V{x})\big)\SPn{\Psi(\V{x})}{\Upsilon_\ve^2\Psi(\V{x})}\Id\V{x}
\\
&=c\,\mathfrak{t}^{\V{0}}_{\diamond}[\Psi,\Upsilon_\ve^2\Psi]+c\SPn{\Psi}{\Upsilon_\ve^2\Psi}
=c\SPn{(T_\diamond^{\V{0}}+1)\Psi}{\Upsilon_\ve^2\Psi},
\end{align*}
with $c:=(1+b)/(1-a)$. Here we further have $\Upsilon_\ve^2\Psi\to\theta\Psi$, $\ve\downarrow0$,
since $\dom(H_\diamond^{\V{0}})\subset L^2(\Lambda,\dom(\Id\Gamma(\omega)))$.
The monotone convergence theorem now implies that $w_j^*\Psi\in\dom(\theta^\eh)$ and
\eqref{bakithi1b} is valid. Since $w_j^*$ is closed and we know that 
$w_{j}^*\Upsilon_\ve\Psi=\Upsilon_\ve w_j^*\Psi$ converges to $\theta^\eh w_j^*\Psi$,
as $\ve\downarrow0$, we further see that $\theta^\eh\Psi\in\dom(w_j^*)$ with
$w_j^*\theta^\eh\Psi=\theta^\eh w_j^*\Psi$.
\end{proof}

In what follows we shall work with several choices of the dispersion relation at the same time.
Hence we make the dependence of the Pauli-Fierz operators on the dispersion relation
explicit in the notation as well, although this leads to an ugly cluttering of sub and superscripts. 
More precisely, we shall consider the family of dispersion
relations $\alpha\omega+m$ with $\alpha\ge1$ and $m\ge0$ and write
$\mathfrak{q}_{\Lambda,\diamond,\alpha\omega+m}^{\V{G},\V{A},V}$ and
$H_{\Lambda,\diamond,\alpha\omega+m}^{\V{G},\V{A},V}$ for the Pauli-Fierz form and operator 
obtained upon replacing $\omega$ by $\alpha\omega+m$ in the construction of 
$\mathfrak{q}_{\Lambda,\diamond}^{\V{G},\V{A},V}
=:\mathfrak{q}_{\Lambda,\diamond,\omega}^{\V{G},\V{A},V}$
and $H_{\Lambda,\diamond}^{\V{G},\V{A},V}=:H_{\Lambda,\diamond,\omega}^{\V{G},\V{A},V}$, 
respectively. 

Furthermore, we shall, besides \eqref{shorty0}, use the following short-hands,
\begin{align*}%\label{shorty}
\vp(\V{G})\cdot\ol{\V{w}}\Psi:=\sum_{j=1}^\nu\vp(G_j)\ol{w}_j\Psi,\quad
\vp(\V{G})^2\Psi:=\sum_{j=1}^\nu\vp(G_j)^2\Psi.
\end{align*}

\begin{proposition}\label{proprep}
Let $\V{A}\in L^2_\loc(\Lambda,\RR^\nu)$ and suppose that 
$V_\pm\in L^1_\loc(\Lambda,\RR)$, $V_\pm\ge0$, 
satisfy \eqref{KLMN} with $\diamond=\mathrm{D}$ for some $a\in[0,1)$ and $b>0$. 
Then the following holds:
\begin{enumerate}[leftmargin=0.67cm]
\item[{\rm(1)}]
Assume that $\V{G}\in L^\infty(\Lambda,\fdom(\omega^{-1})^\nu)$
has a weak divergence denoted as 
$q:=\mathrm{div}\V{G}\in L^\infty(\Lambda,\fdom(\omega^{-1}))$ and
let $\Psi\in\dom(H_{\Lambda,\mathrm{D}}^{\V{0},\V{A},V})$.
Then $\Psi\in\dom(H_{\Lambda,\mathrm{D}}^{\V{G},\V{A},V})$ and
\begin{align}\label{bakithi1}
H_{\Lambda,\mathrm{D}}^{\V{G},\V{A},V}\Psi&=
H_{\Lambda,\mathrm{D}}^{\V{0},\V{A},V}\Psi-\vp(\V{G})\cdot\ol{\V{w}}\Psi
+\frac{1}{2}\vp(\V{G})^2\Psi+\frac{i}{2}\vp(q)\Psi.
\end{align}
\item[{\rm(2)}]
Assume that $\V{G}\in L^\infty(\Lambda,\HP^\nu)$
has a weak divergence $q\in L^\infty(\Lambda,\HP)$ and let $m>0$ and
$\Psi\in\dom(H_{\Lambda,\mathrm{D},\omega+m}^{\V{0},\V{A},V})$. Then again
$\Psi\in\dom(H_{\Lambda,\mathrm{D}}^{\V{G},\V{A},V})$ and \eqref{bakithi1} is satisfied.
\end{enumerate}
\end{proposition}

\begin{proof}
We shall assume without loss of generality that a suitable non-negative constant has been 
added to $V_+$ such that $\mathfrak{s}_{\Lambda,\mathrm{D}}^{\V{A},V}\ge0$. 
From now on we will also drop all sub/superscripts $\Lambda$, $\V{A}$, and $V$ in this proof.

Let $m\ge0$ and $\Psi\in\dom(H_{\mathrm{D},\omega+m}^{\V{0}})$ where 
$H_{\mathrm{D},\omega}^{\V{0}}=H_{\mathrm{D}}^{\V{0}}$. 
If $m=0$, then we assume that the conditions
under (1) are fulfilled. Otherwise we work with the hypotheses of (2).
Firstly, we observe that, in view of \eqref{rbvp}, \eqref{rbvp2}, and Lem.~\ref{lem-free-ham},
the terms $\vp({G}_j){w}_j^*\Psi$, $\vp({G}_j)^2\Psi$, and $\vp(q)\Psi$ are well-defined 
elements of $L^2(\Lambda,\sF)$ in both cases.
Lem.~\ref{lem-free-ham}(3) further ensures that Lem.~\ref{lemmagnLeibniz} applies to $\Psi$,
if we choose $\omega+m$ as dispersion relation in the latter lemma.

Under the conditions of Part~(1), i.e., for $m=0$, Lem.~\ref{lem-ida} implies
$\Psi\in\dom(\mathfrak{q}_{\mathrm{D}}^{\V{0}})=\dom(\mathfrak{q}_{\mathrm{D}}^{\V{G}})$.
If $m>0$, then $\dom(H_{\mathrm{D},\omega+m}^{\V{0}})\subset\dom(H_{\mathrm{D}}^{\V{0}})$ by
Lem.~\ref{lem-free-ham}, so that $\Psi\in\dom(H_{\mathrm{D}}^{\V{0}})
\subset\dom(\mathfrak{q}_{\mathrm{D}}^{\V{0}})$. Then we can further 
apply Lem.~\ref{lem-ida} to the dispersion relation $\omega+m$, which ensures that 
$\Psi\in\dom(\mathfrak{q}_{\mathrm{D},\omega+m}^{\V{0}})
=\dom(\mathfrak{q}_{\mathrm{D},\omega+m}^{\V{G}})$ under 
the conditions of Part~(2). Moreover, if $\Psi_n\in\sD(\Lambda,\fdom(\Id\Gamma(1\vee\omega)))$, 
$n\in\NN$, converge to $\Psi$ with respect to the form norm of 
$\mathfrak{q}_{\mathrm{D},\omega+m}^{\V{G}}$, then they form a Cauchy sequence 
with respect to the form norm of $\mathfrak{q}_{\mathrm{D}}^{\V{G}}$. Thus, 
$\Psi\in\dom(\mathfrak{q}_{\mathrm{D}}^{\V{G}})$ holds in the case $m>0$ as well.

Pick some $\Phi\in\sD(\Lambda,\fdom(\Id\Gamma(1\vee\omega)))$.
Thanks to Step~2 of the proof of Lem.~\ref{lem-ida} (applied to the dispersion relation $m+\omega$)
we already know that 
$\mathfrak{t}_{\mathrm{N}}^{\V{G}}[\Phi,\Psi]=\mathfrak{t}_{\mathrm{N}}^{\V{0}}[\Phi,\Psi]
+\mathfrak{b}^{\V{G}}[\Phi,\Psi]+\mathfrak{c}^{\V{G}}[\Phi,\Psi]$, 
where we use the notation introduced there. Taking also \eqref{egon99}
into account in the first and second equalities, 
 we infer from these remarks and \eqref{bakithi100div} that
\begin{align*}
&\mathfrak{q}_{\mathrm{D}}^{\V{G}}[\Phi,\Psi]
=\mathfrak{t}_{\mathrm{D}}^{\V{G}}[\Phi,\Psi]+\mathfrak{f}[\Phi,\Psi]
\\
&=\mathfrak{q}_{\mathrm{D}}^{\V{0}}[\Phi,\Psi]
-\frac{1}{2}\sum_{j=1}^\nu\big(\SPn{{w}_j^*\Phi}{\vp(G_j)\Psi}+
\SPn{\vp(G_j)\Phi}{{w}_j^*\Psi}-\SPn{\vp(G_j)\Phi}{\vp(G_j)\Psi}\big)
\\
&=\SPB{\Phi}{H_{\mathrm{D}}^{\V{0}}\Psi-\vp(\V{G})\cdot{\V{w}}^*\,\Psi+\frac{i}{2}\vp(q)\Psi
+\frac{1}{2}\vp(\V{G})^2\Psi}.
\end{align*}
Since $\sD(\Lambda,\fdom(\Id\Gamma(1\vee\omega)))$ is a core for 
$\mathfrak{q}_{\mathrm{D}}^{\V{G}}$, this proves $\Psi\in\dom(H_{\mathrm{D}}^{\V{G}})$
and \eqref{bakithi1}; see \cite[Thm.~VI.2.1(iii)]{Kato}.
For we know that $w_j^*\Psi=\ol{w}_j\Psi$ since 
$\dom(\mathfrak{q}_{\mathrm{D},\omega+m}^{\V{0}})\subset\dom(\ol{w}_j)$.
\end{proof}

In combination with Thm.~\ref{thm-dom} below the next remark extends 
\cite[Thm.~7]{HaslerHerbst2008}.

\begin{remark}\label{remaltrep}
Let $\V{A}$ and $V_\pm$ fulfill the hypotheses in Prop.~\ref{proprep}.
Assume that $G_j\in L^\infty(\Lambda,\fdom(\omega^{-1}))$ has a weak partial derivative
with respect to $x_j$ satisfying $\partial_{x_j}G_j\in L^\infty(\Lambda,\fdom(\omega^{-1}))$, for all
$j\in\{1,\ldots,\nu\}$. This easily implies that $\V{G}$ has a weak divergence given by 
$\mathrm{div}\,\V{G}=q:=\sum_{j=1}^\nu\partial_{x_j}G_j\in L^\infty(\Lambda,\fdom(\omega^{-1}))$
and thus strengthens the hypothesis of Prop.~\ref{proprep}(1).
Let $\Psi\in\dom(H_{\Lambda,\mathrm{D}}^{\V{0},\V{A},V})$.
Then the following alternative formula is valid, with the third term on its right hand side
defined in analogy to \eqref{shorty0},
\begin{align}\label{bakithi1alt}
H_{\Lambda,\mathrm{D}}^{\V{G},\V{A},V}\Psi
&=H_{\Lambda,\mathrm{D}}^{\V{0},\V{A},V}\Psi-\frac{1}{2}\vp(\V{G})\cdot\ol{\V{w}}\Psi
-\frac{1}{2}\V{w}^*\cdot\vp(\V{G})\Psi
+\frac{1}{2}\vp(\V{G})^2\Psi.
\end{align}

To see this we apply Lem.~\ref{lemmagnLeibniz} to the vectors
$\V{G}^{(j)}$ with components $G^{(j)}_{\ell}:=\delta_{j,\ell}G_j$, $j,\ell\in\{1,\ldots,\nu\}$,
which reveals that 
$
w_j^*\vp(G_j)\Psi=\vp(G_j)\ol{w}_j\Psi-i\vp(\partial_{x_j}G_j)\Psi
$.
(Recall that $\dom(H_{\Lambda,\mathrm{D}}^{\V{0},\V{A},V})\subset\dom(\theta)
\cap\bigcap_{j=1}^\nu\dom(\theta^\eh w_j^*)$ according to Lem.~\ref{lem-free-ham}(3).) 
Summing these identities over $j$, we see that
$\V{w}^*\cdot\vp(\V{G})\Psi=\vp(\V{G})\cdot\ol{\V{w}}\Psi-i\vp(q)\Psi$, whence \eqref{bakithi1alt}
follows from \eqref{bakithi1}.
\end{remark}

The next lemma already determines the domain of the Dirichlet-Pauli-Fierz operator when the
dispersion relation is sufficiently large compared to $\V{G}$. (Choose $m=0$ in the lemma.)
This is a direct analogue of the well-known weak coupling result \cite{BFS1998b}.

\begin{lemma}\label{lemrbK}
%
% To modify the statement of Lemma 5.4 copy and paste from here .....
%
Let $\V{A}\in L^2_\loc(\Lambda,\RR^\nu)$ and $V_\pm\in L^1_\loc(\Lambda,\RR)$, $V_\pm\ge0$, 
satisfy \eqref{KLMN} with $\diamond=\mathrm{D}$ for some $a\in[0,1)$ and $b>0$. 
Let $\alpha\ge1$, $m\ge0$, and set 
$$
\tilde{\omega}:=\left\{\begin{array}{ll}\omega,& \;\text{if $m=0$,}
\\ 1, &\;\text{if $m>0$.}\end{array}\right.
$$ 
Let $\V{G}\in L^\infty(\Lambda,\fdom(\tilde{\omega}^{-1})^\nu)$ have a weak divergence 
$q:=\Div\V{G}\in L^\infty(\Lambda,\fdom(\tilde{\omega}^{-1}))$, and write 
$$
c_{\V{G}}:=2\|\V{G}\|_{L^\infty(\Lambda,\fdom(\tilde{\omega}^{-1})^\nu)}\quad
\text{and} \quad c_q:=\|q\|_{L^\infty(\Lambda,\fdom(\tilde{\omega}^{-1}))}.
$$
Abbreviate 
$$
K^{\V{0}}:=H_{\Lambda,\mathrm{D},\alpha\omega+m}^{\V{0},\V{A},V},
$$
and define the operator $K^{\V{G}}$ by $\dom(K^{\V{G}}):=\dom(K^{\V{0}})$ and
$$
K^{\V{G}}\Psi:=H_{\Lambda,\mathrm{D}}^{\V{G},\V{A},V}\Psi+\Id\Gamma((\alpha-1)\omega+m)\Psi, 
\quad\Psi\in\dom(K^{\V{G}}).
$$ 
Finally, set $\beta:=\alpha$, if $m=0$, and $\beta=m$, if $m>0$.
Then there exist constants $\gamma_0,\gamma_1>0$ depending only on
$a$, $b$, $c_{\V{G}}$, and $c_q$ such that the following bounds hold, for all
$\Psi\in\dom(K^{\V{0}})$, provided that $\beta\ge\gamma_0$,
%
% ...... to here.
%
\begin{align}\label{bakithi2}
\|(K^{\V{G}}-K^{\V{0}})\Psi\|&=
\|(H_{\Lambda,\mathrm{D}}^{\V{G},\V{A},V}-H_{\Lambda,\mathrm{D}}^{\V{0},\V{A},V})\Psi\|
\le\frac{1}{2}\|K^{\V{0}}\Psi\|+\gamma_1\|\Psi\|,
\\\label{letitia}
\|\Id\Gamma(\alpha\omega+m)\Psi\|&\le\|K^{\V{0}}\Psi\|+\|\Psi\|
\le2\|K^{\V{G}}\Psi\|+(2\gamma_1+1)\|\Psi\|,
\\\label{letitia2}
\|H_{\Lambda,\mathrm{D}}^{\V{G},\V{A},V}\Psi\|&\le3\|K^{\V{G}}\Psi\|+(2\gamma_1+1)\|\Psi\|.
\end{align}
In particular, $K^{\V{G}}$ is self-adjoint, in fact equal to 
$H^{\V{G},\V{A},V}_{\Lambda,\mathrm{D},\alpha\omega+m}$,
and it has the same operator cores as $K^{\V{0}}$, again provided that $\beta\ge\gamma_0$.
\end{lemma}

\begin{proof}
On account of Lem.~\ref{lem-free-ham}(1) and Prop.~\ref{proprep}, $K^{\V{G}}$ is a well-defined
restriction of the self-adjoint operator $H_{\Lambda,\mathrm{D},\alpha\omega+m}^{\V{G},\V{A},V}$.
Hence, if $K^{\V{G}}$ is self-adjoint, then these two operators agree.
The identity in \eqref{bakithi2} follows from \eqref{bakithi1} and Lem.~\ref{lem-free-ham}(1). 
Due to Lem.~\ref{lem-free-ham}(1) we further have
$\dom(K^{\V{0}})\subset\dom(H_{\Lambda,\mathrm{D},\tilde{\omega}}^{\V{0},\V{A},V})$.
Therefore, \eqref{rbvp}, \eqref{rbvp2}, \eqref{bakithi1b} applied to the dispersion relation 
$\tilde{\omega}$, and the representation in \eqref{bakithi1} entail
\begin{align*}
\|(H_{\Lambda,\mathrm{D}}^{\V{G},\V{A},V}&-H_{\Lambda,\mathrm{D}}^{\V{0},\V{A},V})\Psi\|
\\
&\le 2^\eh c_{\V{G}}\sqrt{\frac{1+b}{1-a}}
\|(T_{\Lambda,\mathrm{D}}^{\V{0},\V{A},V}+1)\Psi\|^\eh\|\tilde{\theta}\Psi\|^\eh
+(c_q+c_{\V{G}}^2)\|\tilde{\theta}\Psi\|,
\end{align*}
with $\tilde{\theta}:=1+\Id\Gamma(\tilde{\omega})$.
In view of the bound in Lem.~\ref{lem-free-ham}(1) and
$\tilde{\omega}\le\beta^{-1}(\alpha\omega+m)$ this yields the inequality in \eqref{bakithi2}. 
Now the last assertion of the lemma follows from the Kato-Rellich theorem.

Furthermore, the first bound in \eqref{letitia} follows from Lem.~\ref{lem-free-ham}(1)
while the second one is a consequence of \eqref{bakithi2}. Finally, \eqref{letitia2}
is implied by \eqref{letitia} and the relation
$H_{\Lambda,\mathrm{D}}^{\V{G},\V{A},V}\Psi=K^{\V{G}}\Psi-\Id\Gamma((\alpha-1)\omega+m)\Psi$,
where $0\le(\alpha-1)\omega\le\alpha\omega$.
\end{proof}

The next theorem extends a result of \cite{Falconi2015}. The idea to use operators like
$K^{\V{G}}$, which are more manifestly self-adjoint thanks to an artificially enlarged dispersion 
relation, as comparison operators in an application of Nelson's commutator theorem goes back to
M.~K\"{o}nenberg \cite[Lem.~3.1]{KMS2013}.

\begin{theorem}\label{thmMM}
Let $\V{A}\in L^2_\loc(\Lambda,\RR^\nu)$ and $V_\pm\in L^1_\loc(\Lambda,\RR)$, $V_\pm\ge0$, 
satisfy \eqref{KLMN} with $\diamond=\mathrm{D}$ for some $a\in[0,1)$ and $b>0$. 
Assume that $\V{G}\in L^\infty(\Lambda,\HP^\nu)$ has a weak divergence 
$q:=\Div\V{G}\in L^\infty(\Lambda,\HP)$.
Then every core for $H_{\Lambda,\mathrm{D},\omega+1}^{\V{0},\V{A},V}$ is a core for
$H_{\Lambda,\mathrm{D}}^{\V{G},\V{A},V}$ as well.
\end{theorem}

\begin{proof}
To start with we observe that, in view of Lem.~\ref{lem-free-ham}(1), the graph norms
of $H_{\Lambda,\mathrm{D},\omega+1}^{\V{0},\V{A},V}$ and any 
$H_{\Lambda,\mathrm{D},\omega+m}^{\V{0},\V{A},V}$ with $m>0$ are equivalent and both operators
have the same domain and the same operator cores. In what follows we shall employ
Lem.~\ref{lemrbK} with $\alpha=1$ and choose $m>0$ so large that the operator given by
$K^{\V{G}}\Psi=H^{\V{G},\V{A},V}_{\Lambda,\mathrm{D}}\Psi+\Id\Gamma(m)\Psi$,
$\Psi\in\dom(K^{\V{G}})=\dom(H_{\Lambda,\mathrm{D},\omega+1}^{\V{0},\V{A},V})$,
is self-adjoint, in fact equal to $H^{\V{G},\V{A},V}_{\Lambda,\mathrm{D},\omega+m}$,
satisfies \eqref{letitia2} for all $\Psi\in\dom(K^{\V{G}})$, and has the same
operator cores as $H_{\Lambda,\mathrm{D},\omega+1}^{\V{0},\V{A},V}$.

Let $\Psi\in\dom(S_{\Lambda,\mathrm{D}}^{\V{A},V})\otimes\dom(\Id\Gamma(\omega+m)^2)$,
which is a core for $K^{\V{G}}$ according to the previous remarks and Lem.~\ref{lem-free-ham}(2).
By Prop.~\ref{proprep}(2), $\Psi\in\dom(H_{\Lambda,\mathrm{D}}^{\V{G},\V{A},V})$ and we may 
use \eqref{bakithi1} to represent $H_{\Lambda,\mathrm{D}}^{\V{G},\V{A},V}\Psi$. We thus find
\begin{align}\nonumber
-&2\Im\SPn{K^{\V{G}}\Psi}{H_{\Lambda,\mathrm{D}}^{\V{G},\V{A},V}\Psi}
\\\nonumber
&=-2\Im\SPn{\Id\Gamma(m)\Psi}{H_{\Lambda,\mathrm{D}}^{\V{G},\V{A},V}\Psi}
\\\label{vinnie1}
&=\Im\SPn{\Id\Gamma(m)\Psi}{2\vp(\V{G})\cdot\V{w}^*\Psi-i\vp(q)\Psi}
-\Im\SPn{\Id\Gamma(m)\Psi}{\vp(\V{G})^2\Psi}.
\end{align}
Noticing that the integration by parts formula \eqref{bakithi100div} extends to all
$\Phi\in\dom(\mathfrak{t}_{\Lambda,\mathrm{D}}^{\V{0},\V{A},0})$
and that $\Id\Gamma(m)\Psi\in\dom(\mathfrak{t}_{\Lambda,\mathrm{D}}^{\V{0},\V{A},0})$ 
by Lem.~\ref{lemstST}(1), we further obtain
\begin{align*}
\big|&\Im\SPn{\Id\Gamma(m)\Psi}{2\vp(\V{G})\cdot\V{w}^*\Psi-i\vp(q)\Psi}\big|
\\
&=\Big|\sum_{j=1}^\nu\big(\Im\SPn{\Id\Gamma(m)\Psi}{\vp({G_j}){w}_j^*\Psi}
+\Im\SPn{\Id\Gamma(m){w}_j^*\Psi}{\vp({G_j})\Psi}\big)\Big|
\\
&\le m\sum_{j=1}^\nu\big|\SPn{\vp(i{G_j})\Psi}{{w}_j^*\Psi}\big|
\\
&\le2^{\nf{3}{2}}m^\eh\|\V{G}\|_{L^\infty(\Lambda,\HP^\nu)}\|(m+\Id\Gamma(m))^\eh\Psi\|
\mathfrak{t}_{\Lambda,\mathrm{D}}^{\V{0},\V{A},0}[\Psi]^\eh,
\end{align*}
where we wrote out the imaginary parts and applied \eqref{CRvpdG} in the penultimate step.
In the second and penultimate steps we also used Rem.~\ref{remwCCw}(1).
The identity \eqref{CRvpdG} reveals that $\vp(G_{j,\V{x}})$ maps $\dom(\Id\Gamma(1)^{\nf{3}{2}})$
into $\dom(\Id\Gamma(1))$ and in view of this it further implies
\begin{align*}
\Im\SPn{\Id\Gamma(m)\Psi}{\vp(\V{G})^2\Psi}=m\sum_{j=1}^\nu\Im\SPn{i\vp(iG_j)\Psi}{\vp(G_j)\Psi}.
\end{align*}
Putting these remarks together and employing \eqref{harald3} we deduce that
\begin{align*}
|\Im\SPn{K^{\V{G}}\Psi}{H_{\Lambda,\mathrm{D}}^{\V{G},\V{A},V}\Psi}|&\le
c(m+\|\V{G}\|_{L^\infty(\Lambda,\HP^\nu)}^2)\big(
\mathfrak{q}_{\Lambda,\mathrm{D},\omega+m}^{\V{0},\V{A},V_+}[\Psi]+m\|\Psi\|^2\big),
\end{align*}
for some universal constant $c>0$. Since by Lem.~\ref{lem-ida} (applied to the dispersion relation
$\omega+m$) and \eqref{harald3} the form norms of 
$\mathfrak{q}_{\Lambda,\mathrm{D},\omega+m}^{\V{0},\V{A},V_+}$
and $\mathfrak{q}_{\Lambda,\mathrm{D},\omega+m}^{\V{G},\V{A},V}$
are equivalent, and since $K^{\V{G}}=H^{\V{G},\V{A},V}_{\Lambda,\mathrm{D},\omega+m}$,
the assertion now follows from Nelson's commutator theorem.
\end{proof}

The next lemma will be used in the crucial Step~2 of the proof of Thm.~\ref{thm-dom}.

\begin{lemma}\label{lemecho}
Let $\V{A}\in L^2_\loc(\Lambda,\RR^\nu)$, 
$\V{G}\in L^\infty(\Lambda,\fdom(\omega^{-1}+\omega)^\nu)$, 
$\diamond\in\{\mathrm{D},\mathrm{N}\}$, 
and assume that $V_\pm\ge0$ satisfy \eqref{KLMN} for some $a\in[0,1)$ and $b\ge0$.
Pick some $j\in\{1,\ldots,\nu\}$ and 
$\Psi\in\dom(S_{\Lambda,\diamond}^{\V{A},V})\otimes\dom(\theta^{\nf{3}{2}})$.
Then $\theta\Psi\in\dom(\mathfrak{t}_{\Lambda,\diamond}^{\V{G},\V{A},V})
\cap\dom(\mathfrak{f}_\Lambda)$, $\theta^\eh\Psi\in\dom({v}_j^*)$, and
the following bound holds for every $\beta>0$,
\begin{align*}
\Re\SPn{{v}_j^*\theta\Psi}{{v}_j^*\Psi}&\ge(1-\beta)\|{v}_j^*\theta^\eh\Psi\|^2
-4\Big(1+\frac{1}{\beta}\Big)\|\omega^\eh{G}_{j}\|_{L^\infty(\Lambda,\HP)}^2\|\theta^\eh\Psi\|^2.
\end{align*}
\end{lemma}

\begin{proof}
Since $\theta\Psi\in\dom(S_{\Lambda,\diamond}^{\V{A},V})\otimes\dom(\theta^{\nf{1}{2}})$, 
Cor.~\ref{corformcore} implies
$\theta\Psi\in\dom(\mathfrak{q}_{\Lambda,\diamond}^{\V{G},\V{A},V})\subset\dom({v}_j^*)$.
Furthermore, let $f\in\dom(S_{\Lambda,\diamond}^{\V{A},V})$ and
$\psi\in\dom(\theta^{\nf{1}{2}})$. Then $f\in\dom({\sf w}_j^*)$ and it follows from 
Rem.~\ref{remwCCw}(1) and Rem.~\ref{remvj}(1) that
$f\psi\in\dom({v}_j^*)$ with ${v}_j^*(f\psi)=({\sf w}_j^*f)\psi-\vp(G_j)(f\psi)$. Since
$\vp(G_{j,\V{x}})$ maps $\dom(\theta)$ into $\dom(\theta^\eh)$ by Lem.~\ref{lem-comm}, 
we may thus write
\begin{align*}
\theta^\mh({v}_j^*\theta\Psi)(\V{x})&=({v}_j^*\theta^\eh\Psi)(\V{x})
-(\theta^\mh\vp(G_{j,\V{x}})-\vp(G_{j,\V{x}})\theta^\mh)\theta\Psi(\V{x})
\\
&=({v}_j^*\theta^\eh\Psi)(\V{x})-T(G_{j,\V{x}})^*\theta^\eh\Psi(\V{x}),
\\
\theta^\eh({v}_j^*\Psi)(\V{x})&=({v}_j^*\theta^\eh\Psi)(\V{x})-T(G_{j,\V{x}})\theta^\eh\Psi(\V{x}),
\end{align*}
for a.e. $\V{x}\in\Lambda$, where the bounded operator
$T(G_{j,\V{x}})$ is defined as in Lem.~\ref{lem-comm}(2).
Combining the above identities we obtain
\begin{align*}
\Re\SPn{({v}_j^*&\theta\Psi)(\V{x})}{({v}_j^*\Psi)(\V{x})}
\\
&\ge(1-\beta)\|({v}_j^*\theta^\eh\Psi)(\V{x})\|^2-\Big(1+\frac{1}{\beta}\Big)\
\|T(G_{j,\V{x}})\|^2\|\theta^\eh\Psi(\V{x})\|^2,
\end{align*}
for a.e. $\V{x}\in\Lambda$, and we conclude by applying the bound in Lem.~\ref{lem-comm}(2).
\end{proof}

Finally, we are in a position to prove the main theorem of this paper. In the subsequent
remark we discuss a simple extension to the case where additional spin degrees of freedom are 
taken into account. 

\begin{theorem}\label{thm-dom}
Let $\V{A}\in L^2_\loc(\Lambda,\RR^\nu)$ and
$\V{G}\in L^\infty(\Lambda,\fdom(\omega^{-1}+\omega)^\nu)$.
Assume that $V_\pm\in L^1_\loc(\Lambda,\RR)$, $V_\pm\ge0$, satisfy the form bound
\eqref{KLMN} with $\diamond=\mathrm{D}$ for some $a\in[0,1)$ and $b>0$. Assume further that 
$\V{G}$ has a weak divergence $q:=\mathrm{div}\V{G}\in L^\infty(\Lambda,\fdom(\omega^{-1}))$. Then 
\begin{align*}
\dom(H_{\Lambda,\mathrm{D}}^{\V{G},\V{A},V})
=\dom(H_{\Lambda,\mathrm{D}}^{\V{0},\V{A},V})=
\dom(T_{\Lambda,\mathrm{D}}^{\V{0},\V{A},V})\cap L^2(\Lambda,\dom(\Id\Gamma(\omega))),
\end{align*}
the graph norms of $H_{\Lambda,\mathrm{D}}^{\V{G},\V{A},V}$ and 
$H_{\Lambda,\mathrm{D}}^{\V{0},\V{A},V}$ are equivalent and, in particular,
both operators have the same cores.
\end{theorem}

Recall that Lem.~\ref{lem-free-ham}(2) identifies a natural class of cores for 
$H_{\Lambda,\mathrm{D}}^{\V{0},\V{A},V}$.

\begin{proof}
Without loss of generality we shall assume that $\mathfrak{s}_{\Lambda,\mathrm{D}}^{\V{A},V}\ge0$. 
From now on we will also drop all sub/superscripts $\Lambda$, $\V{A}$, and $V$ in this proof.

{\em Step~1.}
We start by observing that Lem.~\ref{lem-free-ham}(1) implies 
$\dom(H_{\mathrm{D},\alpha\omega}^{\V{0}})=\dom(H_{\mathrm{D}}^{\V{0}})$ 
and $H_{\mathrm{D},\alpha\omega}^{\V{0}}\Psi=T_{\mathrm{D}}^{\V{0}}\Psi
+\alpha\Id\Gamma(\omega)\Psi$, 
$\Psi\in\dom(H_{\mathrm{D}}^{\V{0}})$, for all $\alpha\ge1$.
Furthermore, the graph norms of $H_{\mathrm{D}}^{\V{0}}$ and 
$H_{\mathrm{D},\alpha\omega}^{\V{0}}$ are equivalent and both operators have the same
cores, for all $\alpha\ge1$. 
Thanks to Prop.~\ref{proprep}(1) we already know that
$\dom(H_{\mathrm{D}}^{\V{0}})\subset\dom(H_{\mathrm{D}}^{\V{G}})$.
We employ Lem.~\ref{lemrbK} with $m=0$ and fix $\alpha\ge1$ sufficiently large 
throughout this proof such that the operator given by
$K^{\V{G}}\Psi=H_{\mathrm{D}}^{\V{G}}\Psi+(\alpha-1)\Id\Gamma(\omega)\Psi$,
$\Psi\in\dom(K^{\V{G}})=\dom(H_{\mathrm{D}}^{\V{0}})$, is self-adjoint
and has the same operator cores as $H_{\mathrm{D}}^{\V{0}}$ and such that
\eqref{bakithi2}--\eqref{letitia2} are available.
Then \eqref{bakithi2} and Lem.~\ref{lem-free-ham}(1) imply
\begin{align}\nonumber
\|H_{\mathrm{D}}^{\V{G}}\Psi\|&\le\|H_{\mathrm{D}}^{\V{0}}\Psi\|
+\frac{1}{2}\|H_{\mathrm{D},\alpha\omega}^{\V{0}}\Psi\|+\gamma_1\|\Psi\|
\\\label{bakithi3}
&\le (\alpha+1)\|H_{\mathrm{D}}^{\V{0}}\Psi\|+(\alpha+\gamma_1)\|\Psi\|,
\end{align}
for all $\;\Psi\in\dom(H_{\mathrm{D}}^{\V{0}})$.

{\em Step~2.}
In this step we derive the following key estimate of this proof:
Set $\tilde{c}_{\V{G}}:=2\|\omega^\eh\V{G}\|_{L^\infty(\Lambda,\HP^\nu)}$.
Then there exists some $(a,b,\tilde{c}_{\V{G}})$-dependent constant $c_0>0$ such that, 
for all $\Psi\in\dom(H_{\mathrm{D}}^{\V{0}})$ and $\delta\in(0,1]$,
\begin{align}\nonumber
D[\Psi]&:=\|(H_{\mathrm{D}}^{\V{G}}+1)\Psi\|^2-\|\theta\Psi\|^2
-\|(H_{\mathrm{D}}^{\V{G}}-\Id\Gamma(\omega))\Psi\|^2
\\\label{bakithi4}
&\ge-\delta\|\theta\Psi\|^2-(c_0/\delta)\|\Psi\|^2.
\end{align}

In view of the already proven inclusion 
$\dom(H_{\mathrm{D}}^{\V{0}})\subset\dom(H_{\mathrm{D}}^{\V{G}})$, \eqref{bakithi3}, and the 
closedness of $H_{\mathrm{D}}^{\V{G}}$, it suffices to prove \eqref{bakithi4} for every element 
$\Psi$ in a core of $H_{\mathrm{D}}^{\V{0}}$. A suitable core is
$\dom(S_{\mathrm{D}})\otimes\dom(\theta^{\nf{3}{2}})$; recall Lem.~\ref{lem-free-ham}(2).

So, let $\Psi\in\dom(S_{\mathrm{D}})\otimes\dom(\theta^{\nf{3}{2}})$. Then 
$\Psi\in\dom(H_{\mathrm{D}}^{\V{G}})\cap L^2(\Lambda,\dom(\Id\Gamma(\omega)))$
and Lem.~\ref{lemecho} ensures that
$\theta\Psi\in\dom(\mathfrak{q}_{\mathrm{D}}^{\V{G}})=\dom(\mathfrak{t}_{\mathrm{D}}^{\V{G}})
\cap\dom(\mathfrak{f})$.
Taking also Lem.~\ref{lem-egon} into account we thus obtain
\begin{align}\nonumber
\SPn{\theta\Psi}{(H_{\mathrm{D}}^{\V{G}}-\Id\Gamma(\omega))\Psi}
&=\mathfrak{q}_{\mathrm{D}}^{\V{G}}[\theta\Psi,\Psi]-\mathfrak{f}[\theta\Psi,\Psi]
=\mathfrak{t}_{\mathrm{D}}^{\V{G}}[\theta\Psi,\Psi]
\\\label{echo1}
&=\frac{1}{2}\sum_{j=1}^\nu\SPb{{v}_j^*\theta\Psi}{{v}_j^*\Psi}
+\mathfrak{v}^+[\theta^\eh\Psi]-\mathfrak{v}^-[\theta^\eh\Psi].
\end{align}
By our assumptions on $V_\pm$  and the diamagnetic inequality \eqref{dia1},
\begin{align*}
\frac{a}{2}\sum_{j=1}^\nu\|{v}_j^*\theta^\eh\Psi\|^2+a\mathfrak{v}^+[\theta^\eh\Psi]
-\mathfrak{v}^-[\theta^\eh\Psi]\ge-b\|\theta^\eh\Psi\|^2,
\end{align*}
for all $\Psi\in\dom(S_{\mathrm{D}})\otimes\dom(\theta^{\nf{3}{2}})$.
Combining \eqref{echo1} and Lem.~\ref{lemecho} with $\beta:=1-a>0$ we then arrive at
\begin{align*}
D[\Psi]&=2\Re\SPn{\theta\Psi}{(H_{\mathrm{D}}^{\V{G}}-\Id\Gamma(\omega))\Psi}
\\
&\ge-\Big(\frac{2-a}{1-a}\tilde{c}_{\V{G}}^2+2b\Big)\|\theta^\eh\Psi\|^2
\ge-\delta\|\theta\Psi\|^2-\frac{1}{\delta}\Big(\frac{2-a}{1-a}\frac{\tilde{c}_{\V{G}}^2}{2}+b\Big)^2
\|\Psi\|^2,
\end{align*}
for all $\Psi\in\dom(S_{\mathrm{D}})\otimes\dom(\theta^{\nf{3}{2}})$ and $\delta\in(0,1]$. 

{\em Step 3.} Next, we show that the graph norms of 
$H_{\mathrm{D}}^{\V{G}}\!\!\upharpoonright_{\dom(H_{\mathrm{D}}^{\V{0}})}$ and $K^{\V{G}}$ 
are equivalent. This will imply that 
$H_{\mathrm{D}}^{\V{G}}\!\!\upharpoonright_{\dom(H_{\mathrm{D}}^{\V{0}})}$ 
is closed and that every core for $K^{\V{G}}$ is a core for 
$H_{\mathrm{D}}^{\V{G}}\!\!\upharpoonright_{\dom(H_{\mathrm{D}}^{\V{0}})}$ and vice versa;
by Step~1 of this proof we then also know that every core for $H_{\mathrm{D}}^{\V{0}}$ is a core for 
$H_{\mathrm{D}}^{\V{G}}\!\!\upharpoonright_{\dom(H_{\mathrm{D}}^{\V{0}})}$ and vice versa.

In view of \eqref{letitia2} it only remains to dominate the graph norm of $K^{\V{G}}$ by the
one of $H_{\mathrm{D}}^{\V{G}}\!\!\upharpoonright_{\dom(H_{\mathrm{D}}^{\V{0}})}$. 
To this end let $\Psi\in\dom(H_{\mathrm{D}}^{\V{0}})$.
Using $1\le\alpha$ in the first step and \eqref{bakithi4} in the second one, we obtain
\begin{align*}
\|K^{\V{G}}\Psi\|^2&\le2\alpha^2\big(\|(H_{\mathrm{D}}^{\V{G}}-\Id\Gamma(\omega))\Psi\|^2
+\|\theta\Psi\|^2\big)
\\
&\le2\alpha^2\big(\|(H_{\mathrm{D}}^{\V{G}}+1)\Psi\|^2
+\delta\|\theta\Psi\|^2+(c_0/\delta)\|\Psi\|^2\big),
\end{align*}
which in combination with $\|\theta\Psi\|\le2\|K^{\V{G}}\Psi\|+2(1+\gamma_1)\|\Psi\|$ 
(recall \eqref{letitia}) and for sufficiently small $\delta\in(0,1]$ readily leads to the bound 
$$
\|K^{\V{G}}\Psi\|\le c\alpha\|H_{\mathrm{D}}^{\V{G}}\Psi\|+c'\|\Psi\|,
$$
with a universal constant $c>0$ and a $(\alpha,\gamma_1,c_0)$-dependent constant $c'>0$.

{\em Step 4.} Finally, we show that 
$H_{\mathrm{D}}^{\V{G}}\!\!\upharpoonright_{\dom(H_{\mathrm{D}}^{\V{0}})}$ is a self-adjoint
restriction of the self-adjoint operator $H_{\mathrm{D}}^{\V{G}}$, which implies that
$H_{\mathrm{D}}^{\V{G}}\!\!\upharpoonright_{\dom(H_{\mathrm{D}}^{\V{0}})}
=H_{\mathrm{D}}^{\V{G}}$. To this end we apply the following theorem for operators
in a Hilbert space (\cite{Wuest1972}; see also \cite[Thm.~5.29]{Weidmann}):
\begin{enumerate}
\item[] {\sl If $A$ is self-adjoint, $B$ is symmetric and $A$-bounded, and $A+tB$ is closed,
for all $t\in[0,1]$, then $A+B$ is self-adjoint.}
\end{enumerate}
We apply this theorem with $A=K^{\V{G}}$ and $B=(1-\alpha)\Id\Gamma(\omega)$, so that
$B$ is $A$-bounded by \eqref{letitia}.
Furthermore, with these choices, $A+tB$ is equal to the operator that we
obtained, if we replaced $\omega$ by 
$\omega_t:=(1-t)\alpha\omega+t\omega$ in the construction of
$H_{\mathrm{D}}^{\V{G}}\!\!\upharpoonright_{\dom(H_{\mathrm{D}}^{\V{0}})}$, i.e., $A+tB=
H_{\mathrm{D},\omega_t}^{\V{G}}\!\!\upharpoonright_{\dom(H_{\mathrm{D},\omega_t}^{\V{0}})}$.
In particular, $A+B=H_{\mathrm{D}}^{\V{G}}\!\!\upharpoonright_{\dom(H_{\mathrm{D}}^{\V{0}})}$.
Since the pair $(\omega_t,\V{G})$ satisfies the assumptions of the theorem, for all $t\in[0,1]$,
every $A+tB$, $t\in[0,1]$, is closed according to Step~3.
\end{proof}

\begin{remark}\label{remZeeman}
{\em Linearly coupled fields and Zeeman terms.}

\smallskip

\noindent
Assume that $0\le V_+\in L^1_\loc(\Lambda)$ and that $\V{A}$ and $\V{G}$ 
fulfill all hypotheses of Thm.~\ref{thm-dom}. We
explain how to extend the theorem so as to cover additional linearly coupled fields 
or Zeeman terms accounting for additional spin degrees of freedom. To this end let 
$\tilde{\nu},s\in\NN$ and let $\sigma_1,\ldots,\sigma_{\tilde{\nu}}$ be 
Hermitian $s$\texttimes$s$-matrices with norm equal to one. 
We extend the configuration space as $\Lambda_*:=\Lambda\times\{1,\ldots,s\}$
and consider each matrix $\sigma_j$ as a self-adjoint operator on $L^2(\Lambda_*,\sF)$ 
by setting $(\sigma_j\Psi)(\V{x},\vs):=\sum_{\vs'=1}^s(\sigma_j)_{\vs,\vs'}\Psi(\V{x},\vs')$, for
a.e. $(\V{x},\vs)\in\Lambda_*$ and all $\Psi\in L^2(\Lambda_*,\sF)$.

An easy way to include possibly singular classical Zeeman terms $\vsigma\cdot\V{B}$ is to
generalize the constructions in Sect.~\ref{sec-ham} to (not necessarily non-negative)
matrix-valued  $V_-\in L^1_\loc(\Lambda,\LO(\CC^s))$. Then each $V_-(\V{x})$, $\V{x}\in\Lambda$, 
is assumed to be a Hermitian $s$\texttimes$s$-matrix and the condition \eqref{KLMN} is replaced by
\begin{align*}%\label{KLMNmatrix}
\int_\Lambda\|V_-(\V{x})\|_{\LO(\CC^s)}|f(\V{x})|^2\Id\V{x}&\le 
a\mathfrak{s}_{\Lambda,\diamond}^{\V{0},V_+}[f]+b\|f\|^2,\quad 
f\in\dom(\mathfrak{s}_{\Lambda,\diamond}^{\V{0},V_+}),
\end{align*}
for some $a\in[0,1)$ and $b>0$.
The definition of the Schr\"{o}dinger forms with matrix-valued $V_-$ reads
\begin{align*}
\mathfrak{s}_{\Lambda_*,\diamond}^{\V{A},V}[f]
&:=\sum_{\vs=1}^s\mathfrak{s}_{\Lambda,\diamond}^{\V{A},V_+}[f(\cdot,\vs)]-
\sum_{\vs,\vs'=1}^s\int_\Lambda V_-(\V{x})_{\vs,\vs'}\ol{f}(\V{x},\vs)f(\V{x},\vs')\Id\V{x},
\end{align*}
for all $f\in\dom(\mathfrak{s}_{\Lambda_*,\diamond}^{\V{A},V_+})$,
where $\dom(\mathfrak{s}_{\Lambda_*,\diamond}^{\V{A},V_+})$ is the set of functions
$f:\Lambda_*\to\CC$ with $f(\cdot,\vs)\in\dom(\mathfrak{s}_{\Lambda,\diamond}^{\V{A},V_+})$,
for all $\vs\in\{1,\ldots,s\}$. Likewise, we replace the old definition of $\mathfrak{v}_\Lambda^-$ by
\begin{align*}
\mathfrak{v}_{\Lambda_*}^-[\Psi]&:=\sum_{\vs,\vs'=1}^s\int_{\Lambda}
V_-(\V{x})_{\vs,\vs'}\SPn{\Psi(\V{x},\vs)}{\Psi(\V{x},\vs')}_{\sF}\Id\V{x},
\end{align*}
where $\Psi(\cdot,\vs)\in\fdom(\|V_-\|_{\LO(\CC^s)}\id_{\sF})$, for all $\vs\in\{1,\ldots,s\}$.
All other forms and operators can be 
extended to forms and operators in $L^2(\Lambda_*,\sF)$ in an obvious fashion. The so-obtained
Pauli-Fierz operators will be indicated by the subscript $\Lambda_*$ in what follows. It is then clear 
that all results of Sect.~\ref{sec-ham} and all previous results of the present section hold 
{\em mutatis mutandis} for matrix-valued $V_-$ as well.

Having implemented the generalization just described, we can further add a quantized Zeeman term: 
Let $\V{F}=(F_1,\ldots,F_{\tilde{\nu}})\in L^\infty(\Lambda,\fdom(\omega^{-1})^{\tilde{\nu}})$ 
and abbreviate
\begin{align*}
\V{\sigma}\cdot\vp(\V{F})\Psi&:=\sum_{j=1}^{\tilde{\nu}}\sigma_j\vp(F_j)\Psi,\quad
\Psi\in \dom(\V{\sigma}\cdot\vp(\V{F})):=
L^2(\Lambda_*,\fdom(\Id\Gamma(\omega))).
\end{align*}
On account of \eqref{rbvp} and $\|\sigma_j\|=1$ the previous expressions are well-defined and
\begin{align}\label{rbZeeman}
\|\V{\sigma}\cdot\vp(\V{F})\Psi\|&\le c_{\V{F}}\|\Psi\|_{\fdom(\Id\Gamma(\omega))}
\le \ve\|\Id\Gamma(\omega)\Psi\|+\Big(c_{\V{F}}+\frac{c_{\V{F}}^2}{4\ve}\Big)\|\Psi\|,
\end{align}
for all $\ve>0$ and $\Psi\in L^2(\Lambda_*,\dom(\Id\Gamma(\omega)))$, with 
$$
c_{\V{F}}:=2\sum_{j=1}^{\tilde{\nu}}\|F_j\|_{L^\infty(\Lambda,\fdom(\omega^{-1}))}.
$$
In view of Lem.~\ref{lem-free-ham}(1) and \eqref{rbZeeman}, the symmetric operator
$\V{\sigma}\cdot\vp(\V{F})$ is infinitesimally $H^{\V{0},\V{A},V}_{\Lambda_*,\mathrm{D}}$-bounded.
By the above extension of Thm.~\ref{thm-dom} it is infinitesimally 
$H^{\V{G},\V{A},V}_{\Lambda_*,\mathrm{D}}$-bounded as well.
The Kato-Rellich theorem and the extended Thm.~\ref{thm-dom} thus imply that 
$H^{\V{G},\V{A},V}_{\Lambda_*,\mathrm{D}}-\V{\sigma}\cdot\vp(\V{F})$ is self-adjoint on the
domain $\dom(H^{\V{0},\V{A},V}_{\Lambda_*,\mathrm{D}})$ and that is has the same operator
cores as $H^{\V{0},\V{A},V}_{\Lambda_*,\mathrm{D}}$. 
\end{remark}

%%%%%%%%%%%%%%%%%%%%%%%%%%%%%%%%%%%%%%%%%%
%%%%%%%%%%%%%%%%%%%%%%%%%%%%%%%%%%%%%%%%%%
%%%%%%%%%%%%%%%%%%%%%%%%%%%%%%%%%%%%%%%%%%

%\appendix

\section{Examples}\label{appex}

%%%%%%%%%%%%%%%%%%%%%%%%%%%%%%%%%%%%%%%%%%

\subsection{Examples of coupling functions}\label{appcoupling}

\noindent
In what follows we give several examples for physically relevant choices of $\V{G}$ and $\V{F}$
with $\Lambda\subset\RR^3$. The given formulas are suitable for Pauli-Fierz operators for one 
electron; in Subsect.~\ref{ssecPauli} we shall explain how to deal with several electrons.
In all cases $\V{G}$ and $\V{F}$ are determined by an expansion into proper or generalized 
eigenfunctions of the Maxwell operator on divergence free vector fields satisfying perfect electric 
conductor boundary conditions, if the boundary $\partial\Lambda$ is non-empty. 
The latter boundary conditions require the tangential components of 
the electric field and the normal component of the magnetic field to vanish on $\partial\Lambda$. We 
also have to artificially introduce ultra-violet cut-offs damping the interaction with very high frequency 
modes. The measure space $(\cM,\fA,\mu)$ will always be the mode space in the generalized 
eigenfunction expansion.

For later reference, we first note a very simple observation:

\begin{remark}\label{remapp}
Let $\cM\times\Lambda\ni(k,\V{x})\mapsto\V{G}_{\V{x}}(k)\in\CC^\nu$ be measurable such that,
for every $k\in\cM$, the map $\Lambda\ni\V{x}\mapsto\V{G}_{\V{x}}(k)$ is locally integrable and
has a weak divergence denoted by $\Lambda\ni\V{x}\mapsto q_{\V{x}}(k)\in\CC$. 
Assume in addition that the components of $\V{G}$ and $q$ belong to
$L_\loc^1(\Lambda,\fdom(\omega^{-1}))$.
Then $\V{G}$ has a weak divergence computed in $\fdom(\omega^{-1})$ given by $\Div\V{G}=q$.

This is a direct consequence of the relevant definitions and Fubini's theorem.
\end{remark}

\begin{example}\label{exsmoothbd}
Let $\Lambda\subset\RR^3$ be a bounded, simply connected domain with smooth boundary 
$\partial\Lambda$ and exterior normal field $\V{n}\in C^\infty(\partial\Lambda,\RR^3)$. 
We shall consider a self-adjoint realization of the Maxwell operator in 
$L^2(\Lambda,\CC^{6})$ corresponding to perfect electric conductor (ec) boundary conditions.
To this end we shall briefly summarize some constructions and results of \cite{Schmidt1968}.
(The article \cite{Schmidt1968} deals with exterior domains, but the results quoted below hold
for bounded domains as well; cf. the survey article \cite{BirmanSolomyak1987} and the references
given there for a more general discussion of the Hilbert space theory of the Maxwell operator
on Lipschitz domains.) We start by setting
\begin{align*}
M&:=\begin{pmatrix}0&i\rot\\-i\rot&0\end{pmatrix}\!\!\upharpoonright_{C_0^\infty(\Lambda,\CC^6)},
\\
\sC_{\mathrm{ec}}&:=\{(\V{E},\V{B})\in C^2(\ol{\Lambda},\CC^6)\,|\;\V{n}\times\V{E}=\V{0}\,\text{on $
\partial\Lambda$}\},
\end{align*}
and defining $M_{\mathrm{ec}}:=\ol{M^*\!\!\upharpoonright_{\sC_{\mathrm{ec}}}}$.
It turns out that $M_{\mathrm{ec}}$ is self-adjoint, and so is its restriction, call it 
$M_{\mathrm{ec}}^\perp$, to the orthogonal complement of its kernel $\ker(M_{\mathrm{ec}})$.
The elements in $\dom(M_{\mathrm{ec}}^\perp)$ have a vanishing weak divergence.
It further turns out that $\dom(M_{\mathrm{ec}}^\perp)\subset W^{1,2}(\Lambda,\CC^6)$, so that
$(\V{E},\V{B})\in\dom(M_{\mathrm{ec}}^\perp)$ has a well-defined trace on $\partial\Lambda$.
Denoting this trace by $\upharpoonright_{\partial\Lambda}$, 
every $(\V{E},\V{B})\in\dom(M_{\mathrm{ec}}^\perp)$ indeed 
satisfies the full set of perfect electric conductor boundary conditions in the sense that
$\V{n}\times\V{E}\!\!\upharpoonright_{\partial\Lambda}=\V{0}$ and 
$\V{n}\cdot\V{B}\!\!\upharpoonright_{\partial\Lambda}=0$.
Thm.~2.2.9 in \cite{Schmidt1968} (see also \cite{BirmanSolomyak1987})
implies that $M_{\mathrm{ec}}^\perp$ has a compact resolvent.
We also observe that, if $(\V{E},\V{B})\in\dom(M_{\mathrm{ec}}^\perp)\setminus\{0\}$ is an 
eigenvector for the eigenvalue $\vo\in\RR\setminus\{0\}$ of $M_{\mathrm{ec}}$, then
$(\V{E},-\V{B})$ is an eigenvector for the eigenvalue $-\vo$ of $M_{\mathrm{ec}}$.
In particular, $(\V{E},-\V{B})\perp(\V{E},\V{B})$, thus $\|\V{E}\|=\|\V{B}\|$.
Putting these remarks together, we find a non-decreasing sequence $\{\omega_n\}_{n\in\NN}$ of 
strictly positive eigenvalues (counting multiplicities) and corresponding eigenvectors 
$\{(\V{E}_n,\V{B}_n)\}_{n\in\NN}$ of $M_{\mathrm{ec}}^\perp$, normalized 
such that $\|\V{E}_n\|=\|\V{B}_n\|=1$, with the property that 
$\{(\V{E}_n,\diamond\V{B}_n):n\in\NN,\,\diamond\in\{+,-\}\}$ is a complete orthogonal
system in $\ker(M_{\mathrm{ec}})^\perp$.
Since these eigenvectors are divergence free, we find
\begin{align*}
\omega_n^2\V{E}_n=i\omega_n\rot\V{B}_n=\rot\,\rot\V{E}_n=-\Delta\V{E}_n,\quad n\in\NN,
\end{align*}
and likewise $(\Delta+\omega_n^2)\V{B}_n=\V{0}$.
By elliptic regularity, $\V{E}_n,\V{B}_n\in C^\infty(\Lambda,\CC^3)$,
$n\in\NN$. The following asymptotics, which are {\em uniform in $\V{x}\in\Lambda$}, 
are proven in \cite[Satz~12]{MuellerNiemeyer1961},
\begin{align}\label{EnBnasymp}
\sum_{0<\omega_n<\tau}\frac{|\V{E}_n(\V{x})|^2}{\omega_n^2}\sim\frac{\tau}{\pi^2},
\quad
\;\;\sum_{0<\omega_n<\tau}{|\V{E}_n(\V{x})|^2}\sim\frac{\tau^2}{3\pi^2},\qquad\tau\to\infty.
\end{align}
The same asymptotic relations are valid with $\V{E}_n$ replaced by $\V{B}_n$.

In the above situation the
measure space $(\cM,\fA,\mu)$ equals $(\NN,\fP(\NN),\cZ)$, with $\cZ$ denoting
the counting measure on the power set $\fP(\NN)$. The dispersion relation is given by 
$\omega(n):=\omega_n$, $n\in\NN$.
Furthermore, we pick some measurable ultra-violet cut-off function $\chi:[0,\infty)\to[0,1]$ such that
$\chi(\tau)=\cO(\tau^{-\alpha})$, $\tau\to\infty$, for some $\alpha>2$. 
Then the recipe for quantizing the electromagnetic
radiation field described, e.g., in \cite[\textsection2.4.1]{Dutra2005} or 
\cite{PowerThirunamachandran1982} amounts to defining
\begin{align*}
\V{G}^\chi_{\V{x}}(n)&:=\ee{\chi(\omega_n)}{\omega_n^\mh}\V{E}_n(\V{x}),\quad
\V{F}^\chi_{\V{x}}(n):=\frac{\ee}{2}\chi(\omega_n)\omega_n^\eh\V{B}_n(\V{x}),
\end{align*}
for all $\V{x}\in{\Lambda}$ and $n\in\NN$.
Here $\ee\in\RR$ accounts for some combination of physical constants. Then we have the relation 
$-(i/2)\rot\V{G}^\chi_{\V{x}}(n)=\V{F}^\chi_{\V{x}}(n)$, $\V{x}\in\Lambda$, $n\in\NN$. By the choice
of $\chi$ and by virtue of \eqref{EnBnasymp} and its analogue for the $\V{B}_n$,
\begin{align*}
\|\V{G}^\chi\|_{L^\infty(\Lambda,\fdom(\omega^{-1}+\omega)^3)}^2&=\ee^2
\sup_{\V{x}\in\Lambda}\sum_{n\in\NN}\Big(\frac{1}{\omega_n^2}+1\Big)
\chi(\omega_n)^2|\V{E}_n(\V{x})|^2<\infty,
\\
\|\V{F}^\chi\|_{L^\infty(\Lambda,\fdom(\omega^{-1})^3)}^2&=\frac{\ee^2}{4}
\sup_{\V{x}\in\Lambda}\sum_{n\in\NN}(1+\omega_n)\chi(\omega_n)^2|\V{B}_n(\V{x})|^2<\infty.
\end{align*}
Furthermore, the fact that each $\V{E}_n$ is divergence free and Rem.~\ref{remapp} ensure that
$\V{G}^\chi$ has the weak $\fdom(\omega^{-1})$-valued divergence $\Div\V{G}^\chi=0$.
\end{example}

\begin{example} Next, we recall the two perhaps most common cases where the
Maxwell equations with perfect electric conductor boundary conditions are explicitly solvable.
In both items below, $\ee\in\RR$, the cut-off $\chi:[0,\infty)\to[0,1]$ is measurable and satisfies
$\chi(\tau)=\cO(\tau^{-\alpha})$, $\tau\to\infty$, for some $\alpha>2$, and 
$\omega(\V{k},\lambda):=\omega(\V{k}):=|\V{k}|$, $\V{k}\in\RR^3$, $\lambda\in\{1,2\}$.
Moreover, $\V{\ve}_\lambda:S^2\to S^2$, $\lambda\in\{1,2\}$, are measurable such that
$\{{\V{a}},\V{\ve}_1({\V{a}}),\V{\ve}_2({\V{a}})\}$ is an oriented orthonormal
basis of $\RR^3$, for every ${\V{a}}\in S^2$. Finally, $\mr{\V{k}}:=|\V{k}|^{-1}\V{k}$, 
if $\V{k}\in\RR^3\setminus\{\V{0}\}$. In both cases
$\V{G}^\chi$ fulfills the hypothesis in Thm.~\ref{thm-dom} with $\Div\V{G}^\chi=0$ 
(due to Rem.~\ref{remapp}) and $\V{F}_{\V{x}}^\chi:=-(i/2)\rot\V{G}_{\V{x}}^\chi$ then
fulfills the hypothesis in Rem.~\ref{remZeeman}.
\begin{enumerate}[leftmargin=0.67cm]
\item[(1)] Let $\V{\ell}\in(0,\infty)^3$ and consider the parallelepiped
$\Lambda(\V{\ell}):=(0,\ell_1)\times(0,\ell_2)\times(0,\ell_3)$; see, e.g., 
\cite[\textsection2.7]{Milonni1994}.
Put $\mathbb{L}:=\{(\pi n_1/\ell_1,\pi n_2/\ell_2,\pi n_3/\ell_3)\in\RR^3:n_1,n_2,n_3\in\NN_0\}$ and
let $\mathbb{L}_*$ be the set of all $\V{k}\in\mathbb{L}$ having at most one vanishing component.
Set $\cM:=\mathbb{L}_*\times\{1,2\}$ and suppose that the measure $\mu$ gives weight
$1$ to $(\V{k},\lambda)\in\cM$, if no component of $\V{k}$ vanishes, and weight $1/2$
otherwise. In this situation, analogues of the modes $\V{E}_n$ appearing in Ex.~\ref{exsmoothbd} 
are indexed by $(\V{k},\lambda)\in\cM$. They read
\begin{align*}
\V{E}_{(\V{k},\lambda)}(\V{x})&:=\frac{\sqrt{8}}{\sqrt{\ell_1\ell_2\ell_3}}
\begin{pmatrix}
\ve_{\lambda,1}(\mr{\V{k}})\cos(k_1x_1)\sin(k_2x_2)\sin(k_3x_3)\\
\ve_{\lambda,2}(\mr{\V{k}})\sin(k_1x_1)\cos(k_2x_2)\sin(k_3x_3)\\
\ve_{\lambda,3}(\mr{\V{k}})\sin(k_1x_1)\sin(k_2x_2)\cos(k_3x_3)
\end{pmatrix},
\end{align*}
for all $\V{x}\in\ol{\Lambda(\V{\ell})}$,
if no component of $\V{k}$ vanishes. If precisely one component of $\V{k}$ vanishes, then we have to
replace the components $\ve_{\lambda,\ell}$ of the polarization vectors by a complex number
of modulus $1$ in the above
formula, which yields the same mode for $\lambda=1$ and $\lambda=2$ and explains the choice
of $\mu$. Then every $\V{E}_{(\V{k},\lambda)}$ is normalized, divergence free, and satisfies 
$\V{n}\times\V{E}_{(\V{k},\lambda)}\!\!\upharpoonright_{\partial\Lambda(\V{\ell})}=\V{0}$
and $\Delta\V{E}_{(\V{k},\lambda)}=-\omega(\V{k})^2\V{E}_{(\V{k},\lambda)}$.
Hence, we define $\V{G}_{\V{x}}^\chi(\V{k},\lambda):=\ee{\chi(\omega(\V{k}))}
{\omega(\V{k})^\mh}\V{E}_{(\V{k},\lambda)}(\V{x})$, for all $(\V{k},\lambda)\in\cM$.
%\item[(2)] In the case of the half-space $\Lambda=(0,\infty)\times\RR^2$ we choose $\cM=\RR^3\times\{1,2\}$ and the measure $\mu$ is the product of the Lebesgue-Borel measure on $\RR^3$ with the counting measure on $\{1,2\}$. For all $\V{x}\in\ol{\Lambda}$, we further set\begin{align*}\V{G}_{\V{x}}^\chi(\V{k},\lambda):=\ee(2\pi)^{-\nf{3}{2}}(2\omega(\V{k}))^\mh e^{-ik_2x_2-ik_3x_3}\begin{pmatrix}i\ve_{\lambda,1}(\mr{\V{k}})\cos(k_1x_1)\\\ve_{\lambda,2}(\mr{\V{k}})\sin(k_1x_1)\\\ve_{\lambda,3}(\mr{\V{k}})\sin(k_1x_1)\end{pmatrix}.\end{align*}
\item[(2)] In the case $\Lambda=\RR^3$ we choose $\cM=\RR^3\times\{1,2\}$ and $\mu$ is the
product of the Lebesgue-Borel measure on $\RR^3$ with the counting measure on $\{1,2\}$.
It is common to choose
$\V{G}_{\V{x}}^\chi(\V{k},\lambda):=\ee(2\pi)^{-\nf{3}{2}}(2\omega(\V{k}))^\mh
\chi(\omega(\V{k}))e^{-i\V{k}\cdot\V{x}}\V{\ve}_\lambda(\mr{\V{k}})$.
\end{enumerate}
\end{example}

%%%%%%%%%%%%%%%%%%%%%%%%%%%%%%%%%%%%%%%%%%

\subsection{More on the entire Euclidean space}\label{appLS}

\begin{example}
Consider the case $\Lambda=\RR^\nu$ where $\V{A}\in L^4_\loc(\RR^\nu,\RR^\nu)$ has a
weak divergence $\Div\V{A}\in L^2_\loc(\RR^\nu)$. Let 
$\V{G}\in L^\infty(\RR^\nu,\fdom(\omega^{-1}+\omega)^\nu)$ have a weak divergence
$\Div\V{G}\in L^\infty(\RR^\nu,\fdom(\omega^{-1}))$. Finally, let $V_\pm\in L^2_\loc(\RR^\nu)$,
$V_\pm\ge0$, such that $V_-$ is relatively $-\frac{1}{2}\Delta$-bounded (in the operator sense)
with relative bound $<1$. Then $\mathfrak{s}_{\RR^\nu,\mathrm{D}}^{\V{A},V}
=\mathfrak{s}_{\RR^\nu,\mathrm{N}}^{\V{A},V}$ and 
$\mathfrak{q}_{\RR^\nu,\mathrm{D}}^{\V{G},\V{A},V}
=\mathfrak{q}_{\RR^\nu,\mathrm{N}}^{\V{G},\V{A},V}$ by \cite{SimonJOT1979} and
Cor.~\ref{corminmax}, respectively. In particular, the Dirichlet and Neumann realizations of
the Schr\"{o}dinger operator agree, and the same holds for the Pauli-Fierz operator.
Hence, we may drop all subscripts $\mathrm{D}$ and $\mathrm{N}$ in what follows.

Under the above conditions the Leinfelder-Simader theorem \cite[Thm.~3]{LeinfelderSimader1981}
says that $S^{\V{A},V}_{\RR^\nu}$ is essentially self-adjoint on $\sD(\RR^\nu)$.
Hence, by Thm.~\ref{thm-dom}, $H^{\V{G},\V{A},V}_{\RR^\nu}$ is essentially self-adjoint
on $\sD(\RR^\nu)\otimes\sE$, for every core $\sE$ of the field energy
$\Id\Gamma(\omega)$.
\end{example}

%%%%%%%%%%%%%%%%%%%%%%%%%%%%%%%%%%%%%%%%%%

\subsection{${N}$-particle Hamiltonians and Pauli principle}\label{ssecPauli}

\noindent
The next example clarifies in particular how the examples of Subsect.~\ref{appcoupling} are extended 
to several electrons.

\begin{example}
Let $\Lambda\subset\RR^3$ be open, $N\in\NN$, $N>1$, and 
$\Lambda_*^N:=\Lambda^N\times\{-1,1\}^N$, i.e., the Cartesian products
$\Lambda^N$ and $\{-1,1\}^N$ will play the roles of the position and spin spaces, respectively. 
The corresponding position and spin variables will be denoted as 
$\ul{\V{x}}=(\V{x}_1,\ldots,\V{x}_N)$ and $\ul{\vs}:=(\vs_1,\ldots,\vs_N)$.
%
% To do the correction (42) copy and paste from here ....
%
If $\hat{\sigma}_1$, $\hat{\sigma}_2$, and $\hat{\sigma}_3$ are the standard Pauli matrices, whose
entries we label by $(\vs,\vs')\in\{-1,1\}^2$, then the $3N$ components of $\V{\sigma}$
are the $2^N$\texttimes$2^N$-matrices given by
\begin{align*}
(\sigma_{3(\ell-1)+j})_{\ul{\vs},\ul{\vs}'}:=(\hat{\sigma}_j)_{\vs_\ell,\vs_\ell'}\cdot\left\{
\begin{array}{ll}
1,&\;\;\text{if $\vs_i=\vs_i'$, for $i\in\{1,\ldots,N\}\setminus\{\ell\}$,}
\\
0,&\;\;\text{otherwise,}
\end{array}
\right.
\end{align*}
for $\ul{\vs},\ul{\vs}'\in\{-1,1\}^N$, $\ell\in\{1,\ldots,N\}$, and $j\in\{1,2,3\}$.
%
%   ..... to here.
%
Assuming that $\V{G}\in L^\infty(\Lambda,\fdom(\omega^{-1}+\omega)^3)$ has a weak divergence
$\Div\V{G}\in L^\infty(\Lambda,\fdom(\omega^{-1}))$ and 
$\V{F}\in L^\infty(\Lambda,\fdom(\omega^{-1})^3)$,
we introduce $\V{G}^N\in L^\infty(\Lambda,\fdom(\omega^{-1}+\omega)^{3N})$ 
with a weak divergence $\Div\V{G}^N\in L^\infty(\Lambda,\fdom(\omega^{-1}))$ and
$\V{F}^N\in L^\infty(\Lambda,\fdom(\omega^{-1})^{3N})$ by setting
$$
\V{G}^N_{\ul{\V{x}}}:=(\V{G}_{\V{x}_1},\ldots,\V{G}_{\V{x}_N}),\quad \ul{\V{x}}\in\Lambda^N,
$$
and defining $\V{F}^N$ analogously. We suppose that $\V{A}\in L^2_\loc(\Lambda,\RR^3)$ and put
$$
\V{A}^N(\ul{\V{x}}):=(\V{A}(\V{x}_1),\ldots,\V{A}(\V{x}_N)),\quad \ul{\V{x}}\in\Lambda^N.
$$
Finally, we pick $V_+\in L^1_\loc(\Lambda^N)$, $V_+\ge0$, and 
$V_-\in L_\loc^1(\Lambda^N,\LO(\CC^{2^N}))$
such that $\|V_-\|_{\LO(\CC^{2^N})}$ is relatively form bounded with respect to $-1/2$ times the
Dirichlet-Laplacian on $\Lambda^N$ with relative form bound $<1$. 
Let $\fS_N$ be the set of permutations of $\{1,\ldots,N\}$, put
$\pi^*\ul{\V{x}}:=(\V{x}_{\pi(1)},\ldots,\V{x}_{\pi(N)})$, for all $\pi\in\fS_N$, and define
$\pi^*\ul{\vs}$ analogously. Then we further assume that $V_+(\pi^*\ul{\V{x}})=V_+(\ul{\V{x}})$ and
$$
V_-(\pi^*\ul{\V{x}})_{\pi^*\ul{\vs},\pi^*\ul{\vs}'}
=V_-(\ul{\V{x}})_{\ul{\vs},\ul{\vs}'},\quad\ul{\V{x}}\in\Lambda^N,\;\ul{\vs},\ul{\vs}'\in\{-1,1\}^N,\;\pi\in\fS_N.
$$
Finally, let $\cA_N$ be the orthogonal projection onto the space of functions 
obeying the Pauli principle, i.e.,
$$
(\cA_N\Psi)(\ul{\V{x}},\ul{\vs}):=\frac{1}{N!}\sum_{\pi\in\fS_N}\sgn(\pi)
\Psi(\pi^*\ul{\V{x}},\pi^*\ul{\vs}),
\quad\text{a.e. $(\ul{\V{x}},\ul{\vs})\in\Lambda_*^N$,}
$$
for every $\Psi\in L^2(\Lambda_*^N,\sF)$.
Then it is straightforward to show that $\Ran(\cA_N)$ is a reducing subspace for 
$$
H_N:=H_{\Lambda_*^N,\mathrm{D}}^{\V{G}^N,\V{A}^N,V}-\vsigma\cdot\vp(\V{F}^N).
$$
It now follows from Rem.~\ref{remZeeman} that the restriction
$H_N\!\!\upharpoonright_{\Ran(\cA_N)}$ is self-adjoint on the domain
$\cA_N\dom(H_{\Lambda_*^N,\mathrm{D}}^{\V{0},\V{A}^N,V})$. Furthermore, if
$\sC$ is a core for the Dirichlet-Schr\"{o}dinger operator 
$S_{\Lambda_*^N,\mathrm{D}}^{\V{A}^N,V}$ and $\sE$ is a core for
$\Id\Gamma(\omega)$, then $\cA_N(\sC\otimes\sE)$ is a core for 
$H_N\!\!\upharpoonright_{\Ran(\cA_N)}$.
\end{example}

The potential $V$ in the previous example could for instance be a sum of classical Zeeman terms
and a multi-particle Coulomb potential
for a molecule in a half space bounded by a perfectly conducting wall. In this case the multi-particle 
Coulomb potential contains the electrostatic interactions between all charged particles (electrons
and nuclei) {\em and} their image charges behind the wall; see, e.g., 
\cite{PowerThirunamachandran1982} for explicit formulas.

%%%%%%%%%%%%%%%%%%%%%%%%%%%%%%%%%%%%%%%%%%
%%%%%%%%%%%%%%%%%%%%%%%%%%%%%%%%%%%%%%%%%%

\section{The Neumann case}\label{appNeumann}

\noindent
For mathematical curiosity we derive a version of Thm.~\ref{thm-dom} for Neumann boundary
conditions in this final section. While the behavior of $\V{G}$ at the boundary $\partial\Lambda$
and the regularity of  $\partial\Lambda$ did not play any role for the validity of Thm.~\ref{thm-dom}, 
in the Neumann case $\V{G}$ and $\partial\Lambda$ have to satisfy suitable boundary and regularity 
conditions, respectively. It turns out that perfect {\em magnetic} conductor 
boundary conditions permit to derive an analogue of the integration by parts formula 
\eqref{bakithi100div}. These are boundary conditions imposed on the Maxwell operator requiring
the tangential components of the magnetic field and the normal component of the electric field 
(and hence of $\V{G}$) to vanish on $\partial\Lambda$. In other words, the roles of the electric and
magnetic fields are switched in comparison to perfect electric conductor boundary conditions.
%Notice that Neumann and perfect magnetic conductor boundary conditions are both associated with reflections at the boundary with a phase shift of zero, while Dirichlet and perfect electric conductor boundary conditions correspond to a phase shift of $\pi$.

Throughout the whole section we shall assume that $\nu\ge2$
and the boundary $\partial\Lambda$ is Lipschitz with exterior normal field $\V{n}$; 
see, e.g., \cite[\textsection4.2.1]{EvansGariepy} for a definition of Lipschitz boundaries and their 
normal fields. The symbol $\cH^{\nu-1}$ will denote the $(\nu-1)$-dimensional Hausdorff measure.
The classical vector potential is assumed to have a locally square-integrable extension to the whole
Euclidean space, $\V{A}\in L^2_\loc(\RR^\nu,\RR^\nu)$. The coupling function
$\V{G}\in L^\infty(\Lambda,\fdom(\omega^{-1}+\omega)^\nu)$ has weak partial derivatives with respect
to all variables such that $\partial_{x_j}\V{G}\in L^\infty(\Lambda,\fdom(\omega^{-1})^\nu)$, for
all $j\in\{1,\ldots,\nu\}$. As in the scalar case we then see that (a unique representative of) 
$\V{G}:\Lambda\to\fdom(\omega^{-1})^\nu$ is locally Lipschitz continuous. Since we can
always choose the $L^\infty(\Lambda,\fdom(\omega^{-1})^{\nu\times\nu})$-norm of $\nabla\V{G}$ 
as a local Lipschitz constant, $\V{G}$ has a unique continuous $\fdom(\omega^{-1})^\nu$-valued 
extension to $\ol{\Lambda}$. We postulate that
\begin{align}\label{tracehypG}
\V{n}\cdot\V{G}=0,\quad
\text{$\cH^{\nu-1}$-a.e. on $\partial\Lambda$.}
\end{align}

%Then, for all $h\in\HP$ and $j\in\{1,\ldots,\nu\}$, the function $\SPn{h}{{G}_j}_{\HP}\in W^{1,\infty}(\Lambda)$ has a well-defined trace on $\partial\Lambda$. Denoting traces by $\upharpoonright_{\partial\Lambda}$, we further assume that\begin{align}\label{tracehypG}\V{n}\cdot\SPn{h}{\V{G}}_{\HP}\!\!\upharpoonright_{\partial\Lambda}=0,\quad h\in\HP.\end{align}
%This condition is equivalent to\begin{align*}\V{n}\cdot\V{G}\!\!\upharpoonright_{\partial\Lambda}=0\quad\text{with}\quad\V{G}\!\!\upharpoonright_{\partial\Lambda}\end{align*}

\begin{lemma}\label{lemNeumannBC}
Under the assumptions on $\V{G}$ described in the preceding paragraph,
let $f\in C_0^\infty(\ol{\Lambda})$, $\phi\in\fdom(\Id\Gamma(\omega))$, and let
$\Psi\in L^1_\loc(\Lambda,\sF)$ have weak partial derivatives with respect to all variables.
Then $\SPn{f\vp(\V{G})\phi}{\Psi}_{\sF}\in W^{1,1}(\Lambda,\CC^\nu)$ and
$\V{n}\cdot\SPn{f\vp(\V{G})\phi}{\Psi}_{\sF}\!\!\upharpoonright_{\partial\Lambda}=0$, where
$\upharpoonright_{\partial\Lambda}$ denotes the trace of Sobolev functions.
\end{lemma}

\begin{proof}
Pick some $\psi\in\sF$. Applying Lem.~\ref{lem-Leibniz-vp} to the vectors
$\V{G}^{(j,\ell)}$ with components $G^{(j,\ell)}_{k}:=\delta_{j,k}G_\ell$, for all
$j,\ell,k\in\{1,\ldots,\nu\}$, we convince ourselves
that $\SPn{f\vp(\V{G})\phi}{\psi}_{\sF}\in W^{1,\infty}(\Lambda,\CC^\nu)$ and 
$\SPn{f\vp(\V{G})\phi}{\Psi}_{\sF}\in W^{1,1}(\Lambda,\CC^\nu)$ and that their
weak partial derivatives can be computed by formally applying Leibniz rules.
In view of \eqref{contvp} we further know that 
$\SPn{f\vp(\V{G})\phi}{\psi}_{\sF}\in C(\ol{\Lambda},\CC^\nu)$.

%Let also $g,h\in\dom(\omega)$. Then \eqref{mevpexpv} and \eqref{tracehypG} imply that $\SPn{f\vp({G}_\ell)\expv{g}}{\expv{h}}_{\sF}\in W^{1,\infty}(\Lambda)$ satisfies $\V{n}\cdot\SPn{\vp(\V{G})\expv{g}}{\expv{h}}_{\sF}\!\!\upharpoonright_{\partial\Lambda}=0$. Picking $\psi_n\in\cC[\dom(\omega)]$, $n\in\NN$, with $\psi_n\to\psi$, $n\to\infty$, in $\sF$, and applying the remarks in the first paragraph, we then see that $\SPn{f\vp(\V{G})\phi}{\psi_n}_{\sF}\to\SPn{f\vp(\V{G})\phi}{\psi}_{\sF}$ in $W^{1,1}(\Lambda,\CC^\nu)$. This entails convergence of the traces in $L^1(\partial\Lambda,\cH^{\nu-1})$, whence $\SPn{f\vp(\V{G})\phi}{\psi}_{\sF}\!\!\upharpoonright_{\partial\Lambda}=0$.

Pick some $\vr\in C_0^\infty(\RR^\nu,\RR)$ with $\vr=1$ on $\supp(f)$ and
let $\{e_n:n\in\NN\}$ be an orthonormal basis of $\sF$. Let $\Lambda'$ be the intersection of
$\Lambda$ with a sufficiently large open ball containing the supports of $f$ and $\vr$.
Then $\SPn{e_n}{\vr\Psi}_{\sF}\in W^{1,1}(\Lambda')$, $n\in\NN$, have a well-defined trace on 
$\partial\Lambda'$ and we infer from the above remarks that the functions
$\V{X}_m:=\sum_{n=1}^m\SPn{f\vp(\V{G})\phi}{e_n}_{\sF}\SPn{e_n}{\vr\Psi}_{\sF}
\in W^{1,1}(\Lambda',\CC^\nu)$
satisfy $\V{n}\cdot\V{X}_m\!\!\upharpoonright_{\partial\Lambda'}=0$, for all $m\in\NN$.
The trace $\upharpoonright_{\partial\Lambda'}:W^{1,1}(\Lambda')\to L^1(\partial\Lambda',\cH^{\nu-1})$ 
on the bounded domain $\Lambda'$ is continuous. Hence, we may conclude by observing that 
$\V{X}_m\to\SPn{f\vp(\V{G})\phi}{\Psi}_{\sF}$, $m\to\infty$, in $W^{1,1}(\Lambda',\CC^\nu)$.
\end{proof}

We finally assume that $V_\pm\in L^1_\loc(\Lambda)$, $V_\pm\ge0$, are such 
that \eqref{KLMN} is satisfied with $\diamond=\mathrm{N}$ and for some $a\in[0,1)$ and $b\ge0$.

\begin{theorem}\label{thmdomN}
In the situation described above, $\dom(H_{\Lambda,\mathrm{N}}^{\V{G},\V{A},V})
=\dom(H_{\Lambda,\mathrm{N}}^{\V{0},\V{A},V})$ and
\begin{align}\label{bakithi1Neumann}
H_{\Lambda,\mathrm{N}}^{\V{G},\V{A},V}\Psi&=
H_{\Lambda,\mathrm{N}}^{\V{0},\V{A},V}\Psi-\vp(\V{G})\cdot{\V{w}^*}\Psi
+\frac{1}{2}\vp(\V{G})^2\Psi+\frac{i}{2}\vp(q)\Psi,
\end{align}
for all $\Psi\in\dom(H_{\Lambda,\mathrm{N}}^{\V{G},\V{A},V})$. 
The graph norms of $H_{\Lambda,\mathrm{N}}^{\V{G},\V{A},V}$ and 
$H_{\Lambda,\mathrm{N}}^{\V{0},\V{A},V}$ are equivalent and, in particular,
both operators have the same cores.
\end{theorem}

Recall that a natural class of operator cores for $H_{\Lambda,\mathrm{N}}^{\V{0},\V{A},V}$ has
been identified in Lem.~\ref{lem-free-ham}(2). Zeeman terms accounting for spin degrees of freedom
can be added in the previous theorem by the same arguments as in Rem.~\ref{remZeeman}.

\begin{proof}
Essentially, we only have to extend \eqref{bakithi100div} to test functions $\Phi$ that might be
non-vanishing on the boundary. This is done in the first step below, which is the only one where the 
boundary condition on $\V{G}$ is used explicitly.

{\em Step~1.} Fix $f\in C_0^\infty(\ol{\Lambda})$, $\phi\in\sF$, and 
$\Psi\in\dom(H_{\Lambda,\mathrm{N}}^{\V{0},\V{A},V})$. 
On account of Lem.~\ref{lem-free-ham}(3), $\Psi$ satisfies the
assumptions in Lem.~\ref{lemmagnLeibniz}, and
$\SPn{f\phi}{\vp(\V{G})\Psi}_{\sF}\in W^{1,1}(\Lambda,\CC^\nu)$ by Lem.~\ref{lemNeumannBC}. 
Since the identity \eqref{bakithi99b} in the proof of Lem.~\ref{lemmagnLeibniz} is available,
we further conclude that
\begin{align*}
\Div\SPn{f\phi}{\vp(\V{G})\Psi}_{\sF}&=\SPn{f\phi}{i\vp(\V{G})\cdot\V{w}^*\Psi+\vp(q)\Psi}_{\sF}
\\
&\quad+\sum_{j=1}^\nu\SPn{(\partial_{x_j}-iA_j)f\phi}{\vp({G}_j)\Psi}_{\sF},
\end{align*}
where $q:=\Div\V{G}$. It is easy to see that $f\phi\in\dom(w_j^*)$
with ${w}_j^*(f\phi)=(-i\partial_{x_j}f-A_jf)\phi$. 
%As observed in \cite[Prop.~A.1]{Chrusciel2013} the divergence theorem applies to the bounded Lipschitz domain $\Lambda$ and any Sobolev vector field in $W^{1,1}$
%Let us for a moment consider the divergence theorem \begin{align}\label{Gauss} \int_\Lambda\Div\,\V{\Xi}\Id\V{x}&=\int_{\partial\Lambda}\V{\Xi}\cdot\V{n}\Id\cH^{\nu-1},\end{align} which is known to hold at least for all Lipschitz continuous $\V{\Xi}:\RR^\nu\to\RR^\nu$ with compact support \cite[p.~478]{Federer}. If we only assume that $\V{\Xi}\in W^{1,1}(\Lambda)$, then we find, however, $\V{\Xi}_n\in C^\infty(\ol{\Lambda})$, $n\in\NN$, such that $\V{\Xi}_n\to\V{\Xi}$ in $W^{1,1}(\Lambda)$; see, e.g., \cite[\textsection4.2.1]{EvansGariepy}. Since the trace $\upharpoonright_{\partial\Lambda}$ is a bounded linear map from $W^{1,1}(\Lambda)$ to $L^1(\partial\Lambda,\cH^{\nu-1})$ (see, e.g., \cite[\textsection4.3]{EvansGariepy}), we conclude that \eqref{Gauss} is in fact availbale for every $\V{\Xi}\in W^{1,1}(\Lambda)$.
Combining Thm.~3 on page~127 and Thm.~1 on page 133 of \cite{EvansGariepy} we observe that
the divergence theorem applies to $\SPn{f\phi}{\vp(\V{G})\Psi}_{\sF}$ and the intersection of
$\Lambda$ with some large open ball, which together with Lem.~\ref{lemNeumannBC} and the
above remarks implies
\begin{align}\label{livinia1}
\SPn{f\phi}{&\vp(\V{G})\cdot\V{w}^*\Psi-i\vp(q)\Psi}
-\sum_{j=1}^\nu\SPn{w_j^*(f\phi)}{\vp({G}_j)\Psi}=0.
\end{align}

{\em Step~2.} Next, we observe that $C_0^\infty(\ol{\Lambda})$ is a core for the form
$\mathfrak{s}_{\Lambda,\mathrm{N}}^{\V{A},V_+}$. In fact, the condition 
$\V{A}\in L_\loc^2(\RR^\nu,\RR^\nu)$ ensures that
$C_0^\infty(\ol{\Lambda})\subset\dom(\mathfrak{s}_{\Lambda,\mathrm{N}}^{\V{A},V_+})$.
Furthermore, that
$\mathfrak{s}_{\Lambda,\mathrm{N}}^{\V{A},V_+}\cap L^\infty(\Lambda)$ is a core for
$\mathfrak{s}_{\Lambda,\mathrm{N}}^{\V{A},V_+}$ follows from the argument in
\cite[Step~1 on p.~125]{HundertmarkSimon}. If 
$f\in\dom(\mathfrak{s}_{\Lambda,\mathrm{N}}^{\V{A},V_+})\cap L^\infty(\Lambda)$, then we can pick
$\vt_n\in C_0^\infty(\RR^\nu)$ with $\vt_{n+1}=1$ on $\supp(\vt_n)$, $n\in\NN$,
$\bigcup_{n\in\NN}\supp(\vt_n)=\RR^\nu$, and $\sup_n\|\nabla\vt_n\|_\infty<\infty$.
Observing $\sum_{n=1}^\infty(\nabla\vt_n)f\in L^2(\Lambda,\RR^\nu)$ we can then follow
the reasoning in \cite[Step~2 on p.~125]{HundertmarkSimon} to see that
$\vt_nf\to f$ with respect to the form norm of $\mathfrak{s}_{\Lambda,\mathrm{N}}^{\V{A},V_+}$.
Finally, if $g\in\dom(\mathfrak{s}_{\Lambda,\mathrm{N}}^{\V{A},V_+})\cap L^\infty(\Lambda)$ has a
compact support, then $A_jg\in L^2(\Lambda)$, which implies $\partial_{x_j}g\in L^2(\Lambda)$,
as we a priori know that $\partial_{x_j}g=iA_jg+{\sf w}_j^*g$ in $L^1_\loc(\Lambda)$.
Invoking \cite[Thm.~3 on p.~127]{EvansGariepy}, 
we find $g_n\in C^\infty(\ol{\Lambda})$, $n\in\NN$, such that
$g_n\to g$, $n\to\infty$, in $W^{1,2}(\Lambda)$ and pointwise a.e. on $\Lambda$.
A glance at the proof of \cite[Thm.~3 on p.~127]{EvansGariepy} reveals that we may further
assume that $\|g_n\|_\infty\le\|g\|_\infty$, $n\in\NN$, and that all $g_n$ and $g$ have their
supports contained in some fixed compact set. Employing the dominated
convergence theorem we then conclude that
${\sf w}_j^*g_n=-i\partial_{x_j}g_n-A_jg_n\to -i\partial_{x_j}g-A_jg={\sf w}_j^*g$ and
$V_+^\eh g_n\to V_+^\eh g$ in $L^2(\Lambda)$.

Combining this result with Cor.~\ref{corformcore} we see that
$C_0^\infty(\ol{\Lambda})\otimes\fdom(\Id\Gamma(1\vee\omega))$ is a core for the form
$\mathfrak{q}_{\Lambda,\mathrm{N}}^{\V{G},\V{A},V}$.

{\em Step~3.} If we employ \eqref{zerlqN} and \eqref{livinia1} instead of
\eqref{egon99} and Lem.~\ref{lemmagnLeibniz}, respectively, and choose
$C_0^\infty(\ol{\Lambda})\otimes\fdom(\Id\Gamma(1\vee\omega))$ instead of
$\sD(\Lambda,\fdom(\Id\Gamma(1\vee\omega)))$ as a core, then, apart from the very last sentence,
all arguments in the proof of Prop.~\ref{thm-dom}(1) remain valid after the subscript
$\mathrm{D}$ has been replaced by $\mathrm{N}$ everywhere. This proves
\eqref{bakithi1Neumann} for all $\Psi\in\dom(H_{\Lambda,\mathrm{N}}^{\V{0},\V{A},V})$.

{\em Step~4.} Employing \eqref{bakithi1Neumann} for 
$\Psi\in\dom(H_{\Lambda,\mathrm{N}}^{\V{0},\V{A},V})$ instead of Prop.~\ref{proprep}(1) and
replacing the subscript $\mathrm{D}$ by $\mathrm{N}$ everywhere, we can now 
literally follow the proofs of Lem.~\ref{lemrbK} and Thm.~\ref{thm-dom} to arrive at the full assertion.
(We employ \eqref{zerlqN} instead of Lem.~\ref{lem-egon} in the analogue of \eqref{echo1};
notice that Lem.~\ref{lem-free-ham} and Lem.~\ref{lemecho} cover the Neumann case.)
\end{proof}

%%%%%%%%%%%%%%%%%%%%%%%%%%%%%%%%%%%%%%%%%%
%%%%%%%%%%%%%%%%%%%%%%%%%%%%%%%%%%%%%%%%%%
%%%%%%%%%%%%%%%%%%%%%%%%%%%%%%%%%%%%%%%%%%

\bigskip

\noindent{\bf Acknowledgement.}
The author is grateful for support by the VILLUM foundation via the project grant 
``Spectral Analysis of Large Particle Systems'', and for support by the International Network 
Programme grant ``Exciting Polarons'' from the Danish Ministry for Research and Education.

%%%%%%%%%%%%%%%%%%%%%%%%%%%%%%%%%%%%%%%%%%
%%%%%%%%%%%%%%%%%%%%%%%%%%%%%%%%%%%%%%%%%%
%%%%%%%%%%%%%%%%%%%%%%%%%%%%%%%%%%%%%%%%%%

\end{document}